\documentclass[AMA,STIX1COL]{WileyNJD-v2}
\articletype{Research Article}%

\received{xxxx}
\revised{xxxx}
\accepted{xxxx}
\usepackage{comment}
\usepackage{rotating}
\usepackage{cases}
\usepackage{threeparttable}
\usepackage{siunitx}
\usepackage{amsmath}
\usepackage{enumerate}
\usepackage{mathtools}
\usepackage{bbm}
\usepackage{tikz}
\usepackage{color}
\usepackage{twemojis}
\usepackage{multirow}
\usepackage{booktabs}
\usetikzlibrary{matrix}
\usepackage{enumitem}
\usepackage{graphicx}
\usepackage{subcaption}
\usepackage{float}

\newcommand{\sd}{\mathrm{sd}}
\newcommand{\se}{\mathrm{se}}
\newcommand{\HH}{\text{HH}}
\newcommand{\ETI}{\text{ETI}}

\newcommand{\ETIANT}{\text{ETI-ANT}}
\newcommand{\HHANT}{\text{HH-ANT}}

\newcommand{\bbE}{\mathbb{E}}

\newcommand{\bbV}{\mathbb{V}}
\newcommand{\bbone}{\mathbbm{1}}

\newcommand{\calN}{\mathcal{N}}

\newcommand{\bfA}{\mathbf{A}}

\newcommand{\bfI}{\mathbf{I}}

\newcommand{\bfU}{\mathbf{U}}

\newcommand{\bfW}{\mathbf{W}}
\newcommand{\bfX}{\mathbf{X}}

\newcommand{\bfZ}{\mathbf{Z}}

\newcommand{\sumi}{\sum_{i=1}^I}
\newcommand{\sumj}{\sum_{j=1}^J}
\newcommand{\sumq}{\sum_{q=1}^Q}

\newcommand{\sumjp}{\sum_{j=1}^{J-1}}
\newcommand{\sums}{\sum_{s=1}^{J-1}}

\newcommand{\newcheckmark}{\twemoji{check mark}}
\newcommand{\newcrossmark}{\twemoji{multiply}}

\begin{document}

\title{On anticipation effect in stepped wedge cluster randomized trials}

\author[1,8]{Hao Wang}
\author[2]{Xinyuan Chen}
\author[3,5]{Katherine R. Courtright}
\author[4,5]{Scott D. Halpern}
\author[4,5]{Michael O. Harhay}
\author[6,7]{Monica Taljaard}
\author[1,8]{Fan Li$^{*,}$}

\address[1]{\orgdiv{Department of Biostatistics}, \orgname{Yale University}, \orgaddress{\state{New Haven, CT}, \country{USA}}}
\address[2]{\orgdiv{Department of Mathematics and Statistics}, \orgname{Mississippi State University}, \orgaddress{\state{Mississippi State, MS}, \country{USA}}}
\address[3]{\orgdiv{Department of Medicine}, \orgname{Hospital of the University of Pennsylvania}, \orgaddress{\state{Philadelphia, PA}, \country{USA}}}
\address[4]{\orgdiv{Center for Clinical Trials Innovation, Department of Biostatistics, Epidemiology, and Informatics}, \orgname{Perelman School of Medicine at the University of Pennsylvania}, \orgaddress{\state{Philadelphia, PA}, \country{USA}}}
\address[5]{\orgdiv{Palliative and Advanced Illness Research (PAIR) Center}, \orgname{Perelman School of Medicine at University of Pennsylvania}, \orgaddress{\state{Philadelphia, PA}, \country{USA}}}
\address[6]{\orgdiv{Clinical Epidemiology Program}, \orgname{Ottawa Hospital Research Institute}, \orgaddress{\state{Ottawa, ON}, \country{Canada}}}
\address[7]{\orgdiv{School of Epidemiology and Public Health}, \orgname{University of Ottawa}, \orgaddress{\state{Ottawa, ON}, \country{Canada}}}
\address[8]{\orgdiv{Center for Methods in Implementation and Prevention Science}, \orgname{Yale University}, \orgaddress{\state{New Haven, CT}, \country{USA}}}
\corres{$*$Fan Li, Department of Biostatistics, Yale University, New Haven, Connecticut, USA,\\
\email{fan.f.li@yale.edu}}

\abstract[Summary]{In stepped wedge cluster randomized trials (SW-CRTs), the intervention is rolled out to clusters over multiple periods. A standard approach for analyzing SW-CRTs utilizes the linear mixed model, where the treatment effect is only present after the treatment adoption, under the assumption of no anticipation. This assumption, however, may not always hold in practice because stakeholders, providers, or individuals who are aware of the treatment adoption timing (especially when blinding is challenging or infeasible) can inadvertently change their behaviors in anticipation of the forthcoming intervention. We provide an analytical framework to address the anticipation effect in SW-CRTs and study its impact. We derive expectations of the estimators based on a collection of linear mixed models and demonstrate that when the anticipation effect is ignored, these estimators give biased estimates of the treatment effect. We also provide updated sample size formulas that explicitly account for anticipation effects, exposure-time heterogeneity, or both in SW-CRTs and illustrate their impact on study power. Through simulation studies and empirical analyses, we compare the treatment effect estimators with and without adjusting for anticipation, and provide some practical considerations.}


\jnlcitation{\cname{%
\author{Wang H, Chen X, Courtright KR, Halpern SD, Harhay MO, Taljaard M, Li F}} (\cyear{2025}),
\ctitle{On anticipation effect in stepped wedge cluster randomized trials}, \cjournal{Statistics in Medicine}, \cvol{00;00:0000--0000}.}

\maketitle

\section{Introduction}

In stepped wedge cluster randomized trials (SW-CRTs), the intervention is rolled out to clusters sequentially over multiple periods.\cite{Hussey2007} Each cluster is randomly assigned a treatment adoption time, and all individuals within that cluster receive the same treatment during each period. The SW-CRTs can be cross-sectional, closed-cohort, or open-cohort, depending on whether different individuals are recruited in each period for any cluster.\cite{Copas2015} Because all clusters eventually receive the treatment, SW-CRTs may be preferred for study recruitment when there is a need to ensure a complete roll-out of the intervention prior to the end of the study or when the intervention is perceived to be beneficial by stakeholders; see, for example, four broad justifications of this design given in Hemming and Taljaard.\cite{Hemming2020} Figure \ref{figure_sw} illustrates an example SW-CRT diagram with 8 clusters. In practice, SW-CRTs may only involve a limited number of clusters; the median number of clusters randomized was found to be $11$ in a recent systematic review by Nevins et al.\cite{Nevins2024}

\begin{figure}[t]
    \centering
    \begin{tikzpicture}
        \definecolor{controlcolor}{RGB}{255, 255, 255}
        \definecolor{treatmentcolor}{RGB}{202, 201, 237}
        \definecolor{anticipationcolor}{RGB}{242, 126, 126}
    
        \matrix (sw) [matrix of nodes, nodes in empty cells,
            nodes={draw, minimum width=2cm, minimum height=0.5cm, anchor=center},
            column sep=3mm, row sep=3mm,
            column 1/.style={nodes={draw=none}},
            row 1/.style={nodes={draw=none}},
            ]{
                & \textbf{Period 1} & \textbf{Period 2} & \textbf{Period 3} & \textbf{Period 4} & \textbf{Period 5} \\
                \textbf{Cluster 1} & |[fill=anticipationcolor]| & |[fill=treatmentcolor]| & |[fill=treatmentcolor]| & |[fill=treatmentcolor]| & |[fill=treatmentcolor]| \\
                \textbf{Cluster 2} & |[fill=anticipationcolor]| & |[fill=treatmentcolor]| & |[fill=treatmentcolor]| & |[fill=treatmentcolor]| & |[fill=treatmentcolor]| \\
                \textbf{Cluster 3} & |[fill=controlcolor]| & |[fill=anticipationcolor]| & |[fill=treatmentcolor]| & |[fill=treatmentcolor]| & |[fill=treatmentcolor]| \\
                \textbf{Cluster 4} & |[fill=controlcolor]| & |[fill=anticipationcolor]| & |[fill=treatmentcolor]| & |[fill=treatmentcolor]| & |[fill=treatmentcolor]| \\
                \textbf{Cluster 5} & |[fill=controlcolor]| & |[fill=controlcolor]| & |[fill=anticipationcolor]| & |[fill=treatmentcolor]| & |[fill=treatmentcolor]| \\
                \textbf{Cluster 6} & |[fill=controlcolor]| & |[fill=controlcolor]| & |[fill=anticipationcolor]| & |[fill=treatmentcolor]| & |[fill=treatmentcolor]| \\
                \textbf{Cluster 7} & |[fill=controlcolor]| & |[fill=controlcolor]| & |[fill=controlcolor]| & |[fill=anticipationcolor]| & |[fill=treatmentcolor]| \\
                \textbf{Cluster 8} & |[fill=controlcolor]| & |[fill=controlcolor]| & |[fill=controlcolor]| & |[fill=anticipationcolor]| & |[fill=treatmentcolor]| \\
            };
    
        \foreach \i in {2,3,4,5,6,7,8,9} {
            \foreach \j in {2,3,4,5,6} {
                \draw (sw-\i-\j.north west) rectangle (sw-\i-\j.south east);
            }
        }
    \end{tikzpicture}
    \caption{An illustration of a standard SW-CRT with 8 clusters and 5 periods of equal length. The design includes 4 treatment sequences, each containing 2 clusters. Clusters in treatment sequence $q\in\{1,2,3,4\}$ start treatment in period $j=q+1$, but may experience an anticipation effect in period $j=q$. The purple cell indicates the treatment condition, the white cell indicates the control condition, and the red cell indicates a possible cluster-period in which anticipation might have occurred.}
    \label{figure_sw}
\end{figure}

A standard approach for analyzing SW-CRTs utilizes the linear mixed model, including a period effect (adjusting for the secular trend), treatment effects (including parameters of interest), and sources of heterogeneity (accounting for within-cluster correlations).\cite{Li2020} Within this framework, the most popular choice in practice is the random intercept model by Hussey and Hughes,\cite{Hussey2007} with a categorical period effect and a constant treatment effect. This model implicitly assumes that the full effect of the treatment is reached immediately and remains constant throughout. However, depending on the type of intervention, context, and population, this assumption may not always hold. There has been a flourishing line of literature that addresses the implications for treatment effect structure misspecification under stepped wedge designs. For instance, Hughes et al.\cite{Hughes2015} first suggested several fixed-effects parametrizations for the treatment effect curves, including the delayed, linear, and nonlinear structures, to estimate time-on-treatment effects; similar treatment effect structures were also reviewed in Li et al.\cite{Li2020} Nickless et al.\cite{Nickless2018} conducted a simulation study in a setting where outcomes could be affected cumulatively by the exposure time---the time elapsed since the treatment was first introduced to a cluster. They found that the exposure time indicator (ETI) model, which includes a categorical term for each level of exposure time, consistently achieved better coverage and lower bias for estimating the time-averaged treatment effect (TATE) compared to other parametric formulations, although the ETI model had less efficiency and wider confidence intervals. More recently, Kenny et al.\cite{Kenny2022} found that when the true treatment effect depends on the exposure time, misspecification of the treatment effect structure with a constant treatment effect assumption can lead to severely biased estimates, which can even converge to a value of opposite sign to the true TATE; they suggested the ETI model and its variants for more robust inference under time-varying effects. To improve upon the efficiency of the ETI model in the presence of a larger number of periods, Maleyeff et al.\cite{Maleyeff2022} proposed a random-effects counterpart of the ETI model, as well as a permutation testing procedure to detect exposure-time treatment effect heterogeneity. Under the potential outcomes framework, Wang et al.\cite{Wang2024} proved that the linear mixed model is robust against misspecification of the covariate effects, the random-effects structure, and the error structure, as long as the treatment effect structure is correctly specified, highlighting that the stepped wedge design, in general, has a more stringent requirement for model specification compared to a traditional parallel-arm design. Beyond exposure-time treatment effect heterogeneity, Chen and Li \cite{Chen2024} proposed estimands when there is potential treatment effect heterogeneity over calendar time and informative cluster sizes; they provided an independent estimating equations estimator for model-assisted inference about each estimand in SW-CRTs. Lee et al.\cite{Lee2024} further demonstrated that the constant treatment effect model carries less bias when the treatment effect heterogeneity depends on calendar-time rather than exposure-time.

Almost all previous work has focused on analytic models where the treatment effect is only present after the treatment adoption, under the assumption of no anticipation. This assumption, however, may not always hold in practice because stakeholders, providers, or individuals who are aware of the treatment adoption timing (especially when blinding is challenging or infeasible) can inadvertently change their behaviors in anticipation of the forthcoming intervention. Formally, we define the anticipation effect as the potential for cluster members to change their behavior prior to the intervention adoption simply because they have foreknowledge of the intervention timing. This phenomenon has been discussed in previous empirical studies. For instance, Pfingsten et al.\citep{Pfingsten2001} conducted a randomized controlled trial to investigate whether anticipation of pain and fear-avoidance beliefs led to behavioral avoidance; in that context, patients with chronic low back pain who were informed to expect pain from a task had lower performance levels than those who were not given such information, although objectively this could not happen. As another example, students in a school anticipating the introduction of a new teaching method may study more diligently in expectation of improved learning conditions according to expectancy-value theory.\citep{Meyer2019} Similar anticipatory responses have been discovered in studies that examined behavior changes in response to expected insurance price adjustments,\cite{AronDine2015} and reported in a variety of contexts, including microeconomics,\cite{Malani2015} neuroeconomics,\cite{Knutson2008} and crime control studies.\cite{Ariel2021} In stepped wedge designs, it is always recommended to avoid anticipation effects through careful study design and monitoring, such as blinding clusters to treatment adoption timing and only revealing the randomization sequence shortly prior to the introduction of the intervention. However, a systematic review by Taljaard et al.\cite{Taljaard2017} has found that most SW-CRTs lack full transparency in reporting the blinding throughout the trial. Furthermore, compared to parallel-arm designs, the randomized staggered rollout of clusters introduces additional implementation challenges, which could make blinding impractical from time to time; see an example in Hemming et al,\cite{Hemming2018} where the nature of the intervention makes complete blinding impossible. In general, it is necessary to separate the true intervention effect from changes in outcomes caused by anticipation alone.

In this article, we investigate the implications of the anticipation effect when SW-CRTs are analyzed by linear mixed models. We systematically examine the behaviors of different linear mixed model formulations and estimators in the following settings: 1) when an anticipation effect exists, 2) when an exposure time-varying treatment effect exists, and 3) a combination of 1) and 2). Under these settings, we exhaust 16 combinations of true and working linear mixed models (Table \ref{tab::true_vs_working_models}), which are then classified into four scenarios: correct specification, over-specification, under-specification, and mismatch. Because correct specification scenarios do not introduce bias and over-specification scenarios generally lead to efficiency loss in analysis without introducing bias, we will focus on the remaining seven scenarios where the working model misspecifies the true model, i.e., the under-specification and mismatch scenarios. Among these, we find that only one under-specification scenario has been previously studied in Kenny et al.,\cite{Kenny2022} and there has been little prior work on the remaining scenarios. To quantify the impact under misspecification and provide theoretical insights, we derive a series of formulas that characterize the direction and magnitude of bias under six misspecified scenarios. For analytical tractability and clear interpretations, we assume a first-order anticipation effect, which only occurs one period before the intervention adoption (Figure \ref{figure_sw}). However, a higher-order anticipation effect may exist, i.e., the anticipation effect may occur more than one period before the intervention adoption. Therefore, we provide additional discussions on higher-order anticipation effects in Section \ref{sec::conclusion} and Web Appendix E. Finally, we also derive new sample size formulas that explicitly account for anticipation effects, exposure-time treatment effect heterogeneity, or both in SW-CRTs and illustrate the impact of anticipation on sample size and power.

The rest of this article is structured as follows. In Section \ref{sec::setup}, we introduce the estimands and different model formulations. In Section \ref{sec::est}, we analytically examine the behaviors of linear mixed model estimators under misspecification. In Section \ref{sec::sample_size}, we derive sample size formulas accounting for the anticipation effect in SW-CRTs and compare statistical power when the anticipation effect is unavoidable. In Section \ref{sec::simulation}, we report the results of simulation studies comparing the performance of different linear mixed model estimators under different true data-generating models. In Section \ref{sec::da}, we provide a reanalysis of a completed SW-CRT dataset accounting for anticipation. In Section \ref{sec::conclusion}, we conclude with possible extensions and areas for future research.

\section{Setting, estimands, and model formulations}\label{sec::setup}

\subsection{Setting and notation}

\begin{table}[t]
    \centering
    \caption{Expressions of the TATE $\Delta$ and the point treatment effect $\delta(s_{ij})$ under models given in \eqref{model::HH}-\eqref{model::ETI-ANT}.}
    \label{tab::estimand}
    \begin{threeparttable}
    \begin{tabular}{cccc}
    \toprule
    & \multicolumn{2}{c}{\textbf{Model expressions}}\\
    \cmidrule{2-3}
    \textbf{Estimand} & \textbf{HH / HH-ANT} & \textbf{ETI / ETI-ANT} \\
    \hline
    $\Delta$ & $\delta$ & $Q^{-1}\sumq \delta(q)$ \\
    $\delta(1)$ & $\delta$ & $\delta(1)$ \\
    $\delta(2)$ & $\delta$ & $\delta(2)$ \\
    \vdots & \vdots & \vdots \\
    $\delta(J-1)$ & $\delta$ & $\delta(J-1)$ \\
    \bottomrule
    \end{tabular}\smallskip
    \end{threeparttable}
\end{table}

We consider a cross-sectional and complete SW-CRT with $I$ clusters (indexed by $i\in\{1,\ldots,I\}$), $J$ equally spaced periods (indexed by $j\in\{1,\ldots,J\}$), and $K_{ij}$ individuals per cluster-period (indexed by $k\in\{1,\ldots,K_{ij}\}$). All clusters are in the control condition when $j=1$ and the treatment condition when $j=J$. Let $Y_{ijk}$ be the observed outcome for individual $k$ from cluster $i$ in period $j$, and $Z_{ij}$ be the binary treatment indicator which is equal to 1 when cluster $i$ in period $j$ is under treatment and 0 otherwise. There are $Q=J-1$ treatment sequences (indexed by $q\in\{1,\ldots,Q\}$) characterized by unique period indices when the clusters begin to receive treatment (treatment adoption time) and the number of clusters in the treatment sequence $q$ is $I_q$ with $\sumq I_q = I$. To focus on the main ideas, we consider SW-CRTs without a pre-planned transition or implementation period and return to a discussion on this topic in Section \ref{sec::conclusion}. If cluster $i$ is in treatment sequence $q$, where the treatment starts at period $j$, then its treatment indicator vector is $\bfZ_i=(Z_{i1},\ldots,Z_{iJ})'$ with $Z_{i1} = 0$, $Z_{iJ}=1$, and $Z_{ij'}=1$ if $Z_{ij} = 1$ for $j'>j$. Whenever applicable, the anticipation effect is assumed to exist $\ell$ periods before the treatment adoption time for each cluster. Let $\bfA_i = (A_{i1},A_{i2},\ldots,A_{iJ})'$ denote the anticipation indicator vector that depends on $\bfZ_i$ and $\ell$. Specifically, let $j^*\in\{1,\ldots,J\}$ be treatment adoption period---the period index with $Z_{ij^*}=1$ and $Z_{i,j^*-1}=0$ for cluster $i$. The anticipation effect indicator $A_{ij}$ is then formally defined as
\begin{align*}
    A_{ij} = \begin{cases}
                1, & j^*-\ell \leq j < j^*,\\
                0, & \text{otherwise.}
             \end{cases}
\end{align*}
For simplicity, we assume $\ell = 1$ hereafter (first-order anticipation), but our results can be generalized to higher-order anticipation where $\ell > 1$. We return to a discussion in Section \ref{sec::conclusion}. Under the first-order anticipation assumption, we use $s_{ij} = j - j^* + 1 \in \{1, \ldots, J-1\}$ to denote the exposure time of cluster $i$ in period $j$ for $j \geq j^*$.

\subsection{Estimands}

We first describe a linear mixed model with generic model representation in Li et al.\cite{Li2020} and then introduce specific model variants. The ingredients for a linear mixed model are
\begin{align}
    Y_{ijk} = \text{{secular trend}} + \text{{treatment effect}} + \text{{heterogeneity}}+\text{{residual error}}, \label{model::general}
\end{align}
where $Y_{ijk}$ is assumed to be continuous and normally distributed. Here, the secular trend term consists of design points for the background secular trend without treatment and period-specific parameters encoding systematic time effects; the treatment effect term is the change in the mean outcome at each period from each treatment sequence and includes the treatment effect parameter of interest; the heterogeneity term represents the cluster-specific, time-specific, and/or individual-specific departure from the cluster average, and is usually represented by random effects; the residual error term is independent and identically distributed as $\epsilon_{ijk} \sim \calN(0, \sigma^2)$. Section 3.3 in Li et al.\cite{Li2020} introduced possible treatment effect structures to represent the treatment effect, including the scenario where the treatment effect term can be expressed as an effect curve $\delta(s)$, with $s=j-q \in \{1,\ldots,J-1\}$ indicating the duration of exposure or exposure time. Following Kenny et al.\cite{Kenny2022} and Maleyeff et al.,\cite{Maleyeff2022} given $J$ periods and the maximum value of exposure time is $J-1$, this article focuses on addressing a primary estimand of interest---the (exposure) time-averaged treatment effect (TATE):
\begin{align*}
    \Delta = \frac{1}{J-1}\sums\delta(s).
\end{align*}
In this definition, $\delta(s)$ is the point treatment effect under exposure time $s$. In the special scenario where the treatment effect curve does not change over time, the treatment effect term becomes a constant, with $\delta(s) = \delta$ for $s \in \{1,\ldots,J-1\}$, and the TATE is identical to each point treatment effect regardless of exposure time (i.e., $\Delta=\delta$). Of note, the estimand $\Delta$ we focus on averages over all possible exposure times. However, this is not the only possible estimand, as one may be interested in the effect at a specific exposure time, or the effects averaged over certain sub-intervals, depending on the application.

\begin{table}[t]
    \centering
    \caption{Specification of true and working models with an anticipation effect and time-varying treatment effects (TVEs). A check mark ($\newcheckmark$) indicates the presence of an effect, and a cross mark ($\newcrossmark$) indicates the absence of an effect. In the ``References'' column, bias formulas are derived for five under-specification and two mismatch scenarios in Section \ref{sec::est}; closed-form variance expressions are provided for all working models under 16 scenarios in Section \ref{sec::sample_size}; simulation studies are conducted for 14 selected non-trivial scenarios in Section \ref{sec::simulation}.}\label{tab::true_vs_working_models}
    \resizebox{\linewidth}{!}{
    \begin{threeparttable}
        \begin{tabular}{ccccccccc}
            \toprule
            \multirow{2}{*}{\textbf{Index}} & \multicolumn{3}{@{}c@{}}{\textbf{True Data Generating Process}} & \multicolumn{3}{c}{\textbf{Working Specification}} & \multirow{2}{*}{\textbf{Scenario}} & \multirow{2}{*}{\textbf{References}}\\
            & Model & Anticipation & TVE & Model & Anticipation & TVE & \\
            \midrule
            1 & HH & \newcrossmark & \newcrossmark & HH & \newcrossmark & \newcrossmark & Correct & see Hussey and Hughes\cite{Hussey2007}\\
            2 & HH & \newcrossmark & \newcrossmark & HH-ANT & \newcheckmark & \newcrossmark & Over-specification & Section \ref{sec::sample_size} and \ref{sec::simulation_I}\\
            3 & HH & \newcrossmark & \newcrossmark & ETI & \newcrossmark & \newcheckmark & Over-specification & Section \ref{sec::sample_size}\\
            4 & HH & \newcrossmark & \newcrossmark & ETI-ANT & \newcheckmark & \newcheckmark & Over-specification & Section \ref{sec::sample_size}\\
            \midrule
            5 & HH-ANT & \newcheckmark & \newcrossmark & HH & \newcrossmark & \newcrossmark & Under-specification & Section \ref{sec::HH_1}, \ref{sec::sample_size}, and \ref{sec::simulation_II}\\
            6 & HH-ANT & \newcheckmark & \newcrossmark & HH-ANT & \newcheckmark & \newcrossmark & Correct & Section \ref{sec::sample_size} and \ref{sec::simulation_II}\\
            7 & HH-ANT & \newcheckmark & \newcrossmark & ETI & \newcrossmark & \newcheckmark & Mismatch & Section \ref{sec::ETI_1}, \ref{sec::sample_size}, and \ref{sec::simulation_II}\\
            8 & HH-ANT & \newcheckmark & \newcrossmark & ETI-ANT & \newcheckmark & \newcheckmark & Over-specification & Section \ref{sec::sample_size} and \ref{sec::simulation_II}\\
            \midrule
            9 & ETI & \newcrossmark & \newcheckmark & HH & \newcrossmark & \newcrossmark & Under-specification & see Kenny et al.\cite{Kenny2022}\\
            10 & ETI & \newcrossmark & \newcheckmark & HH-ANT & \newcheckmark & \newcrossmark & Mismatch & Section \ref{sec::HH-ANT_1}, \ref{sec::sample_size}, and \ref{sec::simulation_III}\\
            11 & ETI & \newcrossmark & \newcheckmark & ETI & \newcrossmark & \newcheckmark & Correct & Section \ref{sec::sample_size} and \ref{sec::simulation_III}\\
            12 & ETI & \newcrossmark & \newcheckmark & ETI-ANT & \newcheckmark & \newcheckmark & Over-specification & Section \ref{sec::sample_size} and \ref{sec::simulation_III}\\
            \midrule
            13 & ETI-ANT & \newcheckmark & \newcheckmark & HH & \newcrossmark & \newcrossmark & Under-specification & Section \ref{sec::HH_2}, \ref{sec::sample_size}, and \ref{sec::simulation_IV}\\
            14 & ETI-ANT & \newcheckmark & \newcheckmark & HH-ANT & \newcheckmark & \newcrossmark & Under-specification & Section \ref{sec::HH-ANT_2}, \ref{sec::sample_size}, and \ref{sec::simulation_IV}\\
            15 & ETI-ANT & \newcheckmark & \newcheckmark & ETI & \newcrossmark & \newcheckmark & Under-specification & Section \ref{sec::ETI_2}, \ref{sec::sample_size}, and \ref{sec::simulation_IV}\\
            16 & ETI-ANT & \newcheckmark & \newcheckmark & ETI-ANT & \newcheckmark & \newcheckmark & Correct & Section \ref{sec::sample_size} and \ref{sec::simulation_IV}\\
            \bottomrule
        \end{tabular}\smallskip
    \end{threeparttable}
    }
\end{table}

\subsection{Model Formulations} \label{sec::formulation}

For simplicity, we assume $Y_{ijk}$ is continuous and primarily focus on the random intercept model to deliver our key ideas. We provide a discussion on alternative random-effects specifications in Section \ref{sec::conclusion}. We define the secular trend term as $\mu + \beta_2 \bbone(j=2)+\ldots+\beta_J \bbone(j=J)$ and the heterogeneity term as $\alpha_i$, where $\mu$ is the grand mean, $\beta_j$ is the secular trend parameter for period $j$ ($\beta_1=0$ for identifiability), and $\alpha_i\sim\calN(0,\tau^2)$ is a cluster-level random effect for cluster $i$. The Hussey and Hughes (HH) model for the constant treatment effect is written as:\cite{Hussey2007}
\begin{align}
    Y_{ijk} = \mu + \beta_j + \delta Z_{ij} + \alpha_i + \epsilon_{ijk}, \label{model::HH}
\end{align}
where $\delta$ is the treatment effect parameter and $\epsilon_{ijk}\sim\calN(0,\sigma^2)$ is the residual error term. To accommodate the anticipation effect, the HH model is modified by including the anticipation indicator (HH-ANT):
\begin{align}
    Y_{ijk} = \mu + \beta_j + \gamma A_{ij} + \delta Z_{ij} + \alpha_i + \epsilon_{ijk}, \label{model::HH-ANT}
\end{align}
where $\gamma$ is the anticipation effect parameter. The HH and HH-ANT models assume no exposure-time heterogeneity in the treatment effects. To allow exposure-time heterogeneity, we also introduce the exposure time indicator (ETI) model:\cite{Nickless2018,Granston2014}
\begin{align}
    Y_{ijk} = \mu + \beta_j + \delta(s_{ij}) Z_{ij} + \alpha_i + \epsilon_{ijk}, \label{model::ETI}
\end{align}
where $\delta(s_{ij})$ is the treatment effect with exposure time $s_{ij}$ observed in cluster-period $(i,j)$. Note that $\delta(s_{ij})$ is unique for each exposure time and can capture any functional form of exposure-time treatment effect heterogeneity. Similarly, we can add the anticipation effect to the ETI model (ETI-ANT):
\begin{align}
    Y_{ijk} = \mu + \beta_j + \gamma A_{ij} + \delta(s_{ij}) Z_{ij} + \alpha_i + \epsilon_{ijk}. \label{model::ETI-ANT}
\end{align}
In Table \ref{tab::estimand}, we list the expressions of TATE $\Delta$ and point treatment effect $\delta(s)$ under models given in \eqref{model::HH}-\eqref{model::ETI-ANT}, and then exhaust all 16 combinations of true and working model specification scenarios in Table \ref{tab::true_vs_working_models}. The purpose of Table \ref{tab::true_vs_working_models} is to provide an overview of scenarios when exposure time and/or anticipation are possibly ingredients in modeling the treatment effect structure under the general model \eqref{model::general}. Among the scenarios, 4 scenarios include correct model specification, and hence, no bias arises; 5 scenarios include model over-specification (in the treatment effect structure), and hence, no bias arises except for possible efficiency loss. We focus on the remaining 7 scenarios where the working model misspecifies the true model, i.e., under-specification and mismatch. Among these, only one under-specification scenario has been previously studied in Kenny et al.\cite{Kenny2022} To fill the knowledge gap, we next examine the behaviors of the linear mixed model estimators under the remaining six scenarios, i.e., when either or both 1) the assumption of no anticipation and 2) the assumption of no exposure-time heterogeneity are violated.

\section{Behaviors of linear mixed model estimators under model misspecification} \label{sec::est}

In this section, to generate insights into the behaviors of linear mixed model estimators under model misspecification, we consider a standard cross-sectional and complete SW-CRT. Specifically, we assume $I$ is divisible by $Q$ and treatment sequence $q$ has $I_q = I/Q$ clusters for $q \in \{1,\ldots,Q\}$. The treatment adoption time for clusters in treatment sequence $q$ is in period $j=q+1$. Furthermore, we assume equal numbers of individuals per cluster in each period $K_{ij}=K$ for all $i$ and $j$.

\subsection{Working Model: HH}\label{sec::HH}

We first consider the scenario where the working model is HH, which does not account for the anticipation effect or exposure-time heterogeneity, and examine the behaviors of the HH model estimator when the true model is either HH-ANT or ETI-ANT. For notational clarity, we denote the treatment effect estimator under the HH model as $\widehat{\delta}^{\text{HH}}$, and use the superscript to explicitly represent parameters in a specific model.

\subsubsection{True Model: HH-ANT}\label{sec::HH_1}

For ease of notation, we use $\phi = \tau^2 / (\tau^2 + \sigma^2 / K)$ to denote the correlation of the cluster-period means---a quantity that depends on $K$.\cite{Matthews2017} This quantity is different from the intra-cluster correlation coefficient (ICC), $\rho = \tau^2/(\tau^2 + \sigma^2)$, and $\phi > \rho$ when $K > 1$. Theorem \ref{theorem::HH_1} gives the expectation of the estimator $\widehat{\delta}^{\text{HH}}$, which is biased even for the constant treatment effect when the anticipation effect is present.

\begin{theorem}\label{theorem::HH_1}
    \emph{In a standard SW-CRT with data generated from the HH-ANT model, the expectation of the under-specified HH model estimator $\widehat{\delta}^{\text{HH}}$ is
    \begin{align}
        \bbE(\widehat\delta^\HH) = \delta^\HHANT+\omega_\HH^{\HHANT} \gamma^\HHANT, \label{theorem::HH_1_eqn}
    \end{align}
    where $\delta^\HHANT$ is the true treatment effect and $\gamma^\HHANT$ is the true anticipation effect. The weight of $\gamma^\HHANT$ is given by
    \begin{align*}
        \omega_\HH^\HHANT =-\frac{6(1+\phi Q)}{(Q+1)(2+\phi Q)}.
    \end{align*}}
\end{theorem}

The weight, $\omega_\HH^{\HHANT}$, is negative because $\phi > 0$, indicating that, without accounting for the anticipation effect, the HH model estimator $\widehat\delta^\HH$ under- or overestimates the treatment effect depending on the direction of the anticipation effect. The magnitude of the bias depends on design parameters $\phi$ and $Q$, and the true anticipation effect $\gamma^\HHANT$, and this magnitude becomes larger when $Q$ is smaller, $\phi$ is closer to 1, and $\gamma^\HHANT$ is larger. A numerical example illustrating the bias is given in Figure \ref{figure_HH}(a). In Figure \ref{figure_HH}(a), we assume a constant treatment effect and a positive anticipation effect, presenting both the true effect curve and the estimated effect curves under the HH working model; clearly, the results confirm that $\widehat\delta^\HH$ underestimates $\delta^\HHANT$.

\begin{figure}[t]
    \centering
    \includegraphics[width=\linewidth]{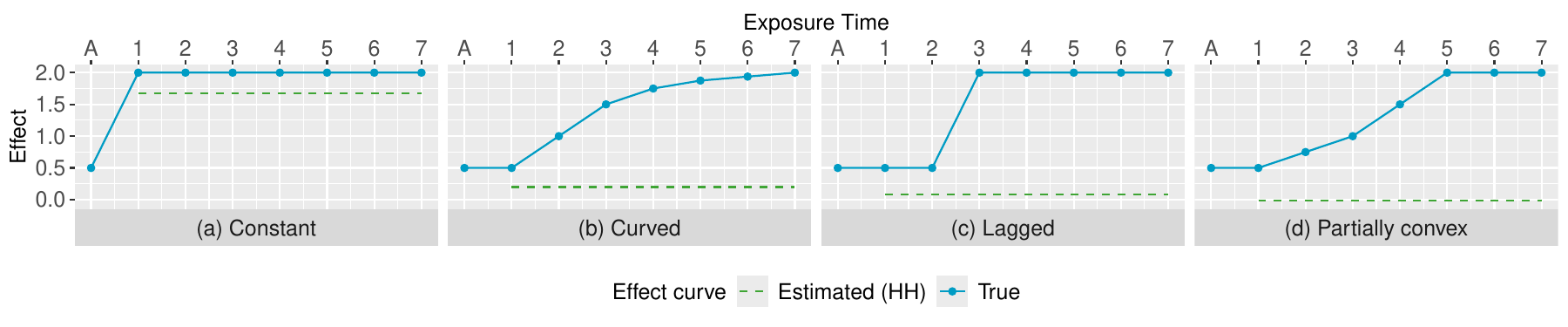}
    \caption{Four types of true treatment effect curves with their estimated effect curves under the HH working model. We consider standard cross-sectional SW-CRTs with $I = 28$, $J = 8$, and $K = 50$, which gives $Q=7$ treatment sequences with $I_q=4$ clusters each and $\phi \approx 0.80$. The true data-generating processes are HH-ANT (constant, panel (a)) or ETI-ANT (curved, lagged, partially convex, panels (b)-(d)), where $\mu = 14$, $\beta_j = 0.5\sin\{2\pi (j-1)/7\}$, $\tau = 0.141$, $\sigma = 0.5$, $\gamma^\HHANT = \gamma^\ETIANT = 0.5$, and $\ell = 1$. The anticipation effect occurs at exposure time $A$, and the point treatment effects occur when exposure times range from 1 to 7. For the constant treatment effect, $\delta^\HHANT = 2$ for $s\in\{1,\ldots,7\}$. For the curved treatment effect, $\delta^\ETIANT(1) = 0.5$, $\delta^\ETIANT(s) = 2-0.5^{s-2}$ for $s\in\{2,\ldots,6\}$, and $\delta^\ETIANT(7) = 2$. For the lagged treatment effect, $\delta^\ETIANT(s) = 0.5$ for $s\in\{1,2\}$ and $\delta^\ETIANT(s) = 2$ for $s\in\{3,\ldots,7\}$. For the partially convex treatment effect, $\delta^\ETIANT(s) = 0.25+0.25s$ for $s\in\{1,2,3\}$, $\delta^\ETIANT(4) = 1.5$, and $\delta^\ETIANT(s) = 2$ for $s\in\{5,6,7\}$. The blue line is the true effect curve, and the green dashed line is the estimated effect curve using an HH model via 2,000 simulated datasets.} \label{figure_HH}
\end{figure}

\subsubsection{True Model: ETI-ANT}\label{sec::HH_2}

When the true model is ETI-ANT, the assumptions of no anticipation and no exposure-time heterogeneity are both violated under the HH working model. Theorem \ref{theorem::HH_2} gives a closed-form expression for the HH model estimator and its expectation when the true model is ETI-ANT.

\begin{theorem}\label{theorem::HH_2}
    \emph{For a standard SW-CRT with data generated from the ETI-ANT model, we denote the mean outcome in treatment sequence $q \in \{1, \ldots, Q\}$ in period $j$ by $\overline Y_j(q)$. Then, the HH estimator $\widehat\delta^\HH$ can be expressed as:
    \begin{align*}
        \widehat\delta^\HH = \frac{12(1 + \phi Q)}{Q(Q+1)(\phi Q^2 + 2Q - \phi Q - 2)} \sumq \sumj \left[ Q\bbone(j > q) - j + 1 + \frac{\phi Q(2q - Q - 1)}{2(1+\phi Q)} \right] \overline Y_j(q).
    \end{align*}
    Furthermore, the expectation of $\widehat\delta^\HH$ can be expressed as a weighted average:
    \begin{align}
        \bbE(\widehat\delta^\HH) = \sumjp \left\{\pi_\HH^\ETIANT(j) \delta^\ETIANT(j) + \omega_\HH^{\ETIANT}(j) \gamma^\ETIANT \right\}, \label{theorem::HH_2_eqn}
    \end{align}
    where $\delta^\ETIANT(j)$ is the true point treatment effect with exposure time $j \in \{1, \ldots, J-1\}$ and $\gamma^\ETIANT$ is the true anticipation effect. The weights of $\delta^\ETIANT(j)$ and $\gamma^\ETIANT$ are given by
    \begin{align*}
        \pi_\HH^\ETIANT(j)=\frac{6(j - Q - 1)\{(1 + 2\phi Q)j - (1 + \phi + \phi Q)Q\}}{Q(Q + 1)(\phi Q^2 + 2Q - \phi Q - 2)},
    \end{align*}
    and
    \begin{align*}
        \omega_\HH^\ETIANT(j)=- \frac{6(\phi Q^2 - \phi Q + 2j - 2)}{Q(\phi Q^3 + 2Q^2 - \phi Q - 2)},
    \end{align*}
    respectively.}
\end{theorem}

Interestingly, the weight, $\pi_\HH^\ETIANT(j)$, is identical to that derived in Kenny et al.\cite{Kenny2022} assuming the true model is ETI, and hence following their results, it is possible that $\pi_\HH^\ETIANT(j) < 0$ for some $j\in\{1,\ldots,J\}$. When the true model is ETI-ANT, the anticipation effect introduces $\gamma^\ETIANT$ into \eqref{theorem::HH_2_eqn}, with its weight $\omega_\HH^{\ETIANT}(j)<0$ and $0<|\omega_\HH^\ETIANT(j)|<1$ for all $j\in\{1,\ldots,J\}$. Figure \ref{figure_coeff_HH} illustrates the relationships between weights, $\pi_\HH^\ETIANT(j)$ and $\omega_\HH^\ETIANT(j)$, and design parameters, $\phi$ and $Q$, where the absolute values of all weights increase as $\phi$ increases from 0 to 1. Of note, for $\pi_\HH^\ETIANT(j)$, we see that three out of seven weights are negative as $\phi \to 1$ when $Q = 7$; for $\omega_\HH^\ETIANT(j)$, each weight becomes more negative as $\phi \to 1$. 

To provide more insights, we assume that the true anticipation effect $\gamma^\ETIANT > 0$ (similar interpretations hold for $\gamma^\ETIANT < 0$). Using numerical examples (shown in Figures \ref{figure_HH}(b)-(d)), when the exposure-time heterogeneity (curved, lagged, partially convex) exists, both the estimated effect curves under the HH working model underestimate the true effect curve and lie entirely below it. In the partially convex scenario, the estimated effect curve is negative or near zero even when all the point treatment effects are positive. Although $\pi_\HH^\ETIANT(1)$ is positive and largest in magnitude (see Figure \ref{figure_coeff_HH}), the first exposure time has minimal contribution because $\delta^\ETIANT(1) = \gamma^\ETIANT = 0.5$ is close to zero, and negative weights at later exposure times drive $\widehat\delta^\HH$ toward zero or even into negative values. This implies that if we use the HH working model to estimate the TATE when the true model is ETI-ANT, the estimator $\widehat\delta^\HH$ might give an estimate of the opposite sign to the true TATE. This finding is consistent with that in Kenny et al.\cite{Kenny2022} in the absence of anticipation. Furthermore, given $\omega_\HH^\ETIANT(j)<0$, the positive anticipation effect worsens the downward biases, and hence ignoring anticipation can introduce a higher bias toward the opposite direction on top of the bias due to ignoring exposure-time treatment effect heterogeneity.

\begin{figure}[t]
    \centering
    \includegraphics[width=\linewidth]{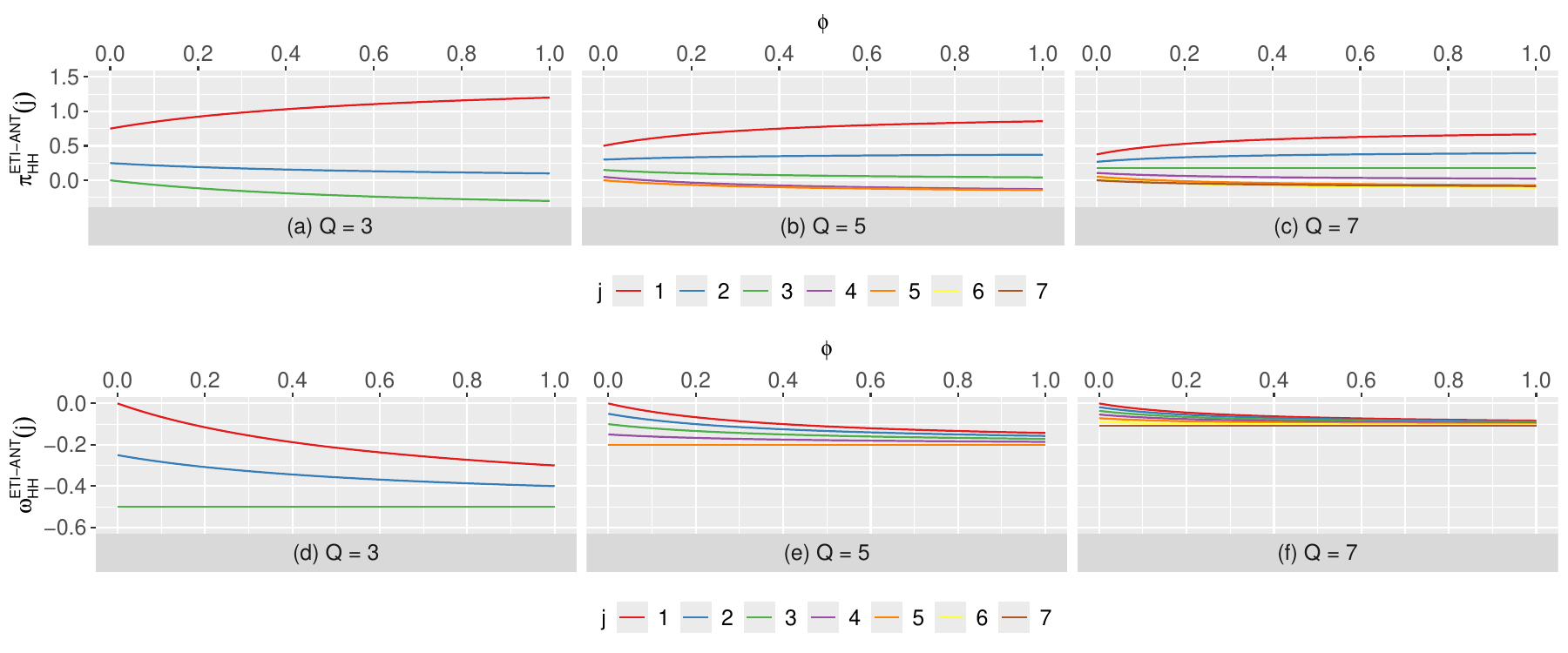}
    \caption{Weights $\pi_\HH^\ETIANT(j)$ for $j\in\{1,\ldots,J-1\}$ in $\bbE(\widehat\delta^\HH)$ under the HH working model (panels (a)-(c)), and $\omega_\HH^\ETIANT(j)$ for $j\in\{1,\ldots,J-1\}$ in $\bbE(\widehat\delta^\HH)$ under the HH working model (panels (d)-(f)), assuming the true model is ETI-ANT. Here, $Q\in\{3, 5, 7\}$ and $\phi\in[0,1]$.}\label{figure_coeff_HH}
\end{figure}

\subsection{Working Model: ETI}\label{sec::ETI}

We next consider the scenario where the working model is ETI, which accounts for exposure-time heterogeneity but not for the anticipation effect, and examine the behaviors of the ETI model estimators when the true model is either HH-ANT or ETI-ANT. For simplicity of exposition, we discuss via an example of a standard SW-CRT with $J=3$ and $Q=2$, which is sufficient to provide the key insights. The results can be extended to more general settings.

\subsubsection{True Model: HH-ANT}\label{sec::ETI_1}

When the true model is HH-ANT, the assumption of no anticipation is violated under the ETI working model. In this scenario of mismatch, Theorem \ref{theorem::ETI_1} gives the expectations of the ETI model estimators $\widehat\delta^\ETI(1)$ and $\widehat\delta^\ETI(2)$ when the true model is HH-ANT.

\begin{theorem}\label{theorem::ETI_1}
    \emph{For a standard SW-CRT with data generated from the HH-ANT model, the expectations of the ETI model estimators $\widehat\delta^\ETI(1)$ and $\widehat\delta^\ETI(2)$ are
    \begin{align}
        \bbE\{\widehat\delta^\ETI(1)\} = \delta^\HHANT-(1+\phi)\gamma^\HHANT \quad\text{and}\quad\bbE\{\widehat\delta^\ETI(2)\} = \delta^\HHANT-(1+2\phi)\gamma^\HHANT, \label{theorem::ETI_1_eqn}
    \end{align}
    where $\delta^\HHANT$ is the true treatment effect and $\gamma^\HHANT$ is the true anticipation effect.}
\end{theorem}

The weights of $\gamma^\HHANT$ in \eqref{theorem::ETI_1_eqn} are negative, indicating that, without accounting for the anticipation effect, the ETI model estimators $\widehat\delta^\ETI(1)$ and $\widehat\delta^\ETI(2)$ under- or overestimate the treatment effect $\delta^\HHANT$ depending on the direction of the anticipation effect. Moreover, Theorem \ref{theorem::ETI_1} implies that $|\bbE\{\widehat\delta^\ETI(1)\} - \delta^\HHANT| < |\bbE\{\widehat\delta^\ETI(2)\} - \delta^\HHANT|$ (since $\phi > 0$). That is, there is more bias in $\widehat\delta^\ETI(2)$ than in $\widehat\delta^\ETI(1)$ for estimating $\delta^\HHANT$. Of note, both $|\bbE\{\widehat\delta^\ETI(1)\} - \delta^\HHANT|$ and $|\bbE\{\widehat\delta^\ETI(2)\} - \delta^\HHANT|$ increase as $|\gamma^\HHANT|$ increases and $\phi$ approaches 1. We illustrate these results via a numerical example in Figure \ref{figure_ETI}(a), where, when the treatment effect is constant, $\widehat\delta^\ETI(s)$ underestimates $\delta^\HHANT$ for all $s$. The estimated effect curves decline as the exposure time increases, indicating $\widehat\delta^\ETI(s)$ becomes more biased as $s$ increases. Similar patterns are observed from the results of a more general scenario with $J=8$ in Figure \ref{figure_ETI}(e).

\begin{figure}[t]
    \centering
    \includegraphics[width=\linewidth]{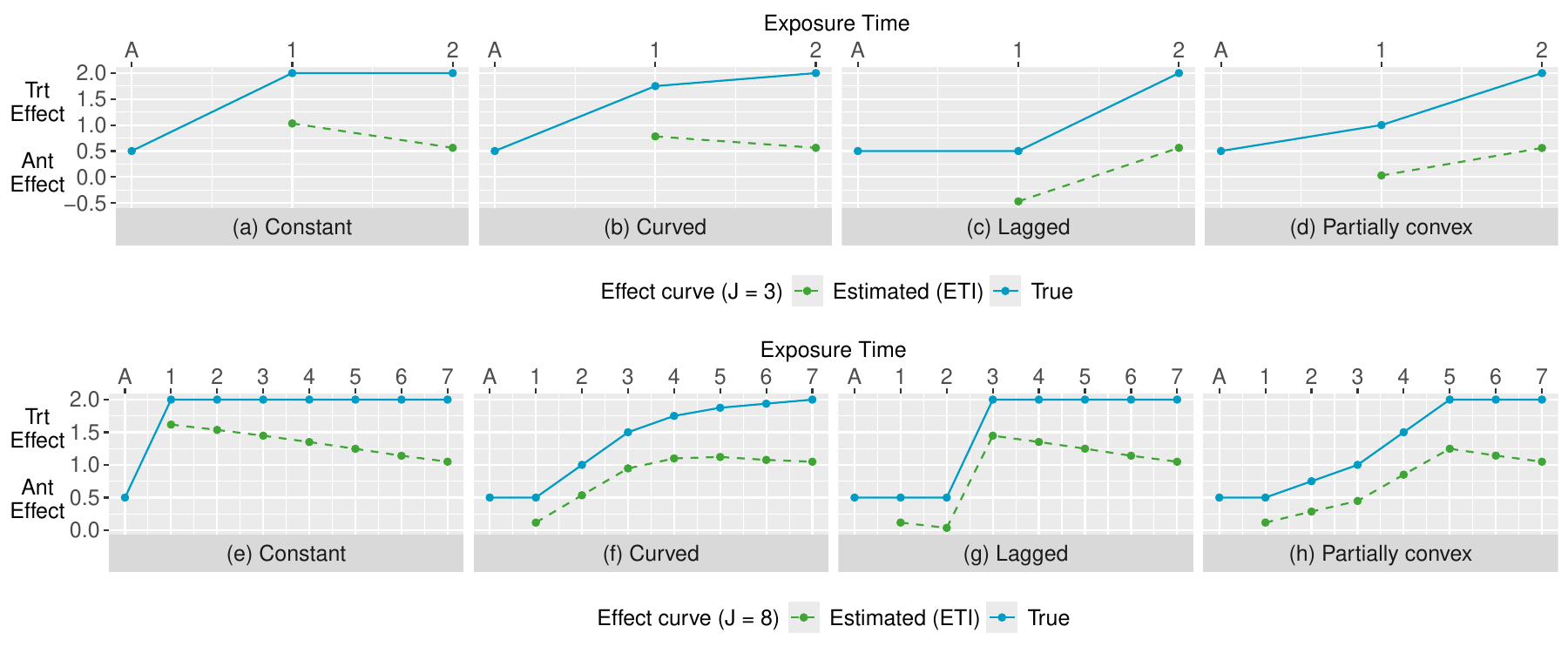}
    \caption{Four types of true treatment effect curves with their estimated effect curves under the ETI working model. We consider standard cross-sectional SW-CRTs with $I = 28$, $J = 3$ for panels (a)-(d) (and $J = 8$ for panels (e)-(h)), and $K = 50$. When $J=3$, there are $Q=2$ treatment sequences with $I_q=14$ clusters each and $\phi \approx 0.80$. The true data-generating processes are HH-ANT (constant, panels (a) and (e)) or ETI-ANT (curved, lagged, partially convex, panels (b)-(d) and (f)-(h)), where $\mu = 14$, $\beta_j = 0.5\sin\{2\pi (j-1)/7\}$, $\tau = 0.141$, $\sigma = 0.5$, $\gamma^\HHANT = \gamma^\ETIANT = 0.5$, and $\ell = 1$. The anticipation effect occurs at exposure time $A$ (one period before the treatment adoption), and the point treatment effects occur at other periods. When the treatment effect is constant, $\delta^\HHANT = 2$ for $s\in\{1,2\}$. When the treatment effect is curved, $\delta^\ETIANT(1) = 1.75$ and $\delta^\ETIANT(2) = 2$. When the treatment effect is lagged, $\delta^\ETIANT(1) = 0.5$ and $\delta^\ETIANT(2) = 2$. When the treatment effect is partially convex, $\delta^\ETIANT(1) = 1$ and $\delta^\ETIANT(2) = 2$. The blue line is the true effect curve, and the green dashed line is the estimated effect curve using an ETI model via 2,000 simulated datasets. When $J=8$, the settings are the same as the ones in Figure \ref{figure_HH}.} \label{figure_ETI}
\end{figure}

\subsubsection{True Model: ETI-ANT}\label{sec::ETI_2}

When the true model is ETI-ANT, the assumption of no anticipation is violated under the ETI working model. Theorem \ref{theorem::ETI_2} gives expectations of the ETI model estimators $\widehat\delta^\ETI(1)$ and $\widehat\delta^\ETI(2)$ when the true model is ETI-ANT.

\begin{theorem}\label{theorem::ETI_2}
    \emph{For a standard SW-CRT with data generated from the ETI-ANT model, the expectations of the ETI model estimators $\widehat\delta^\ETI(1)$ and $\widehat\delta^\ETI(2)$ are
    \begin{align}
        \bbE\{\widehat\delta^\ETI(1)\} = \delta^\ETIANT(1)-(1+\phi)\gamma^\ETIANT \quad\text{and}\quad\bbE\{\widehat\delta^\ETI(2)\} = \delta^\ETIANT(2)-(1+2\phi)\gamma^\ETIANT, \label{theorem::ETI_2_eqn}
    \end{align}
    where $\delta^\ETIANT(j)$ is the true point treatment effect with exposure time $j \in \{1, 2\}$ and $\gamma^\ETIANT$ is the true anticipation effect.}
\end{theorem}

Theorem \ref{theorem::ETI_2} resembles Theorem \ref{theorem::ETI_1}, and also implies that $|\bbE\{\widehat\delta^\ETI(1)\}-\delta^\ETIANT(1)| < |\bbE\{\widehat\delta^\ETI(2)\}-\delta^\ETIANT(2)|$. These results can be illustrated via numerical examples given in Figures \ref{figure_ETI}(b)-(d). When the exposure-time heterogeneity (curved, lagged, partially convex) exists, the gap between the estimated and true curves widens as the exposure time increases, indicating $\widehat\delta^\ETI(s)$ becomes more biased for larger $s$. Similar patterns are observed from the results of a more general scenario with $J=8$ in Figures \ref{figure_ETI}(f)-(h).

\subsection{Working Model: HH-ANT}\label{sec::HH-ANT}

We next consider the scenario where the working model is HH-ANT, which accounts for the anticipation effect but not the exposure-time heterogeneity, and examine the behaviors of the HH-ANT model estimators when the true model is either ETI or ETI-ANT.

\subsubsection{True Model: ETI}\label{sec::HH-ANT_1}

When the true model is ETI, Theorem \ref{theorem::HH-ANT_1} provides a closed-form expression for the HH-ANT model estimators and their expectations.

\begin{theorem}\label{theorem::HH-ANT_1}
    \emph{For a standard SW-CRT with data generated from the ETI model, the HH-ANT estimators $\widehat\delta^\HHANT$ and $\widehat\gamma^\HHANT$ can be expressed as:
    \begin{align*}
        \widehat\delta^\HHANT &= \frac{6}{\phi Q(1 - Q^2) + (2Q^2 - 3Q + 1)(\phi Q + 1)}\displaybreak[0]\\
        &\quad\times\sumq\sumj \left[-\phi(Q - 2q + 1) + \frac{\phi Q + 1}{Q}\{Q\bbone(j=q) + 2Q\bbone(j>q) + \bbone(j=J) - 2j + 1\}\right]\overline Y_j(q),
    \end{align*}
    and
    \begin{align*}
        \widehat\gamma^\HHANT &=\frac{1}{(Q - 1)(\phi Q^2 - 2\phi Q + 2Q - 1)} \displaybreak[0]\\
        &\quad\times\sumq\sumj \left[(Q+1)(\phi Q + 2)\bbone(j=q)+6(\phi Q + 1)\bbone(j>q)+\frac{(Q+1)(\phi Q + 2)}{Q}\bbone(j = J)\right.\displaybreak[0]\\
        &\qquad\qquad\qquad\left.-4\phi Q - 6\phi j + 6\phi q + 2\phi - 2 - \frac{6j}{Q} + \frac{4}{Q}\right]\overline Y_j(q).
    \end{align*}
    Furthermore, the expectations of $\widehat\delta^\HHANT$ and $\widehat\gamma^\HHANT$ can be expressed as a weighted average of the point treatment effects:
    \begin{align}
        \bbE(\widehat\delta^\HHANT) = \sumjp \pi_\HHANT^\ETI(j) \delta^\ETI(j) \quad\text{and}\quad \bbE(\widehat\gamma^\HHANT) = \sumjp \psi_\HHANT^\ETI(j) \delta^\ETI(j),\label{theorem::HH-ANT_delta_gamma_1}
    \end{align}
    where $\delta^\ETI(j)$ is the true point treatment effect with exposure time $j \in \{1, \ldots, J-1\}$. The weights of $\delta^\ETI(j)$ for $\bbE(\widehat\delta^\HHANT)$ and $\bbE(\widehat\gamma^\HHANT)$ are given by
    \begin{align*}
        \pi_\HHANT^\ETI(j)&=\frac{6 \left( \phi Q^3 - 3 \phi Q^2 j + \phi Q^2 + Q^2 + 2 \phi Q j^2 - 2 \phi Q j + \phi Q - 2 Q j + j^2 \right)}{Q \left( \phi Q^3 - 3 \phi Q^2  + 2 Q^2 + 2 \phi Q - 3 Q + 1 \right)},
    \end{align*}
    and
    \begin{align*}
        \psi_\HHANT^\ETI(j)&=\frac{2\phi Q^3 - 8\phi Q^2 j + 5\phi Q^2 + Q^2 + 6\phi Qj^2 - 8\phi Qj + 3\phi Q - 4Qj + Q + 3j^2 - j}{Q(Q - 1)(\phi Q^2 - 2\phi Q + 2Q - 1)},
    \end{align*}
    respectively.}
\end{theorem}

\begin{figure}[t]
    \centering
    \includegraphics[width=\linewidth]{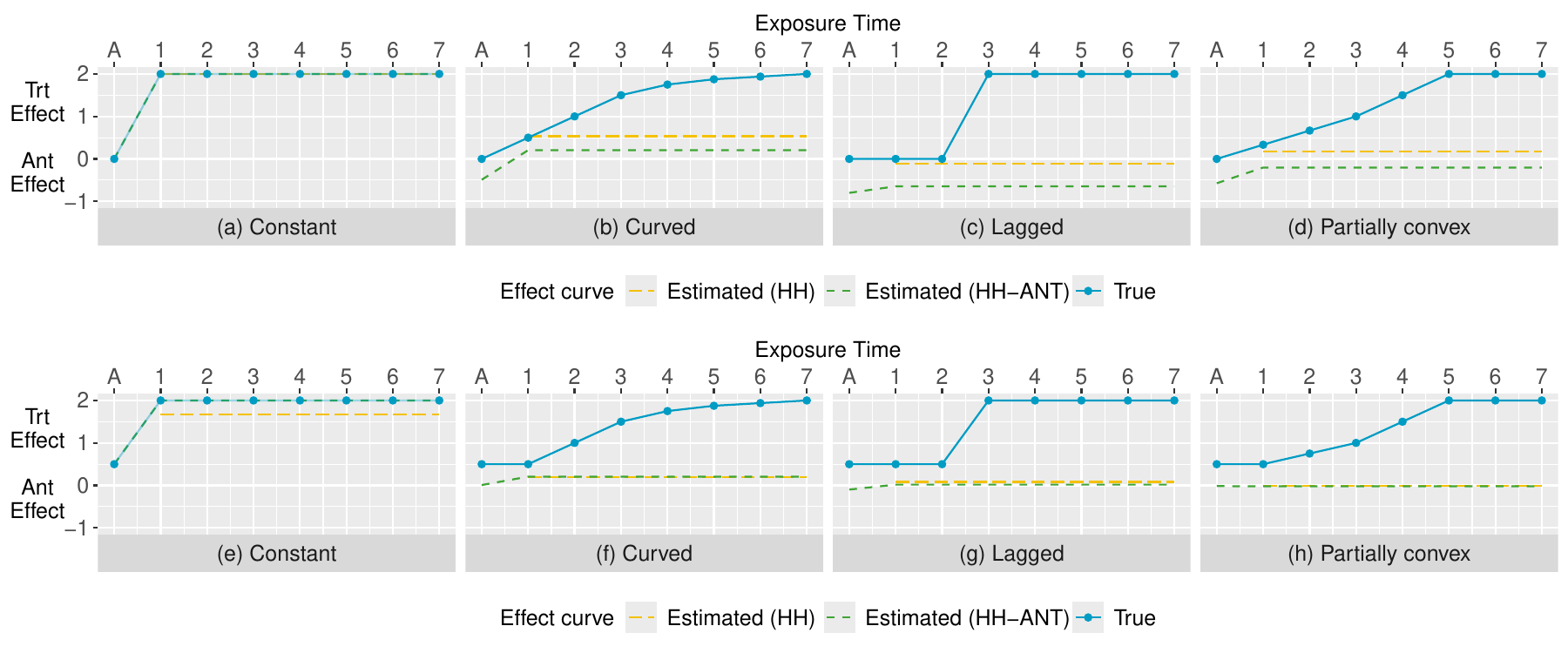}
    \caption{Four types of true treatment effect curves with their estimated curves under the HH or HH-ANT working model. We consider standard cross-sectional SW-CRTs with $I = 28$, $J = 8$, and $K = 50$, which gives $Q=7$ treatment sequences with $I_q=4$ clusters each and $\phi \approx 0.80$. The true data-generating processes are HH (constant, panel (a)), ETI (curved, lagged, partially convex, panels (b)-(d)), HH-ANT (constant, panel (e)) or ETI-ANT (curved, lagged, partially convex, panels (f)-(h)), where $\mu = 14$, $\beta_j = 0.5\sin\{2\pi (j-1)/7\}$, $\tau = 0.141$, $\sigma = 0.5$, $\gamma^\ETI = 0$ (no anticipation, panels (a)-(d)) or $\gamma^\ETIANT = 0.5$ (with anticipation, panels (e)-(f)), and $\ell = 1$. The anticipation effect occurs at exposure time $A$ (one period before the treatment adoption), and the point treatment effects occur when exposure times range from $1$ to $7$. In panels (a)-(d), for the constant treatment effect, $\delta^\HH = 2$ for $s\in\{1,\ldots,7\}$; for the curved treatment effect, $\delta^\ETI(1) = 0.5$, $\delta^\ETI(s) = 2-0.5^{s-2}$ for $s\in\{2,\ldots,6\}$, and $\delta^\ETI(7) = 2$; for the lagged treatment effect, $\delta^\ETI(s) = 0$ for $s\in\{1,2\}$ and $\delta^\ETI(s) = 2$ for $s\in\{3,\ldots,7\}$; for the partially convex treatment effect, $\delta^\ETI(s) = s/3$ for $s\in\{1,2,3\}$, $\delta^\ETI(4) = 1.5$, and $\delta^\ETI(s) = 2$ for $s\in\{5,6,7\}$. In panels (e) and (f), the settings are the same as that in Figure \ref{figure_HH}. The blue line is the true effect curve, and the yellow (or green) dashed line is the estimated effect curve using the HH (or HH-ANT) model via 2,000 simulated datasets.} \label{figure_HH-ANT}
\end{figure}

The results in Theorem \ref{theorem::HH-ANT_1} can be better illustrated via examples in Figures \ref{figure_HH-ANT}(a)-(d). As a check, when the treatment effect is constant, both the HH and HH-ANT models give correct estimated effect curves that align with the true curve (shown in  Figure \ref{figure_HH-ANT}(a)). When the exposure-time heterogeneity (curved, lagged, and partially convex) exists, as shown in Figures \ref{figure_HH-ANT}(b)-(d), both the HH and HH-ANT model estimators underestimate the true treatment effect curves, but interestingly, with more severe bias observed under the HH-ANT model. This suggests that when the true model is ETI, ignoring the exposure-time treatment effect heterogeneity using the HH-ANT working model could be less advantageous than using the HH; that is, a mismatch in model specification can be even ``worse'' than under-specification for treatment effect estimation in SW-CRTs.

As a side note, although the anticipation effect is often not of primary interest, the HH-ANT model does provide an estimate for such an effect. Therefore, it is of theoretical interest to understand the behavior of the anticipation effect estimator. The expectation of the anticipation effect estimator $\bbE(\widehat\gamma^\HHANT)$ depends on point treatment effects $\delta^\ETI(j)$ with weights $\psi_\HHANT^{\ETI}(j)$ for $j\in\{1,\ldots,J-1\}$. In Web Figure 1, most of the weights $\psi_\HHANT^\ETI(j)$ are negative for $j$ closer to $J-1$ and $\phi$ closer to 1. For instance, when $Q=3$, except for $\psi_\HHANT^\ETI(1)$, both $\psi_\HHANT^\ETI(2)$ and $\psi_\HHANT^\ETI(3)$ become more negative when $\phi$ approaches 1. Similar patterns can be observed for $Q\in\{5, 7\}$. In Figures \ref{figure_HH-ANT}(b)-(d), when the exposure-time heterogeneity (curved, lagged, and partially convex) exists, $\bbE(\widehat\gamma^\HHANT)$ is negative even when the true anticipation effect is zero, further highlighting the risk of bias under the mismatch scenario.

\begin{remark}
    Web Figure 2 compares $\pi_\HHANT^\ETI(j)$ with $\pi_\HH^\ETI(j)$ using examples with $Q\in\{3, 5, 7\}$ and $\phi\in[0,1]$. Specifically, when $Q=3$ and $\phi=1$, the HH working model only has one negative weight, i.e., $\pi_\HH^\ETI(3) < 0$, whereas the HH-ANT working model has two, i.e., $\pi_\HHANT^\ETI(j) < 0$ for $j\in\{2,3\}$. Similar patterns are observed for $Q=5$ and $Q=7$ with $\phi$ approaching 1. The number of negative $\pi_\HHANT^\ETI(j)$ is greater than that of $\pi_\HH^\ETI(j)$ for $j\in\{1,\ldots,J-1\}$, which, to a certain extent, can explain the reason why when the true model is ETI, the HH-ANT working model estimators result in more severe underestimation.
\end{remark}

\subsubsection{True Model: ETI-ANT}\label{sec::HH-ANT_2}

When the true model is ETI-ANT, the assumption of no exposure-time heterogeneity is violated under the HH-ANT working model. Theorem \ref{theorem::HH-ANT_2} provides a closed-form expression for the HH-ANT model estimators and their expectations when the true model is ETI-ANT.

\begin{theorem}\label{theorem::HH-ANT_2}
    \emph{For a standard SW-CRT with data generated from the ETI-ANT model, the HH-ANT estimators $\widehat\delta^\HHANT$ and $\widehat\gamma^\HHANT$ can be expressed as:
    \begin{align*}
        \widehat\delta^\HHANT &= \frac{6}{\phi Q(1 - Q^2) + (2Q^2 - 3Q + 1)(\phi Q + 1)}\displaybreak[0]\\
        &\quad\times\sumq\sumj \left[-\phi(Q - 2q + 1) + \frac{\phi Q + 1}{Q}\{Q\bbone(j=q) + 2Q\bbone(j>q) + \bbone(j=J) - 2j + 1\}\right]\overline Y_j(q),
    \end{align*}
    and
    \begin{align*}
        \widehat\gamma^\HHANT &=\frac{1}{(Q - 1)(\phi Q^2 - 2\phi Q + 2Q - 1)} \displaybreak[0]\\
        &\quad\times\sumq\sumj \left[(Q+1)(\phi Q + 2)\bbone(j=q)+6(\phi Q + 1)\bbone(j>q)+\frac{(Q+1)(\phi Q + 2)}{Q}\bbone(j = J)\right.\displaybreak[0]\\
        &\qquad\qquad\qquad\left.-4\phi Q - 6\phi j + 6\phi q + 2\phi - 2 - \frac{6j}{Q} + \frac{4}{Q}\right]\overline Y_j(q).
    \end{align*}
    Furthermore, the expectation of $\widehat\delta^\HHANT$ and $\widehat\gamma^\HHANT$ can be expressed as a weighted average of the point treatment effects:
    \begin{align}
        \bbE(\widehat\delta^\HHANT) = \sumjp \pi_\HHANT^\ETIANT(j) \delta^\ETIANT(j) \quad\text{and}\quad \bbE(\widehat\gamma^\HHANT) = \gamma^\ETIANT + \sumjp \psi_\HHANT^\ETIANT(j) \delta^\ETIANT(j),\label{theorem::HH-ANT_delta_gamma_2}
    \end{align}
    where  $\delta(j)$ is the true point treatment effect with exposure time $j \in \{1, \ldots, J-1\}$. The weights of $\delta(j)$ for $\bbE(\widehat\delta^\HHANT)$ and $\bbE(\widehat\gamma^\HHANT)$ are given by
    \begin{align*}
        \pi_\HHANT^\ETIANT(j)&=\frac{6 \left( \phi Q^3 - 3 \phi Q^2 j + \phi Q^2 + Q^2 + 2 \phi Q j^2 - 2 \phi Q j + \phi Q - 2 Q j + j^2 \right)}{Q \left( \phi Q^3 - 3 \phi Q^2  + 2 Q^2 + 2 \phi Q - 3 Q + 1 \right)},
    \end{align*}
    and
    \begin{align*}
        \psi_\HHANT^\ETIANT(j)&=\frac{2\phi Q^3 - 8\phi Q^2 j + 5\phi Q^2 + Q^2 + 6\phi Qj^2 - 8\phi Qj + 3\phi Q - 4Qj + Q + 3j^2 - j}{Q(Q - 1)(\phi Q^2 - 2\phi Q + 2Q - 1)},
    \end{align*}
    respectively.}
\end{theorem}

Results in Theorem \ref{theorem::HH-ANT_2} are illustrated via numerical examples in Figure \ref{figure_HH-ANT}(e)-(h). As a check, when the treatment effect is constant, the HH-ANT model estimator gives the correct estimated effect curve that aligns with the true curve (Figure \ref{figure_HH-ANT}(e)). The HH model estimator underestimates the true treatment effect curve, consistent with the results in Section \ref{sec::HH_1}. When the exposure-time heterogeneity (curved, lagged, and partially convex) exists, both the HH and HH-ANT model estimators underestimate the true treatment effect curve and lie entirely below it (Figures \ref{figure_HH-ANT}(f)-(h)). 

For the anticipation effect estimator $\widehat\gamma^\HHANT$, its expectation $\bbE(\widehat\gamma^\HHANT)$ depends on the true anticipation effect $\gamma^\ETIANT$ and point treatment effects $\delta^\ETIANT(j)$ with weights $\psi_\HHANT^\ETIANT(j)$ for $j\in\{1,\ldots,J-1\}$. In Figures \ref{figure_HH-ANT}(e)-(h), when the treatment effect is constant, the HH-ANT model estimator gives a correct estimate of the anticipation effect (Figure \ref{figure_HH-ANT}(e)). When the exposure-time heterogeneity (curved, lagged, and partially convex) exists, the expectation $\bbE(\widehat\gamma^\HHANT)$ is negative even though the true anticipation effect is positive (Figures \ref{figure_HH-ANT}(f)-(h)), showcasing that there is generally a high risk of bias in this under-specification scenario.

\section{Variance formulas when the Anticipation Effect is accounted for}\label{sec::sample_size}

Although Section \ref{sec::est} primarily focuses on bias, it does not address variance when the anticipation effect has been accounted for. Below, we investigate the implications on variance when using models accounting for the anticipation effect. We will derive a set of variance formulas for the treatment effect estimators accounting for the anticipation effect, which serve two purposes. First, they can answer to what extent the variance of the treatment effect estimator will change when we additionally include an anticipation effect. Second, in settings where anticipation is unavoidable (e.g., a healthcare policy intervention announced to the public several months before its implementation), the variance formulas we derived can motivate closed-form sample size estimation methods to ensure sufficient power for studying the primary treatment effect.

We follow the setting in Section \ref{sec::setup} by considering a general SW-CRT with $I$ clusters, $J$ periods, and equal cluster-period sizes $K_{ij}=K$ as a convention. Unlike the standard SW-CRT assumed in Section \ref{sec::est} where each treatment sequence contains exactly $I/Q$ clusters, the general setting allows for unequal allocation with $I_q$ clusters in sequence $q$, where $\sumq I_q = I$. Also, we let $\lambda_1 = 1-\rho$ and $\lambda_2 = 1+(JK-1)\rho$, with the total variance $\sigma_t^2 = \tau^2+\sigma^2$ and the ICC $\rho = \tau^2/\sigma_t^2$. Results of the variance expressions for the HH-ANT and ETI-ANT models are given in Theorem \ref{theorem::var}.

\begin{theorem} \label{theorem::var}
    \emph{For an SW-CRT with $I$ clusters, $J$ periods, and $K$ individuals per cluster-period, variances of treatment effect estimators are
    \begin{align} \label{var::HH-ANT_general}
        \bbV(\widehat\delta^\HHANT)= \frac{IJ\lambda_1\lambda_2\sigma_t^2}{K} \left\{\lambda_2\left(U^2+IJU-JW_1-IW_2-\frac{JW_5^2}{I^2-W_3}\right) + \lambda_1\left(IW_2-U^2\right)\right\}^{-1}
    \end{align}
    from the HH-ANT model, and
    \begin{align} \label{var::ETI-ANT_general}
        \bbV\left\{\frac{1}{J-1}\sums\widehat\delta^\ETIANT(s)\right\}= \frac{IJ\lambda_1\lambda_2\sigma_t^2}{K(J-1)^2} \mathbf{1}_{J-1}'\left\{\lambda_2\left(\bfU_1^{\otimes 2}+IJ\bfU_2-J\bfW_1-I\bfW_2 - \frac{J \bfW_5^{\otimes 2}}{I^2-W_3}\right) + \lambda_1\left(I\bfW_2-\bfU_1^{\otimes 2}\right)\right\}^{-1}\mathbf{1}_{J-1}
    \end{align}
    from the ETI-ANT model. In \eqref{var::HH-ANT_general}, $U = \sumi\sumj Z_{ij}$, $W_1 = \sumj (\sumi Z_{ij})^2$ and $W_2 = \sumi(\sumj Z_{ij})^2$ are design constants that only depend on the treatment assignment; $W_3 = (\sumi \bfA_i')(\sumi \bfA_i)$ and $W_5 = (\sumi \bfZ_i')(\sumi \bfA_i)$ are design constants that additionally depend on the anticipation indicator. When $\ell = 1$, $W_4$ is omitted because $W_4 = \sumi \bfA_i'\mathbf{1}_J\mathbf{1}_J'\bfA_i = I$, where $\mathbf{1}_J$ is a column vector of $J$ ones. In \eqref{var::ETI-ANT_general}, $\bfU_1 = \sumi \bfX_i'\mathbf{1}_J$, $\bfU_1^{\otimes 2} = \bfU_1\bfU_1'$, $\bfU_2 = \sumi \bfX_i'\bfX_i$, $\bfW_1 = (\sumi \bfX_i')(\sumi \bfX_i)$, $\bfW_2 = \sumi \bfX_i'\mathbf{1}_J\mathbf{1}_J'\bfX_i$, $W_3 = (\sumi \bfA_i')(\sumi \bfA_i)$, and $\bfW_5 = (\sumi \bfX_i')(\sumi \bfA_i)$. Here, $\bfX_i = \{\mathbf{0}_{(j^*-1)\times (J-1)}', (\bfI_{J-j^*+1}, \mathbf{0}_{(J-j^*+1)\times (j^*-2)})'\}'$ for the treatment adoption time $j^*\geq 2$ (by the definition of SW-CRTs), where $\mathbf{0}_{(j^*-1)\times (J-1)}$ is a $(j^*-1)\times (J-1)$ zero matrix and $\bfI_{J-j^*+1}$ is a $(J-j^*+1)\times (J-j^*+1)$ identity matrix.}
\end{theorem}

Theorem \ref{theorem::var} indicates that variances of the treatment effect estimators, under both the HH-ANT and ETI-ANT models, are independent of the true values of the point treatment effects and anticipation effect. Instead, they only depend on the trial resources (number of periods, clusters, individuals), treatment assignment (design constants), and the ICC, which are either pre-specified during the trial planning stage or approximated based on prior knowledge. Under the same setting in Theorem \ref{theorem::var}, Hussey and Hughes\cite{Hussey2007} have shown that the variance of the treatment effect estimator under the HH model (without anticipation effect) is
    \begin{align} \label{var::HH_general}
        &\bbV(\widehat\delta^\HH)= \frac{(\sigma_t^2/K)IJ\lambda_1\lambda_2}{(U^2+IJU-JW_1-IW_2)\lambda_2-(U^2-IW_2)\lambda_1}.
    \end{align}
Maleyeff et al.\cite{Maleyeff2022} have shown that the variance of the treatment effect estimator under the ETI model (without anticipation effect) is
\begin{align} \label{var::ETI_general}
    &\bbV\left\{\frac{1}{J-1}\sums\widehat\delta^\ETI(s)\right\}= \frac{IJ\lambda_1\lambda_2\sigma_t^2}{K(J-1)^2} \mathbf{1}_{J-1}'\left\{\lambda_2\left(\bfU_1^{\otimes 2}+IJ\bfU_2-J\bfW_1-I\bfW_2\right) + \lambda_1\left(I\bfW_2-\bfU_1^{\otimes 2}\right)\right\}^{-1}\mathbf{1}_{J-1}.
\end{align}
For the HH and HH-ANT models, the direct comparison of the general variance expressions \eqref{var::HH-ANT_general} and \eqref{var::HH_general} is challenging. Therefore, we evaluate them using a standard SW-CRT with fixed values of $(I, Q, K, \lambda_1, \lambda_2, \sigma_t^2)$ to give more insights. Under these conditions, the general variance expressions \eqref{var::HH-ANT_general} and \eqref{var::HH_general} reduce to the following simplified form
\begin{align*}
    \bbV(\widehat\delta^\HH) = \frac{12 Q \sigma_t^2 \lambda_1 \lambda_2}{I K (Q - 1)\{Q\lambda_1+(Q+2)\lambda_2\}}\quad\text{and}\quad\bbV(\widehat\delta^\HHANT) &= \frac{12 Q \sigma_t^2 \lambda_1 \lambda_2}{I K (Q - 1)\{Q\lambda_1+(Q-1)\lambda_2\}},
\end{align*}
which directly implies that $\bbV(\widehat\delta^\HH) < \bbV(\widehat\delta^\HHANT)$ due to the larger denominator in $\bbV(\widehat\delta^\HH)$. In contrast, for the ETI and ETI-ANT models, the general expressions \eqref{var::ETI-ANT_general} and \eqref{var::ETI_general} are directly comparable, and $\bbV(\widehat\Delta^\ETI) < \bbV(\widehat\Delta^\ETIANT)$ due to the smaller denominator in \eqref{var::ETI-ANT_general}. These comparisons clearly suggest that accounting for the anticipation effect would generally lead to an inflated variance of the treatment effect estimator, regardless of exposure-time treatment effect heterogeneity.

\begin{figure}[t]
    \centering
    \includegraphics[width=0.8\linewidth]{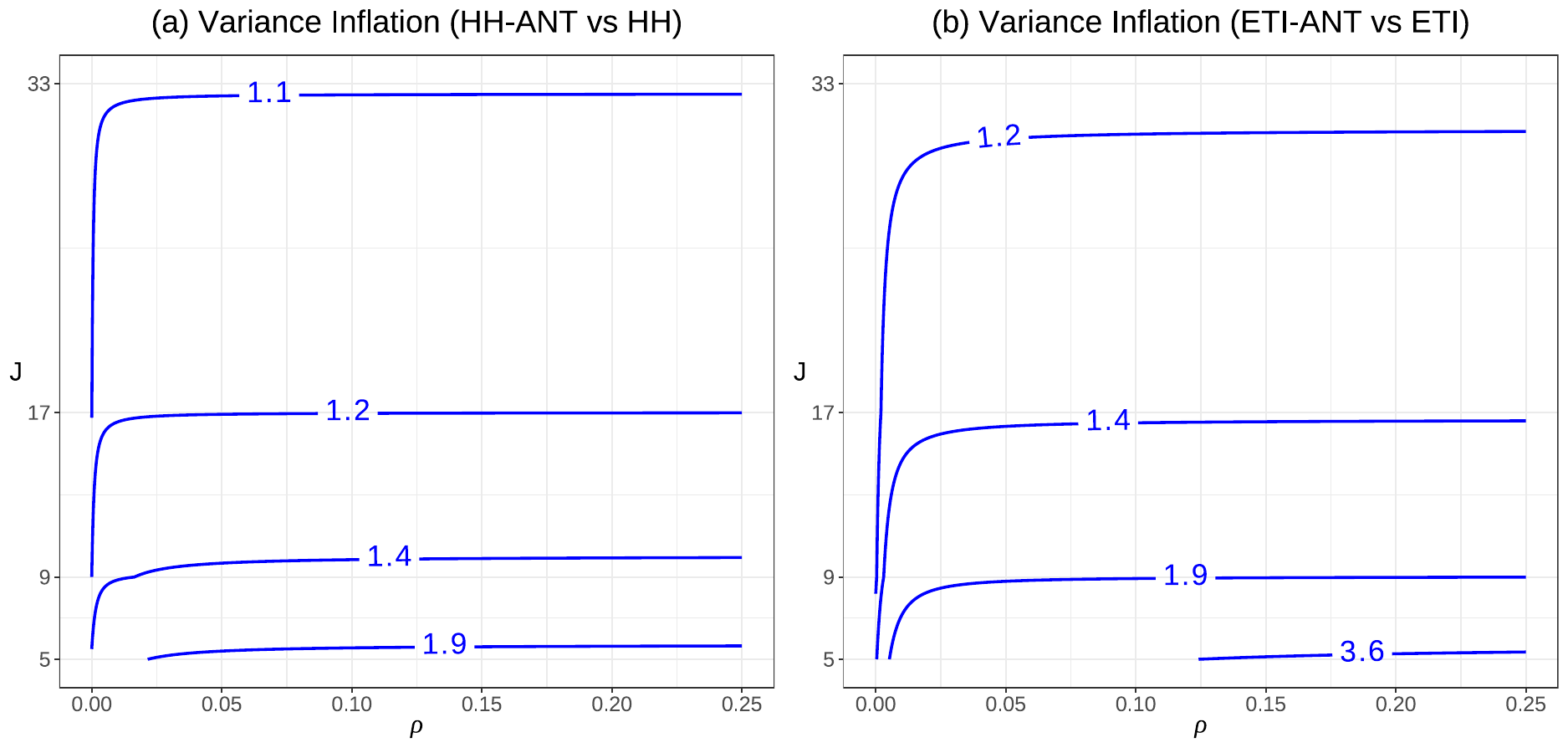}
    \caption{Contour plots of variance inflation as functions of $\rho$ and $J$ in a general SW-CRT. Here, $I = 32$, $J \in \{5, 9, 17, 33\}$, $K = 100$, $\sigma^2 = 1$, and $\rho \in [0, 0.25]$. Blue lines are $\bbV(\widehat\delta^\HHANT)/\bbV(\widehat\delta^\HH)$ in panel (a) and $\bbV(\widehat\Delta^\ETIANT)/\bbV(\widehat\Delta^\ETI)$ in panel (b).}\label{figure_coeff_HH_vs_HH-ANT}\label{figure_variance_inflation}
\end{figure}

To offer additional insights, we empirically investigate the amount of variance inflation in Figure \ref{figure_variance_inflation} under different combinations of design parameters. Specifically, we plot the values of $\bbV(\widehat\delta^\HHANT)/\bbV(\widehat\delta^\HH)$ and $\bbV(\widehat\Delta^\ETIANT)/\bbV(\widehat\Delta^\ETI)$ with varying $\rho$ and $J$. These figures confirm that the variance inflation becomes more severe with larger $\rho$ and smaller $J$. For example, under a constant treatment effect scenario, the variance inflation due to accounting for the anticipation effect will be between 40\%-50\% even for a trial with $J=9$ periods over a wide range of non-trivial ICC values. The variance inflation only reduces to 10\% for a trial with over 30 periods when $\rho>0.02$. In comparison, the variance inflation is even higher under an exposure-time-dependent treatment effect scenario. Based on these explorations, we conclude that the study team should reflect on the potential for anticipation effect before adopting the HH-ANT or ETI-ANT models for the sample size calculation due to the substantial amount of variance inflation. Addressing the anticipation effect when it does not exist could be a severely conservative approach that requires a much higher level of sample size and, hence, should be considered with caution.

When the investigators consider the anticipation effect of being highly likely, the general variance expressions \eqref{var::HH-ANT_general}-\eqref{var::ETI-ANT_general} can be used to facilitate power calculations. For instance, to test the hypothesis $H_0: \delta = 0$ versus $H_1: \delta = \delta^*$ in an SW-CRT with $I$ clusters, $J$ periods, and $K$ individuals per cluster-period, a Wald test can be used. At the significance level of $\alpha$, the power of a two-tailed test is expressed as
\begin{align}
    \text{Power} = \Phi \left\{{|\delta^*|}/{\sqrt{\bbV(\widehat\delta)}}-z_{1-\alpha/2}\right\},\label{power_formula}
\end{align}
where $\delta^*$ is the treatment effect under the alternative hypothesis, $\widehat\delta$ is the estimated treatment effect, $\Phi$ is the cumulative standard Normal distribution function, and $z_{1-\alpha/2}$ is the $(1-\alpha/2)$-th quantile of the standard Normal distribution. In a similar fashion, the power formula for the TATE can be retrieved by replacing $(\delta^*, \widehat\delta)$ in \eqref{power_formula} with $(\Delta^*, \widehat\Delta)$, where $\Delta^*$ is the true average treatment effect under the alternative hypothesis, and $\widehat\Delta$ is the estimated average treatment effect (see Web Appendix C for additional details). Of note, these formulas assume correctly specified models. When the working model is misspecified, however, the interpretation changes substantially. Consider the scenario where the true model is HH-ANT but the working model is HH. Under this misspecification, the estimator $\widehat\delta^\HH$ converges to $\delta^\HHANT + \omega_\HH^\HHANT \gamma^\HHANT$ rather than the true treatment effect $\delta^\HHANT$. This bias has implications for hypothesis testing. Specifically, when the true treatment effect is null ($\delta^\HHANT = 0$) but an anticipation effect exists ($\gamma^\HHANT \neq 0$), such a test carries inflated type I error rates and becomes invalid; in this case, the power of this test is generally less meaningful due to the invalid test size. A similar type I error inflation occurs when the working model is ETI and the true model is ETI-ANT.

\section{Simulation Studies}\label{sec::simulation}

We conduct simulation studies under different SW-CRT designs to evaluate the finite-sample performances of the linear mixed model estimators based on model formulations given in Section \ref{sec::formulation}. Scenario numbers in each section correspond to row indices in Table \ref{tab::true_vs_working_models}, with Web Tables 3 and 4 summarizing the specific research questions addressed by each study. For all simulations, we consider a standard SW-CRT with $J=9$, $I=32$, and $K=100$. We fix $\alpha_i \sim \calN(0, 0.141^2)$ and $\epsilon_{ijk} \sim \calN(0, 1)$, yielding an ICC of $\tau^2/(\tau^2+\sigma^2) \approx 0.0195$ with $\tau^2 = 0.141^2$ and $\sigma^2 = 1$. We also set $\mu = 0$ and $\beta_j = j$. Results are based on 2,000 simulated datasets for each scenario. In Simulation Studies III and IV, we choose a sinusoidal function $\delta^\ETI(s) = \delta^\ETIANT(s) = -1.41\sin\{2\pi(s-1)/7\}+0.12$ to represent exposure-time treatment effect heterogeneity, following one of the functional forms proposed in Maleyeff et al.\cite{Maleyeff2022} Of note, this function is different from the constant, curved, lagged, and partially convex patterns in Section \ref{sec::est}, and is meant to complement the earlier analytic results; we provide a visualization of this sinusoidal pattern in Web Figure 12. We return to a discussion of alternative functional forms and their implications in Section \ref{sec::conclusion}.

\subsection{Simulation Study I: True Model is HH (Scenarios 1-2)}\label{sec::simulation_I}

In Simulation Study I, outcomes are generated from the HH model in the absence of any anticipation effect and exposure time heterogeneity. The working models are HH and HH-ANT. We first set $\delta^\HH = 0$ to compare the type I error rate of these two working models.

Results are summarized in Table \ref{simulation::HH-ANT}. We found that treatment effect estimators from both working models are unbiased. The empirical standard deviation and the average model-based standard error of the HH model estimator are correspondingly smaller than those of the HH-ANT model estimator, which is consistent with the theoretical results. We additionally set $\delta^\HH = 0.075$ to compare the power of tests for the treatment effect based on the HH model and the HH-ANT model. Table \ref{simulation::HH-ANT} shows that the power of the test based on the HH model is greater than that based on the HH-ANT model, demonstrating the efficiency loss when unnecessarily accounting for anticipation.

\subsection{Simulation Study II: True Model is HH-ANT (Scenarios 5-8)}\label{sec::simulation_II}

In Simulation Study II, outcomes are generated from the HH-ANT model with $\gamma^\HHANT = 0.04$ and $\delta^\HHANT = 0.075$. We consider four working models in \eqref{model::HH}-\eqref{model::ETI-ANT} with results in Table \ref{simulation::HH-ANT}. If the working model fails to account for the anticipation effect, i.e., the HH model and the ETI model (defined in Section \ref{sec::formulation},) the treatment effect is underestimated, which is consistent with results in Theorems \ref{theorem::HH_1} and \ref{theorem::ETI_1}. Noticeably, the ETI model leads to more severe underestimation than the HH model, confirming that the mismatch scenario can inflate bias. On the other hand, the treatment effect estimators and the anticipation effect estimator from the HH-ANT model and the ETI-ANT model are unbiased. The empirical standard deviation and the average model-based standard error of $\widehat\delta^\HHANT$ are correspondingly smaller than those of $\widehat\Delta^\ETIANT$, which is consistent with the theoretical results of $\bbV(\widehat\delta^\HHANT) \approx 0.0240^2$ and $\bbV(\widehat\Delta^\ETIANT) \approx 0.0426^2$. 

In this empirical evaluation, tests for the treatment effect and the anticipation effect under the HH-ANT and the ETI-ANT working models all maintain the nominal type I error rate (5\%). Finally, the test for the treatment effect and the anticipation effect based on the HH-ANT model has the highest power. To provide a more in-depth investigation, Web Appendix D.1 presents a numerical study comparing the analytic power for the treatment effect between the HH model and the HH-ANT model in the presence of an anticipation effect. Interestingly, the relative power of these two models critically depends on the ICC $\rho$ and relative magnitude of the anticipation effect (to the treatment effect) $\gamma^\HHANT/\delta^\HHANT$. For example, in Web Figure 3 (a) under the scenario where $J = 5$ and $\delta^\HHANT = 0.1$, over a range of ICC parameters from 0 to 0.25, the test based on the HH-ANT model has higher power when the anticipation effect is relatively large (e.g., $\gamma^\HHANT/\delta^\HHANT>0.295$), whereas the test based on the HH model is more powerful otherwise (e.g., $\gamma^\HHANT/\delta^\HHANT\leq0.295$). Further discussions on power comparison between these two models can be found in Web Appendix D.1.

\begin{table}[t]
    \centering
    \caption{Results over $2,000$ simulated datasets generated under the setting in Simulation Studies I and II. $\sd(\cdot)$: the empirical (Monte Carlo) standard deviation. $\se(\cdot)$: the average model-based standard error. C$(\cdot)$ (\%): the empirical coverage percentage of 95\% confidence intervals. P$(\cdot)$ (\%): the statistical power of rejecting $H_0:\delta=0$ or $H_0:\gamma=0$.} \label{simulation::HH-ANT}
    \resizebox{\linewidth}{!}{
    \begin{tabular}{llcccccccccccc}
        \toprule
        \textbf{True Model} & \textbf{Working Model} & $\delta$ & $\widehat\delta$ & $\widehat\gamma$ & $\widehat\tau$ & $\sd(\widehat\delta)$ & $\se(\widehat\delta)$ & C$(\delta)$ (\%) & P$(\delta)$ (\%) & $\sd(\widehat\gamma)$ & $\se(\widehat\gamma)$ & C$(\gamma)$ (\%) & P$(\gamma)$ (\%) \\
        \midrule
        \multirow{4}{*}{HH} & HH & 0 & 0.0001 & - & 0.1402 & 0.0201 & 0.0203 & 95.50 & 4.50 & - & - & - & - \\ 
        & HH-ANT & & 0.0000 & -0.0001 & 0.1402 & 0.0239 & 0.0240 & 95.55 & 4.45 & 0.0222 & 0.0224 & 94.60 & 5.40 \\[2.5px]
        \cline{2-14} \\[-7.5pt]
        & HH & 0.075 & 0.0751 & - & 0.1402 & 0.0201 & 0.0203 & 95.50 & 96.25 & - & - & - & - \\ 
        & HH-ANT & & 0.0750 & -0.0001 & 0.1402 & 0.0239 & 0.0240 & 95.55 & 88.15 & 0.0222 & 0.0224 & 94.60 & 5.40 \\ 
        \midrule
        \multirow{4}{*}{HH-ANT} & HH & 0.075 & 0.0521 & - & 0.1403 & 0.0201 & 0.0203 & 79.85 & 73.35 & - & - & - & - \\ 
        & HH-ANT & & 0.0750 & 0.0399 & 0.1402 & 0.0239 & 0.0240 & 95.55 & 88.15 & 0.0222 & 0.0224 & 94.60 & 42.75 \\ 
        & ETI & & 0.0318 & - & 0.1407 & 0.0326 & 0.0324 & 73.85 & 15.85 & - & - & - & - \\ 
        & ETI-ANT & & 0.0756 & 0.0400 & 0.1402 & 0.0427 & 0.0424 & 95.45 & 42.75 & 0.0251 & 0.0253 & 94.45 & 34.80 \\ 
        \bottomrule
    \end{tabular}
    }
\end{table}

\subsection{Simulation Study III: True Model is ETI (Scenarios 9-12)}\label{sec::simulation_III}

In Simulation Study III, we focus on the setting where the outcomes are generated from the ETI model with $\delta^\ETI(s) = -1.41\sin\{2\pi(s-1)/7\}+0.12$ and $\Delta^\ETI = 0.12$. The assumption of no anticipation holds, but the assumption of no exposure-time heterogeneity is violated due to time-varying point treatment effects. As in Simulation Study II, we consider four working models. Results are given in Table \ref{simulation::ETI-ANT}.

The HH and HH-ANT model estimators do not account for exposure-time heterogeneity and give negatively biased estimates, which is consistent with results in Theorems \ref{theorem::HH_2} and \ref{theorem::HH-ANT_1}. The treatment effect estimator of the HH-ANT model is more negatively biased compared to the HH model estimator, consistent with the theoretical results discussed in Section \ref{sec::HH-ANT} and also confirming that the mismatch scenario can amplify the bias. In contrast, the ETI model and the ETI-ANT model estimators are consistent for both the treatment and anticipation effects, and tests based on these estimators maintain the nominal type I error rate. Also, the empirical standard deviation and the average model-based standard error of $\widehat\Delta^\ETI$ from the ETI model are correspondingly smaller than those of $\widehat\Delta^\ETIANT$ from the ETI-ANT model, which is consistent with the theoretical results of $\bbV(\widehat\Delta^\ETI) \approx 0.0325^2$ and $\bbV(\widehat\Delta^\ETIANT) \approx 0.0426^2$.

\begin{table}[t]
    \centering
    \caption{Results over $2,000$ simulated datasets generated under the setting in Simulation Studies III and IV. $\sd(\cdot)$: the empirical (Monte Carlo) standard deviation. $\se(\cdot)$: the average model-based standard error. C$(\cdot)$ (\%): the empirical coverage percentage of 95\% confidence intervals. P$(\cdot)$ (\%): the statistical power of rejecting $H_0:\Delta=0$ or $H_0:\gamma=0$.} \label{simulation::ETI-ANT}
    \resizebox{\linewidth}{!}{
    \begin{tabular}{llcccccccccccc}
        \toprule
        \textbf{True Model} & \textbf{Working Model} & $\Delta $ & $\widehat\Delta$ & $\widehat\gamma$ & $\widehat\tau$ & $\sd(\widehat\Delta)$ & $\se(\widehat\Delta)$ & C$(\Delta)$ (\%) & P$(\Delta)$ (\%) & $\sd(\widehat\gamma)$ & $\se(\widehat\gamma)$ & C$(\gamma)$ (\%) & P$(\gamma)$ (\%) \\ 
        \midrule
        \multirow{4}{*}{ETI} & HH & 0.12 & -0.8480 & - & 0.3402 & 0.0206 & 0.0235 & 0.00 & 100.00 & - & - & - & - \\ 
        & HH-ANT & & -1.1072 & -0.4302 & 0.3977 & 0.0247 & 0.0279 & 0.00 & 100.00 & 0.0224 & 0.0256 & 0.00 & 100.00 \\ 
        & ETI & & 0.1206 & - & 0.1402 & 0.0326 & 0.0324 & 94.65 & 95.75 & - & - & - & - \\ 
        & ETI-ANT & & 0.1206 & 0.0000 & 0.1402 & 0.0427 & 0.0424 & 95.45 & 80.75 & 0.0251 & 0.0253 & 94.45 & 5.55 \\
        \midrule
        \multirow{4}{*}{ETI-ANT} & HH & 0.12 & -0.8721 & - & 0.3454 & 0.0206 & 0.0235 & 0.00 & 100.00 & - & - & - & - \\ 
        & HH-ANT & & -1.1072 & -0.3902 & 0.3977 & 0.0247 & 0.0279 & 0.00 & 100.00 & 0.0224 & 0.0256 & 0.00 & 100.00 \\ 
        & ETI & & 0.0768 & - & 0.1407 & 0.0326 & 0.0324 & 73.85 & 66.50 & - & - & - & - \\ 
        & ETI-ANT & & 0.1206 & 0.0400 & 0.1402 & 0.0427 & 0.0424 & 95.45 & 80.75 & 0.0251 & 0.0253 & 94.45 & 34.80 \\ 
        \bottomrule
    \end{tabular}
    }
\end{table}

\subsection{Simulation Study IV: True Model is ETI-ANT (Scenarios 13-16)}\label{sec::simulation_IV}

In Simulation Study IV, we generate the outcomes from the ETI-ANT model with $\gamma^\ETIANT = 0.04$, $\delta^\ETIANT(s) = -1.41\sin\{2\pi(s-1)/7\}+0.12$, and $\Delta^\ETIANT = 0.12$. The assumptions of no anticipation and no exposure-time heterogeneity are both violated due to the anticipation effect and time-varying point treatment effects. Results for the four working models are given in Table \ref{simulation::ETI-ANT}. If the working model does not account for exposure-time heterogeneity, i.e., the HH model and the HH-ANT model, the estimators are negatively biased for both the treatment and anticipation effects, which is consistent with results in Theorems \ref{theorem::HH_2} and \ref{theorem::HH-ANT_2}. The treatment effect estimator from the HH-ANT model is more negatively biased compared to the HH model estimator. Interestingly, the anticipation effect estimator from the HH-ANT model is negative even though the true anticipation effect is positive. This aligns with the pattern of negative weights $\psi_\HHANT^\ETIANT(j)$ in Web Figure 1, where the bias caused by ignoring exposure-time heterogeneity can reverse the sign of $\widehat\gamma^\HHANT$ (Theorem \ref{theorem::HH-ANT_2}).

If the working model accounts for exposure-time heterogeneity but not the anticipation effect, i.e., the ETI model, the treatment effect is underestimated, consistent with results in Theorems \ref{theorem::ETI_2}. In contrast, the ETI-ANT model estimators are consistent for both the average treatment and anticipation effects, and tests based on these estimators maintain the nominal type I error rate. In Web Appendix D.2, we present a numerical study to compare the analytic power of the ETI and ETI-ANT models in the presence of an anticipation effect. In general, the comparative findings are similar to those reported for the HH model and the HH-ANT model in Web Appendix D.1, and depend on the ICC and the relative magnitude of anticipation effect to the true treatment effects.

\begin{figure}[t]
    \centering
    \includegraphics[width=\linewidth]{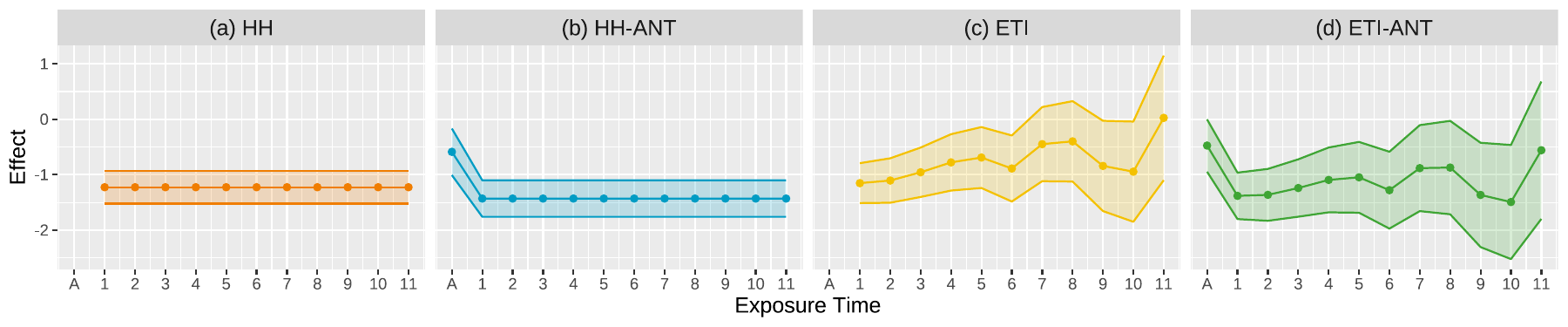}
    \caption{Estimated effect curves with 95\% confidence intervals for the REDAPS study using the HH model (panel (a)), the HH-ANT model (panel (b)), the ETI model (panel (c)), and the ETI-ANT model (panel (d)). The anticipation effect occurs at the exposure time $A$, and the point treatment effects occur when exposure times range from 1 to 11.} \label{figure_data_analysis_REDAPS}
\end{figure}

\section{An illustrative example with the REDAPS SW-CRT}\label{sec::da}

We analyze data from the Randomized Evaluation of Default Access to Palliative Services (REDAPS) trial, which is an SW-CRT that was conducted over a 32-month period.\cite{Courtright2024} This trial compared the intervention (a default order for palliative care consultation) with usual care (clinicians choose to order palliative care) among adults with chronic serious illnesses admitted to 11 Ascension hospitals across 8 US states. This study involved a population of $6,811$ patients who had lengths of stay of at least 72 hours and received a completed palliative care consult, consistent with the secondary outcome analysis in Courtright et al.\cite{Courtright2024} The 11 Ascension hospitals were treated as clusters, each assigned to one of 11 distinct sequences over 12 discrete periods. 

We specifically examine the continuous outcome of the days until patients consult. This factor is crucial for evaluating administrative efficiency. In the REDAPS trial, patients were blinded while clinicians were unblinded to preserve their ability to cancel the default order. Since the nature of the intervention made blinding of clinicians infeasible and the outcome of interest was closely related to clinicians, there was a potential that clinicians might alter their behavior before treatment adoption time, thus introducing an anticipation effect. In fact, the possible presence of such an anticipation effect is further supported by the educational preparations before intervention initiation: palliative care staff received targeted in-service conferences and webinars 3 months before the intervention phase at each hospital; other clinical staff received notification about the study via established hospital communication channels (e.g., e-mails, posters, leadership meetings, and weekly communications) starting 2 weeks before the scheduled intervention and continuing throughout the intervention period.\cite{Courtright2016} Therefore, these early notifications raised clinician awareness of the upcoming intervention, leading to possible anticipatory behaviors that may or may not improve the administrative efficiency.

Estimation results are shown in Figure \ref{figure_data_analysis_REDAPS}. To test for exposure-time treatment effect heterogeneity in the REDAPS trial, we follow the procedures based on likelihood ratio test mentioned in Maleyeff et al.\cite{Maleyeff2022} Specifically, we use the likelihood ratio test of the treatment effect term in the ETI model with the null hypothesis $H_0: \delta(1) = \delta(2) = \cdots = \delta(J-1)$, which gives a p-value of 0.15. Similarly, we perform the likelihood ratio test in the ETI-ANT model and obtain a p-value of 0.36. Based on these results, we conclude that there is no strong evidence of exposure-time treatment effect heterogeneity in this trial. Therefore, our subsequent empirical analysis mainly focuses on results from models that assume a constant treatment effect (HH and HH-ANT).

To assess the presence of an anticipation effect, we check the estimated anticipation effects from the HH-ANT model and the ETI-ANT model, which are both negative and statistically significant. In particular, $\widehat\gamma^\HHANT = -0.59$ $(-1.01$ to $-0.17)$ under the HH-ANT model, which indicates a decrease in time until consultation before actual exposure to the treatment. Such a reduction suggests that clinicians who are aware of the upcoming implementation of default consultation orders may have proactively adjusted their clinical and administrative practices to enhance operational efficiency even before the intervention began. Moreover, the estimated treatment effects from the HH-ANT model and the ETI-ANT model are lower than those obtained from models that do not account for the anticipation effect. For instance, $\widehat\delta^\HH = -1.23$ $(-1.52$ to $-0.93)$ under the HH model and $\widehat\delta^\HHANT = -1.43$ $(-1.76$ to $-1.10)$ under the HH-ANT model. This suggests that neglecting the negative anticipation effect may result in an underestimation of the magnitude of the true treatment effect, which is aligned with results in Theorem \ref{theorem::HH_1}. As a side note, $|\widehat\gamma^\HHANT| < |\widehat\delta^\HHANT|$ under the HH-ANT model, which can be explained by initial gains in treatment benefit from early educational preparations prior to the intervention, followed by even greater operational improvements realized once the intervention was formally implemented. This empirical example illustrates the potential impact of the anticipation effect in the working model, especially in scenarios where blinding of providers who are delivering the intervention is partially or entirely infeasible.

\section{Concluding Remarks}\label{sec::conclusion}

Our study has two sets of primary findings. First, when neither exposure-time heterogeneity nor the anticipation effect is present, the HH model is the most efficient choice. However, using HH-ANT as a conservative approach when anticipation is uncertain causes variance inflation of approximately 10\%-90\% depending on the number of periods and ICC (see Figure \ref{figure_coeff_HH_vs_HH-ANT}), though it provides protection against bias if anticipation exists. When only the anticipation effect is present, the HH-ANT model provides consistent estimates for both the treatment effect and the anticipation effect. On the contrary, using the HH model in the presence of the anticipation effect leads to biased (often attenuated) estimates of the true treatment effect. Furthermore, when only exposure-time heterogeneity is present, both the HH model and the HH-ANT model lead to severely biased estimates of the treatment effect. In the mismatch scenario, the HH-ANT model may surprisingly give even more biased estimates of the treatment effect than the HH model. When only the exposure-time heterogeneity is expected but its functional form is unknown, the ETI model provides consistent estimates because it allows $\delta^\ETI(s)$ to be unique for each exposure time, accommodating any pattern of exposure-time dependent treatment effects. If researchers expect exposure-time heterogeneity and an anticipation effect, then the ETI-ANT model is recommended. To aid in the determination of the anticipation effect and exposure-time heterogeneity in SW-CRTs via a data-driven approach (often in secondary analyses), the ETI-ANT model can be fitted as a preliminary, exploratory step. Based on the general principles of comparing nested models, in the secondary analysis, one can simplify the model formulations by picking ETI, HH-ANT, or HH models, or consider analyses under these nested models as sensitivity results to assess the robustness of the primary findings. However, we caution that an SW-CRT may not be powered for such tests based on nested models, and hence, the testing results should be interpreted with caution. 

Second, we provide expressions for the variance of treatment effect estimators, which can be used in sample size calculations based on prior content knowledge regarding the anticipation effect and exposure-time heterogeneity during the trial planning stage. We have also numerically explored the power for testing treatment effects under the HH-ANT versus HH models, as well as under ETI-ANT versus ETI models, in the presence of an anticipation effect. The results show that, depending on the relative strength of the anticipation effect to the true treatment effect, one set of models may be more powerful. More specifically, the comparison depends on the relative reduction in effect size and variance and hence could be indeterminate in general. Regardless, just because the simpler HH or ETI models are more powerful than their anticipation model counterparts, the returned treatment effect estimates by the simpler models could be biased, as we show in Theorem \ref{theorem::HH_1} and \ref{theorem::ETI_2}. Such a biased treatment effect estimate may go against the call for clarity and transparency of estimands in the ICH E9(R1) Estimands Addendum for clinical trials,\cite{ICH_E9R1_2020, Kahan2024} and hence should be interpreted with caveats in the presence of anticipation. 

We wish to reiterate that the primary purpose of the current study is to formally characterize the risks of not accounting for the anticipation effect, time-varying treatment effects, or both in the context of an SW-CRT. It calls for greater awareness of this historically neglected anticipation effect in SW-CRTs. From a trial conduct and implementation perspective, the best approach to avoid anticipation effects in SW-CRTs is through treatment blinding and careful introduction of the intervention just prior to the randomized transition time point. However, when blinding is infeasible due to administrative constraints or the nature of the research question, it could be helpful to account for the anticipation effect in both the design and analysis of an SW-CRT, or at a minimum, consider addressing anticipation effects in sensitivity analyses. Finally, another potential approach to reduce the impact of anticipation effect is to consider one or more implementation periods, which lead to examples of a general class of incomplete designs.\cite{Hemming2015} For example, by omitting the data from the transition period or its adjacent prior period, one could have a ``clean'' subset for data analysis without the need to consider the anticipation effect. This approach has been implemented in many empirical studies, and including an implementation period has actually become one of the most common incomplete SW-CRTs in practice.\cite{Zhang2023}

Although we attempt to characterize the impact of model misspecification regarding the treatment effect structure under a wide array of settings in Table \ref{tab::true_vs_working_models}, the current study has several limitations. First, we assume the treatment effect depends on exposure time instead of calendar time. However, this assumption may be violated in clinical trials sensitive to seasonal fluctuations or experiencing significant external impacts. Lee et al.\cite{Lee2024} demonstrated that when the true underlying treatment effect varies with calendar time, the ETI model can give severely misleading estimates of the (calendar) time-averaged treatment effect. Additionally, Chen and Li\cite{Chen2024} studied independence generalized estimating equations estimators to quantify calendar-time varying treatment effect in SW-CRTs. Nonetheless, the role of anticipation in the presence of a calendar-time varying treatment effect remains unclear and deserves future studies. 

Second, although our simulation study and data application only include continuous outcomes, the models in \eqref{model::HH}-\eqref{model::ETI-ANT} can be generalized to other types of exponential family outcomes (e.g., binary or count outcomes) by selecting appropriate link and variance functions. However, unlike models with continuous outcomes discussed in this article, these extensions typically do not have closed-form estimating equations, making it challenging to analytically derive the expectations of the treatment effect estimator. Therefore, future simulation studies are needed to numerically explore and quantify the impact of anticipation under these more complex outcome scenarios. 

Third, although more complex correlation structures (such as nested exchangeable\cite{Hooper2016} and discrete time decay\cite{Kasza2019a}) have been proposed in the stepped wedge literature, for analytical tractability and clear interpretations, we assume a simple exchangeable correlation structure. This assumption allows us to derive closed-form expressions and characterize the bias when the working model fails to account for the anticipation effect, the exposure-time treatment effect heterogeneity, or both, assuming the correlation structure is correctly specified. Future research is needed to investigate the behaviors of linear mixed model estimators when both the assumed treatment effect structure and the correlation structure are misspecified in the presence of an anticipation effect.

Fourth, our simulation studies use a sinusoidal function for exposure-time treatment effects, while our theoretical illustrations use curved, lagged, and partially convex patterns. The sinusoidal function was chosen specifically because it represents a challenging scenario with both positive and negative deviations from the mean effect, providing a stringent test of our theoretical results. This diversity of functional forms is intentional, as our bias formulas in Section \ref{sec::est} are expressed in terms of PTEs $\delta(s)$ at each exposure time $s$, making them applicable to any functional form. We highlight the need for future research to explore how different functional forms of exposure-time dependent treatment effects interact with anticipation effects in SW-CRTs, particularly examining whether the model misspecification has more severe consequences for certain patterns of treatment effect heterogeneity in the presence of anticipation. 

Finally, we have only considered the first-order anticipation effect, whereas in practice, a higher-order anticipation effect may exist. In general principle, if the working model does not account for higher-order anticipation, the magnitude of bias will depend on the lag duration of these anticipation effects (see a discussion in Web Appendix E). This lag duration itself varies according to the specific trial context and the nature of the intervention. Thus, it may be of interest to conduct hypothesis tests on the lag duration and the anticipated-time heterogeneity for the anticipation effect (an analog to exposure-time heterogeneity for the treatment effect). We leave a detailed exploration of this topic to future research.

\section*{Acknowledgments}
Research in this article was supported by a Patient-Centered Outcomes Research Institute Award\textsuperscript{\textregistered} (PCORI\textsuperscript{\textregistered} Award ME-2022C2-27676). The statements presented are solely the responsibility of the authors and do not necessarily represent the official views of PCORI\textsuperscript{\textregistered}, its Board of Governors, or the Methodology Committee. The authors thank Vanessa Madden and Brian Bayes for help with accessing the de-identified REDAPS trial data. 

\section*{Data availability statement}
The REDAPS trial data used as an illustrative example in Section \ref{sec::da} were obtained from the REDAPS study investigators. Restrictions may apply to the availability of these data. Data requests should be submitted to Dr. Fan Li and Dr. Michael O. Harhay, who will correspond with the REDAPS investigators to obtain data permission.

\section*{Supporting information}
Additional supporting information, including Web Appendices A--F, Web Tables, and Figures, may be found online in the supporting information tab for this article. All R code for the simulation study and trial planning software is publicly available at \href{https://github.com/haowangbiostat/anticipation_swcrt}{https://github.com/haowangbiostat/anticipation\_swcrt}.

\bibliography{SWDanticipation}

@article{AronDine2015,
    author = {Aron-Dine, A and Einav, L and Finkelstein, A and Cullen, M},
    title = {Moral hazard in health insurance: do dynamic incentives matter?},
    journal = {The Review of Economics and Statistics},
    volume = {97},
    number = {4},
    pages = {725--741},
    year = {2015},
    month = {10},
    issn = {0034-6535},
}

@Article{Ariel2021,
    author={Ariel, B and Sutherland, A and Bland, M},
    title={The trick does not work if you have already seen the gorilla: how anticipatory effects contaminate pre-treatment measures in field experiments},
    journal={Journal of Experimental Criminology},
    year={2021},
    month={Mar},
    day={01},
    volume={17},
    number={1},
    pages={55--66},
    issn={1572-8315},
}

@article{Courtright2016,
    author = {Courtright, KR and Madden, V and Gabler, NB and Cooney, E and Small, DS and Troxel, A and Casarett, D and Ersek, M and Cassel, JB and Nicholas, LH and Escobar, G and Hill, SH and O’Brien, D and Vogel, M and Halpern, SD},
    title = {Rationale and design of the randomized evaluation of default access to palliative services (REDAPS) trial},
    journal = {Annals of the American Thoracic Society},
    volume = {13},
    number = {9},
    pages = {1629--1639},
    year = {2016},
}

@article{Chen2024,
    author = {Chen, X and Li, F},
    title = {Model-assisted analysis of covariance estimators for stepped wedge cluster randomized experiments},
    journal = {Scandinavian Journal of Statistics},
    volume = {52},
    number = {1},
    pages = {416--446},
    year = {2025}
}

@article{Courtright2024,
    author = {Courtright, KR and Madden, V and Bayes, B and Chowdhury, M and Whitman, C and Small, DS and Harhay, MO and Parra, S and Cooney-Zingman, E and Ersek, M and Escobar, GJ and Hill, SH and Halpern, SD},
    title = {Default palliative care consultation for seriously ill hospitalized patients: a pragmatic cluster randomized trial},
    journal = {JAMA},
    volume = {331},
    number = {3},
    pages = {224--232},
    year = {2024},
    month = {01},
    issn = {0098-7484},
}

@article{Copas2015,
    author = {Copas, AJ and Lewis, JJ and Thompson, JA and Davey, C and Baio, G and Hargreaves, JR},
    journal = {Trials},
    number = {1},
    publisher = {Trials},
    title = {{Designing a stepped wedge trial: three main designs, carry-over effects and randomisation approaches}},
    volume = {16},
    year = {2015}
}

@article{ICH_E9R1_2020,
    author = {Medicines Agency European},
    title = {ICH E9 (R1) addendum on estimands and sensitivity analysis in clinical trials to the guideline on statistical principles for clinical trials},
    year = {2020},
    url = {https://www.ema.europa.eu/en/ich-e9-statistical-principles-clinical-trials-scientific-guideline},
    journal = {https://www.ema.europa.eu/en/ich-e9-statistical-principles-clinical-trials-scientific-guideline},
}

@article{Granston2014,
    title = {Addressing lagged effects and interval censoring in the stepped wedge design of cluster randomized clinical trials},
    journal = {(Ph.D Thesis) University of Washington},
    year = {2014},
    url = {http://hdl.handle.net/1773/25120},
    author = {TS Granston},
}

@article{Hussey2007,
    title = {Design and analysis of stepped wedge cluster randomized trials},
    journal = {Contemporary Clinical Trials},
    volume = {28},
    number = {2},
    pages = {182--191},
    year = {2007},
    issn = {1551-7144},
    author = {Hussey, MA and Hughes, JP}
}

@article{Hemming2015,
    author = {Hemming, K and Lilford, R and Girling, AJ},
    title = {Stepped-wedge cluster randomised controlled trials: a generic framework including parallel and multiple-level designs},
    journal = {Statistics in Medicine},
    volume = {34},
    number = {2},
    pages = {181--196},
    year = {2015}
}

@article{Hughes2015,
    title = {Current issues in the design and analysis of stepped wedge trials},
    journal = {Contemporary Clinical Trials},
    volume = {45},
    pages = {55--60},
    year = {2015},
    note = {10th Anniversary Special Issue},
    issn = {1551-7144},
    author = {JP Hughes and TS Granston and PJ Heagerty}
}

@article{Hooper2016,
    author = {Hooper, Richard and Teerenstra, Steven and de Hoop, Esther and Eldridge, Sandra},
    title = {Sample size calculation for stepped wedge and other longitudinal cluster randomised trials},
    journal = {Statistics in Medicine},
    volume = {35},
    number = {26},
    pages = {4718--4728},
    year = {2016}
}

@article{Hemming2018,
    author = {Hemming, K and Taljaard, M and McKenzie, JE and Hooper, R and Copas, A and Thompson, JA and Dixon-Woods, M and Aldcroft, A and Doussau, A and Grayling, M and Kristunas, C and Goldstein, CE and Campbell, M K and Girling, A and Eldridge, S and Campbell, MJ and Lilford, RJ and Weijer, C and Forbes, AB and Grimshaw, JM},
    title = {Reporting of stepped wedge cluster randomised trials: extension of the CONSORT 2010 statement with explanation and elaboration},
    volume = {363},
    elocation-id = {k1614},
    year = {2018},
    publisher = {BMJ Publishing Group Ltd},
    issn = {0959-8138},
    journal = {BMJ}
}

@article{Hemming2020,
    author = {Hemming, K and Taljaard, M},
    title = {Reflection on modern methods: when is a stepped-wedge cluster randomized trial a good study design choice?},
    journal = {International Journal of Epidemiology},
    volume = {49},
    number = {3},
    pages = {1043--1052},
    year = {2020},
    month = {05},
    issn = {0300-5771},
}

@article{Li2020,
    author = {F Li and JP Hughes and K Hemming and M Taljaard and ER Melnick and PJ Heagerty},
    title ={Mixed-effects models for the design and analysis of stepped wedge cluster randomized trials: an overview},
    journal = {Statistical Methods in Medical Research},
    volume = {30},
    number = {2},
    pages = {612--639},
    year = {2021},
}

@article{Lee2024,
    title={Analysis of Stepped-Wedge Cluster Randomized Trials When Treatment Effects Vary by Exposure Time or Calendar Time},
    author={Lee, Kenneth M and Turner, Elizabeth L and Kenny, Avi},
    journal={Statistics in Medicine},
    volume={44},
    number={20-22},
    pages={e70256},
    year={2025},
    publisher={Wiley Online Library}
}

@article{Knutson2008,
    author = {Knutson, B and Greer, SM},
    title = {Anticipatory affect: neural correlates and consequences for choice},
    journal = {Philosophical Transactions of the Royal Society B: Biological Sciences},
    volume = {363},
    number = {1511},
    pages = {3771--3786},
    year = {2008},
}

@article{Kasza2019a,
    author = {J Kasza and K Hemming and R Hooper and JNS Matthews and AB Forbes},
    title ={Impact of non-uniform correlation structure on sample size and power in multiple-period cluster randomised trials},
    journal = {Statistical Methods in Medical Research},
    volume = {28},
    number = {3},
    pages = {703--716},
    year = {2019},
}

@article{Kenny2022,
    author = {Kenny, A and Voldal, EC and Xia, F and Heagerty, PJ and Hughes, JP},
    title = {Analysis of stepped wedge cluster randomized trials in the presence of a time-varying treatment effect},
    journal = {Statistics in Medicine},
    volume = {41},
    number = {22},
    pages = {4311--4339},
    year = {2022}
}

@article{Kahan2024,
    author = {Kahan, BC and Hindley, J and Edwards, M and Cro, S and Morris, TP},
    title = {The estimands framework: a primer on the ICH E9(R1) addendum},
    volume = {384},
    elocation-id = {e076316},
    year = {2024},
    doi = {10.1136/bmj-2023-076316},
    publisher = {BMJ Publishing Group Ltd},
    journal = {BMJ}
}

@article{Malani2015,
    title = {Interpreting pre-trends as anticipation: impact on estimated treatment effects from tort reform},
    journal = {Journal of Public Economics},
    volume = {124},
    pages = {1--17},
    year = {2015},
    issn = {0047-2727},
    author = {A Malani and J Reif},
}

@article{Matthews2017,
    author = {Matthews, JNS and Forbes, AB},
    title = {Stepped wedge designs: insights from a design of experiments perspective},
    journal = {Statistics in Medicine},
    volume = {36},
    number = {24},
    pages = {3772--3790},
    year = {2017}
}

@article{Meyer2019,
    title = {Expectancy value interactions and academic achievement: differential relationships with achievement measures},
    journal = {Contemporary Educational Psychology},
    volume = {58},
    pages = {58--74},
    year = {2019},
    issn = {0361-476X},
    author = {J Meyer and J Fleckenstein and O Köller},
}

@article{Maleyeff2022,
    author = {Maleyeff, L and Li, F and Haneuse, S and Wang, R},
    title = {Assessing exposure-time treatment effect heterogeneity in stepped-wedge cluster randomized trials},
    journal = {Biometrics},
    volume = {79},
    number = {3},
    pages = {2551--2564},
    year = {2022},
    month = {11},
    issn = {0006-341X},
}

@article{Nickless2018,
    doi = {10.1371/journal.pone.0208876},
    author = {Nickless, A and Voysey, M and Geddes, J and Yu, L and Fanshawe, TR},
    journal = {PLoS One},
    publisher = {Public Library of Science},
    title = {Mixed effects approach to the analysis of the stepped wedge cluster randomised trial--investigating the confounding effect of time through simulation},
    year = {2018},
    month = {12},
    volume = {13},
    pages = {1--22},
    number = {12},
}

@article{Nevins2024,
    author = {Pascale Nevins and Mary Ryan and Kendra Davis-Plourde and Yongdong Ouyang and Jules Antoine Pereira Macedo and Can Meng and Guangyu Tong and Xueqi Wang and Luis Ortiz-Reyes and Agnès Caille and Fan Li and Monica Taljaard},
    title ={Adherence to key recommendations for design and analysis of stepped-wedge cluster randomized trials: A review of trials published 2016–2022},
    journal = {Clinical Trials},
    volume = {21},
    number = {2},
    pages = {199--210},
    year = {2024},
}

@article{Pfingsten2001,
    author = {Pfingsten, M and Leibing, E and Harter, W and Kröner-Herwig, B and Hempel, D and Kronshage, U and Hildebrandt, J},
    title = {Fear-avoidance behavior and anticipation of pain in patients with chronic low back pain: a randomized controlled study},
    journal = {Pain Medicine},
    volume = {2},
    number = {4},
    pages = {259--266},
    year = {2001},
    month = {12},
    issn = {1526-2375},
}

@article{Rohde1965,
    ISSN = {03684245},
    URL = {http://www.jstor.org/stable/2946422},
    author = {CA Rohde},
    journal = {Journal of the Society for Industrial and Applied Mathematics},
    number = {4},
    pages = {1033--1035},
    publisher = {Society for Industrial and Applied Mathematics},
    title = {Generalized Inverses of Partitioned Matrices},
    volume = {13},
    year = {1965}
}

@article{Taljaard2017,
    author = {M Taljaard and K Hemming and L Shah and B Giraudeau and JM Grimshaw and C Weijer},
    title ={Inadequacy of ethical conduct and reporting of stepped wedge cluster randomized trials: results from a systematic review},
    journal = {Clinical Trials},
    volume = {14},
    number = {4},
    pages = {333--341},
    year = {2017},
}

@article{Wang2024,
    author = {Wang, B and Wang, X and Li, F},
    title = {How to achieve model-robust inference in stepped wedge trials with model-based methods?},
    journal = {Biometrics},
    volume = {80},
    number = {4},
    pages = {ujae123},
    year = {2024},
    month = {11},
    issn = {0006-341X},
}

@article{Zhang2023,
    author = {Y Zhang and JS Preisser and EL Turner and PJ Rathouz and M Toles and F Li},
    title ={A general method for calculating power for GEE analysis of complete and incomplete stepped wedge cluster randomized trials},
    journal = {Statistical Methods in Medical Research},
    volume = {32},
    number = {1},
    pages = {71--87},
    year = {2023},
}

\end{document}


\begin{center}
{\bf \LARGE Web Appendices for \\[5pt]
``On anticipation effect in stepped wedge cluster randomized trials''}\\
\vspace{0.5cm}

{\large Hao Wang ~ Xinyuan Chen ~ Katherine R. Courtright ~ Scott D. Halpern\\ 
Michael O. Harhay ~ Monica Taljaard ~and~ Fan Li$^*$}\\
{*fan.f.li@yale.edu}

\end{center}




\section*{Web Appendix A. Behaviors of Linear Mixed Model Estimators Under Model Misspecification} \label{sec::theorems}

Recall the following linear mixed models: the HH model,
\begin{align}
    Y_{ijk}(q) = \mu + \beta_j + \delta Z_{ij} + \alpha_i + \epsilon_{ijk}; \label{model::HH}
\end{align}
the HH-ANT model,
\begin{align}
    Y_{ijk}(q) = \mu + \beta_j + \gamma A_{ij} + \delta Z_{ij} + \alpha_i + \epsilon_{ijk}; \label{model::HH-ANT}
\end{align}
the ETI model,
\begin{align}
    Y_{ijk}(q) = \mu + \beta_j + \delta(s_{ij}) Z_{ij} + \alpha_i + \epsilon_{ijk}; \label{model::ETI}
\end{align}
and the ETI-ANT model,
\begin{align}
    Y_{ijk}(q) = \mu + \beta_j + \gamma A_{ij} + \delta(s_{ij}) Z_{ij} + \alpha_i + \epsilon_{ijk}. \label{model::ETI-ANT}
\end{align}
Of note, the index $q$ is auxiliary because each cluster $i$ only comes from one sequence. We add the index $q$ to facilitate future derivations related to each sequence of clusters.

In this section, we consider a standard cross-sectional SW-CRT, where $I$ is divisible by $Q$ with $Q = J-1$, and sequence $q$ has $I_q = I/Q$ clusters for $q \in \{1,\ldots,Q\}$. The treatment adoption time for clusters in treatment sequence $q$ is in period $j=q+1$ with the number of individuals per cluster in each period $K_{ij}=K$ for all $i$ and $j$. To facilitate the derivation, we re-parameterize \eqref{model::HH-ANT} and \eqref{model::ETI} by dropping $\mu$, the grand mean, and then rewrite them in terms of cluster-period means, so that they become
\begin{align}
    \overline Y_{ij}(q) = \beta_j + \gamma A_{ij} + \delta Z_{ij} + \alpha_i + \overline\epsilon_{ij}, \label{model::HH-ANT_re}
\end{align}
for the HH-ANT model,
\begin{align}
    \overline Y_{ij}(q) = \beta_j + \delta(s_{ij}) Z_{ij} + \alpha_i + \overline\epsilon_{ij}, \label{model::ETI_re}
\end{align}
for the ETI model. Of note, \eqref{model::HH-ANT_re} and \eqref{model::ETI_re} require no restrictions on $\beta_j$, $\overline\epsilon_{ij} =  K^{-1}\sumk \epsilon_{ijk}$, $\overline Y_{ij}(q) = K^{-1}\sumk Y_{ijk}(q)$, $\bbV\{\overline Y_{ij}(q)\} = \tau^2 + \sigma^2/K$, and $\text{cov}\{\overline Y_{ij}(q), \overline Y_{ij'}(q)\}=\tau^2$ for $j\neq j'$. Denote $\overline \bfY_i(q) = (\overline Y_{i1}(q), \ldots, \overline Y_{iJ}(q))'$ as a $J$-dimensional vector, with $\overline Y_{ij}(q)$ being the mean outcome of cluster $i$ from sequence $q$ in period $j$. Similarly, let  $\overline \bfY(q) = Q/I\sum_{n=1}^{I/Q} \overline \bfY_{n+(q-1)I/Q}(q) = (\overline Y_1(q), \ldots, \overline Y_J(q))'$ be $J$-dimensional vector, with $\overline Y_j(q)$ being the mean outcome from sequence $q$ in period $j$. Here, without loss of generality, we assume sequence $q$ is consist of clusters from $i = 1+(q-1)I/Q$ to $i = qI/Q$ for $q\in\{1,\ldots,Q\}$.

For \eqref{model::HH-ANT_re} and \eqref{model::ETI_re}, within each treatment sequence, all clusters share the same design matrix. To distinguish the version of any previously defined cluster-level scalar, vector, or matrix that aligns with this common design matrix within each sequence, we attach a superscript ${}^\seq$. For instance, $Z_{ij}$ denotes a treatment effect indicator for cluster $i$ in period $j$, whereas $Z_{qj}^\seq$ denotes the treatment status indicator for any cluster from sequence $q$ in period $j$. For \eqref{model::HH-ANT_re}, let $\bfE_q = (\bfI_J, \bfA_q^\seq, \bfZ_q^\seq)$ be the $J \times (J+2)$ design matrix for a cluster assigned to sequence $q$. Here, $\bfI_J$ represents a $J \times J$ identity matrix, $\bfA_q^\seq = (A_{q1}^\seq, A_{q2}^\seq, \ldots, A_{qJ}^\seq)'$ is the anticipation indicator vector for any cluster in sequence $q$, and $\bfZ_q^\seq =(Z_{q1}^\seq, \ldots, Z_{qJ}^\seq)'$ is the treatment indicator vector for any cluster in sequence $q$. For \eqref{model::ETI_re}, let $\bfF_q = (\bfI_J, \bfX_q^\seq)$ be the $J \times (2J-1)$ design matrix for a cluster assigned to sequence $q$. Here, $\bfI_J$ represents a $J \times J$ identity matrix and
\begin{align*}
    \bfX_q^\seq = \left(\begin{array}{c}
         \mathbf{0}_{(j^*-1)\times (J-1)} \\
         \bfI_{J-j^*+1}, \mathbf{0}_{(J-j^*+1)\times (j^*-2)}
    \end{array}\right),
\end{align*}
where $j^*$ is the treatment adoption time with $j^*\geq 2$ (by the definition of SW-CRTs), $\mathbf{0}_{(j^*-1)\times (J-1)}$ is a $(j^*-1)\times (J-1)$ zero matrix, and $\bfI_{J-j^*+1}$ is a $(J-j^*+1)\times (J-j^*+1)$ identity matrix. Of note, $\bfX_q^\seq$ depends on $Z_{ij}$ in \eqref{model::ETI_re} through $j^*$, where $j^*\in\{1,\ldots,J\}$ be the period index with $Z_{ij^*}=1$ and $Z_{i,j^*-1}=0$ for cluster $i$.

In addition, let $\bfV$ be the $J \times J$ covariance matrix for a single cluster, and we assume $\bfV$ is identical for all clusters (by definition of an exchangeable correlation structure). For convenience, Matthews and Forbes\cite{Matthews2017} suggest $\bfV = \eta^2\{(1-\phi)\bfI_J+\phi \mathbf{1}_J\mathbf{1}_J'\}$ with $\phi = \tau^2 / (\tau^2 + \sigma^2 / K)$, $\eta^2 = \tau^2/\phi$, and $\mathbf{1}_J\mathbf{1}_J'$ being a $J \times J$ matrix of all ones. Note that $\bfV^{-1} = x\bfI_J-y\mathbf{1}_J\mathbf{1}_J'$, where 
\begin{align*}
    x = \frac{1}{\eta^2(1-\phi)} \quad\text{and}\quad y = \frac{\phi}{\eta^2(1-\phi)(1+\phi J-\phi)}.
\end{align*}

\subsection*{A.1 Working Model: HH} \label{sec::HH}

We consider the scenario where the working model is HH, which does not account for the anticipation effect or exposure-time heterogeneity, and examine the behaviors of the HH model estimator when the true model is either HH-ANT or ETI-ANT.

\subsubsection*{A.1.1 Closed-form Expressions of the HH Estimator} \label{sec::HH_est}

For the HH working model in \eqref{model::HH}, its estimator $\widehat\delta^\HH$ can be expressed as:\cite{Kenny2022}
\begin{align}
    \widehat\delta^\HH = \frac{12(1 + \phi Q)}{Q(Q + 1)(\phi Q^2 + 2Q - \phi Q - 2)} \sumq \sumj \left\{Q \bbone(j > q) - j + 1 + \frac{\phi Q(2q - Q - 1)}{2(1 + \phi Q)} \right\} \overline Y_j(q). \label{est::HH}
\end{align}
For ease of notation, let $T_1 = 12(1 + \phi Q)/\{Q(Q + 1)(\phi Q^2 + 2Q - \phi Q - 2)\}$. We proceed to use \eqref{est::HH} to derive its expectation $\bbE(\widehat\delta^\HH)$ when the true model is either HH-ANT in \eqref{model::HH-ANT} (Theorem \ref{theorem::HH_1} from Web Appendix A.1.2) or ETI-ANT in \eqref{model::ETI-ANT} (Theorem \ref{theorem::HH_2} from Web Appendix A.1.3).

\subsubsection*{A.1.2 True Model: HH-ANT} \label{sec::HH_under_HH-ANT}

\begin{theorem}\label{theorem::HH_1}
    \emph{For a standard SW-CRT with data generated from the HH-ANT model, the expectation of the HH model estimator $\widehat\delta^\HH$ is
    \begin{align*}
        \bbE(\widehat\delta^\HH) = \delta^\HHANT+\omega_\HH^\HHANT\gamma^\HHANT ,
    \end{align*}
    where $\delta^\HHANT$ is the true treatment effect and $\gamma^\HHANT$ is the true anticipation effect. The weight of $\gamma^\HHANT$ is given by
    \begin{align*}
        \omega_\HH^\HHANT =-\frac{6(1+\phi Q)}{(Q+1)(2+\phi Q)}.
    \end{align*}}
\end{theorem}

\begin{proof}

    When the true model is HH-ANT, the expectation of the HH model estimator $\widehat\delta^\HH$ is
    \begin{align*}
        \bbE(\widehat\delta^\HH) &= T_1 \sumq \sumj \left\{ Q \bbone(j > q) - j + 1 + \frac{\phi Q(2q - Q - 1)}{2(1+\phi Q)} \right\} \bbE\{\overline Y_j(q)\} \displaybreak[0]\\
        &= T_1 \sumq \sumj \left\{ Q \bbone(j > q) - j + 1 + \frac{\phi Q(2q - Q - 1)}{2(1 + \phi Q)} \right\} \left\{\mu + \beta_j + \bbone(j = q) \gamma^\HHANT + \bbone(j > q) \delta^\HHANT \right\} 
        \displaybreak[0]\\
        &= T_1 \sumq \sumj \left\{ Q \bbone(j > q) - j + 1 + \frac{\phi Q(2q - Q - 1)}{2(1 + \phi Q)} \right\} \left\{\bbone(j = q) \gamma^\HHANT + \bbone(j > q) \delta^\HHANT \right\}. 
    \end{align*}
    We first simplify terms related to $\gamma^\ETIANT$,
    \begin{align*}
        &~T_1 \sumq \sumj \left\{ Q \bbone(j > q) - j + 1 + \frac{\phi Q(2q - Q - 1)}{2(1 + \phi Q)} \right\}\bbone(j = q) \gamma^\HHANT \displaybreak[0]\\
        =&~T_1 \sumq \sumj \left\{- j + 1 + \frac{\phi Q(2q - Q - 1)}{2(1 + \phi Q)} \right\} \bbone(j = q) \gamma^\HHANT 
        \displaybreak[0]\\
        =&~T_1 \left\{ -\frac{Q(Q + 1)}{2} + Q + \frac{\phi Q}{1 + \phi Q}\frac{Q^2 + Q}{2} -\frac{\phi Q^2(Q + 1)}{2(1 + \phi Q)}\right\} \gamma^\HHANT 
        \displaybreak[0]\\
        =&~T_1 \frac{Q(1 - Q)}{2} \gamma^\HHANT 
        \displaybreak[0]\\
        =&~-\frac{6(1 + \phi Q)}{(Q + 1)(2 + \phi Q)} \gamma^\HHANT. 
    \end{align*}
    We then simplify terms related to $\delta^\ETIANT$,
    \begin{align*}
        &~T_1 \sumq \sumj \left\{ Q\bbone(j > q) - j + 1 + \frac{\phi Q(2q - Q - 1)}{2(1 + \phi Q)} \right\} \bbone(j > q) \delta^\HHANT \displaybreak[0]\\
        =&~T_1 \sumq \sumj \left\{ Q\bbone(j > q) - j + 1 + \frac{\phi Q q}{1 + \phi Q} - \frac{\phi Q(Q + 1)}{2(1 + \phi Q)} \right\} \bbone(j > q) \delta^\HHANT 
        \displaybreak[0]\\
        =&~T_1 \left\{\frac{Q^2(Q + 1)}{2} - \frac{Q(Q^2 + 3Q + 2)}{3} +  \frac{Q(Q + 1)}{2} + \frac{Q(\phi Q)(Q^2 + 3Q + 2)}{6(1 + \phi Q)} -  \frac{\phi Q(Q + 1)}{2(1 + \phi Q)} \frac{Q(Q + 1)}{2} \right\} \delta^\HHANT 
        \displaybreak[0]\\
        =&~T_1 \frac{Q(\phi Q^3 + 2Q^2 - \phi Q - 2)}{12(1 + \phi Q)} \delta^\HHANT 
        \displaybreak[0]\\
        =&~\delta^\HHANT. 
    \end{align*}
    Putting everything together,
    \begin{align*}
        \bbE(\widehat\delta^\HH) = \delta^\HHANT+\omega_\HH^\HHANT\gamma^\HHANT,
    \end{align*}
    where
    \begin{align*}
        \omega_\HH^\HHANT =-\frac{6(1+\phi Q)}{(Q+1)(2+\phi Q)}.
    \end{align*}
    
\end{proof}

\subsubsection*{A.1.3 True Model: ETI-ANT} \label{sec::HH_under_ETI-ANT}

\begin{theorem}\label{theorem::HH_2}
    \emph{For a standard SW-CRT with data generated from the ETI-ANT model, the expectation of the HH estimator $\widehat\delta^\HH$ can be expressed as a weighted average of the point treatment effects and the anticipation effect:
    \begin{align*}
        \bbE(\widehat\delta^\HH) = \sumjp \left\{\pi_\HH^\ETIANT(j) \delta^\ETIANT(j) + \omega_\HH^\ETIANT(j) \gamma^\ETIANT \right\},
    \end{align*}
    where $\delta^\ETIANT(j)$ is the true point treatment effect with exposure time $j \in \{1, \ldots, J-1\}$ and $\gamma^\ETIANT$ is the true anticipation effect. The weights of $\delta^\ETIANT(j)$ and $\gamma^\ETIANT$ are given by
    \begin{align*}
        \pi_\HH^\ETIANT(j)=\frac{6(j - Q - 1)\{(1 + 2\phi Q)j - (1 + \phi + \phi Q)Q\}}{Q(Q + 1)(\phi Q^2 + 2Q - \phi Q - 2)},
    \end{align*}
    and
    \begin{align*}
        \omega_\HH^\ETIANT(j)=- \frac{6(\phi Q^2 - \phi Q + 2j - 2)}{Q(\phi Q^3 + 2Q^2 - \phi Q - 2)}.
    \end{align*}}
\end{theorem}

\begin{proof}

    When the true model is ETI-ANT, the expectation of the HH model estimator $\widehat\delta^\HH$ is
    \begin{align*}
        \bbE(\widehat\delta^\HH) &= T_1 \sumq \sumj \left\{ Q\bbone(j > q) - j + 1 + \frac{\phi Q(2q - Q - 1)}{2(1 + \phi Q)} \right\} \bbE\{\overline Y_j(q)\} \displaybreak[0]\\
        &= T_1 \sumq \sumj \left\{ Q\bbone(j > q) - j + 1 + \frac{\phi Q(2q - Q - 1)}{2(1 + \phi Q)} \right\} \left\{\mu + \beta_j + \bbone(j = q) \gamma^\ETIANT + \bbone(j > q) \delta^\ETIANT(j-q) \right\} 
        \displaybreak[0]\\
        &= T_1 \sumq \sumj \left\{ Q\bbone(j > q) - j + 1 + \frac{\phi Q(2q - Q - 1)}{2(1 + \phi Q)} \right\} \left\{\bbone(j = q) \gamma^\ETIANT + \bbone(j > q) \delta^\ETIANT(j-q) \right\}. 
    \end{align*}
    We first simplify terms related to $\gamma^\ETIANT$,
    \begin{align*}
        &~T_1 \sumq \sumj \left\{ Q\bbone(j > q) - j + 1 + \frac{\phi Q(2q - Q - 1)}{2(1 + \phi Q)} \right\} \bbone(j = q) \gamma^\ETIANT \displaybreak[0]\\
        =&~T_1 \sumq \sumj \left\{- j + 1 + \frac{\phi Q(2q - Q - 1)}{2(1+\phi Q)} \right\}\bbone(j = q) \gamma^\ETIANT 
        \displaybreak[0]\\
        =&~T_1 \sumq \left\{-q + 1 +  \frac{\phi Q(2q - Q - 1)}{2(1 + \phi Q)} \right\} \gamma^\ETIANT 
        \displaybreak[0]\\
        =&~T_1 \sumjp \left\{-j + 1+\frac{\phi Q(2j - Q - 1)}{2(1+\phi Q)} \right\} \gamma^\ETIANT 
        \displaybreak[0]\\
        =&~\sumjp \left\{-\frac{12j(1 + \phi Q)}{Q(Q + 1)(\phi Q^2 + 2Q - \phi Q - 2)} + \frac{6(-\phi Q^2 + 2\phi Q j + \phi Q + 2)}{Q(Q + 1)(\phi Q^2 + 2Q - \phi Q - 2)} \right\}\gamma^\ETIANT 
        \displaybreak[0]\\
        =&~-\sumjp \frac{6(\phi Q^2 - \phi Q - 2 + 2j)}{Q(\phi Q^3 + 2Q^2 - \phi Q - 2)} \gamma^\ETIANT.
    \end{align*}
    For terms related to $\delta^\ETIANT(j-q)$, Kenney et al.\cite{Kenny2022} showed that
    \begin{align*}
        &~T_1 \sumq \sumj \left\{ Q\bbone(j > q) - j + 1 + \frac{\phi Q(2q - Q - 1)}{2(1 + \phi Q)} \right\} \bbone(j > q) \delta^\ETIANT(j-q) \displaybreak[0]\\
        =&~\sumjp \frac{6(j - Q - 1)\{(1 + 2\phi Q)j - (1 + \phi + \phi Q)Q\}}{Q(Q + 1)(\phi Q^2 + 2Q - \phi Q - 2)} \delta^\ETIANT(j).
    \end{align*}
    Putting everything together,
    \begin{align*}
        \bbE(\widehat\delta^\HH) = \sumjp \left\{\pi_\HH^\ETIANT(j) \delta^\ETIANT(j) + \omega_\HH^\ETIANT(j) \gamma^\ETIANT \right\},
    \end{align*}
    where
    \begin{align*}
        \pi_\HH^\ETIANT(j)=\frac{6(j - Q - 1)\{(1 + 2\phi Q)j - (1 + \phi + \phi Q)Q\}}{Q(Q + 1)(\phi Q^2 + 2Q - \phi Q - 2)},
    \end{align*}
    and
    \begin{align*}
        \omega_\HH^\ETIANT(j)=- \frac{6(\phi Q^2 - \phi Q + 2j - 2)}{Q(\phi Q^3 + 2Q^2 - \phi Q - 2)}.
    \end{align*}
    
\end{proof}

\subsection*{A.2  Working Model: ETI} \label{sec::ETI}

We consider the scenario where the working model is ETI, which accounts for the exposure-time heterogeneity but not the anticipation effect, and examine the behaviors of the ETI model estimators when the true model is either HH-ANT or the ETI-ANT. For the simplicity of exposition, we discuss via an example of a standard SW-CRT with $J=3$ and $Q=2$, and the results can be easily extended to more general settings. 

\subsubsection*{A.2.1 Closed-form Expressions of the ETI Estimators} \label{sec::ETI_est}

Based on standard results in mixed model theory, the weighted least squares estimator of the fixed effects $(\beta_1, \beta_2, \beta_3, \delta^\ETI(1), \delta^\ETI(2))'$ can be expressed as:\cite{Matthews2017}
\begin{align}
    \left(\widehat\beta_1, \widehat\beta_2, \widehat\beta_3, \widehat\delta^\ETI(1), \widehat\delta^\ETI(2)\right)' = \left(\sum_{q=1}^{Q=2} \bfF_q'\bfV^{-1}\bfF_q\right)^{-1}\sum_{q=1}^{Q=2} \bfF_q'\bfV^{-1}\overline \bfY(q). \label{est::ETI}
\end{align}
We proceed to derive the closed-form expressions of $\widehat\delta^\ETI(1)$ and $\widehat\delta^\ETI(2)$ in \eqref{est::ETI}.

The closed-form expressions of $\left(\sum_{q=1}^{Q=2} \bfF_q'\bfV^{-1}\bfF_q\right)^{-1}$ and $\sum_{q=1}^{Q=2} \bfF_q'\bfV^{-1}\overline \bfY(q)$ are the following
\begin{align*}
    \sum_{q=1}^{Q=2} \bfF_q'\bfV^{-1}\bfF_q &= \left(\begin{array}{ccc|cc}
         2x-2y & -2y & -2y & -2y & -y\\
         -2y & 2x-2y & -2y & x-2y & -y\\
         -2y & -2y & 2x-2y & x-2y & x-y\\
         \hline
         -2y & x-2y & x-2y & 2x-2y & -y\\
         -y & -y & x-y & -y & x-y
    \end{array}\right), \displaybreak[0]\\
    \left(\sum_{q=1}^{Q=2} \bfF_q'\bfV^{-1}\bfF_q \right)^{-1} &= \left(\begin{array}{ccc|cc}
            \frac{x/2- y}{x (x - 3 y)} & \frac{y}{2 x (x - 3 y)} & \frac{y}{2 x (x - 3 y)} & 0 & 0\\[5pt]
            \frac{y}{2 x (x - 3 y)} & \frac{x^2 - 4 x y + 7y^2/2}{x (x^2 - 5 x y + 6 y^2)} & \frac{x^2 - 3 x y + y^2/2}{x (x^2 - 5 x y + 6 y^2)} & \frac{- x + y}{x (x - 2 y)} & - \frac{1}{x - 2 y}\\[5pt]
            \frac{y}{2 x (x - 3 y)} & \frac{x^2 - 3 x y + y^2/2}{x (x^2 - 5 x y + 6 y^2)} & \frac{6 x^2 - 20 x y + 7 y^2}{2 x (x^2 - 5 x y + 6 y^2)} & \frac{- 2 x + y}{x (x - 2 y)} & - \frac{3}{x - 2 y}\\[5pt]
            \hline
            \rule[0pt]{0pt}{\heightof{0}+5pt}0 & \frac{- x + y}{x (x - 2 y)} & \frac{- 2 x + y}{x (x - 2 y)} & \frac{2 (x - y)}{x (x - 2 y)} & \frac{2}{x - 2 y}\\[5pt]
            0 & - \frac{1}{x - 2 y} & - \frac{3}{x - 2 y} & \frac{2}{x - 2 y} & \frac{4}{x - 2 y}
    \end{array}\right),
\end{align*}
and
\begin{align*}
    \sum_{q=1}^{Q=2} \bfF_q'\bfV^{-1}\overline \bfY(q) = \left(\begin{array}{c}
        x\sum_{q=1}^{Q=2}\overline Y_1(q) - y\sum_{q=1}^{Q=2}\sum_{j=1}^{J=3} \overline Y_j(q)\\[5pt]
        x\sum_{q=1}^{Q=2}\overline Y_2(q) - y\sum_{q=1}^{Q=2}\sum_{j=1}^{J=3} \overline Y_j(q)\\[5pt]
        x\sum_{q=1}^{Q=2}\overline Y_3(q) - y\sum_{q=1}^{Q=2}\sum_{j=1}^{J=3} \overline Y_j(q)\\[5pt]
        x\sum_{q=1}^{Q=2}\sum_{j=1}^{J=3} \bbone(j=q+1)\overline Y_j(q)-y\sum_{q=1}^{Q=2}\sum_{j=1}^{J=3}\overline Y_j(q)\\[5pt]
        x\sum_{q=1}^{Q=2}\sum_{j=1}^{J=3} \bbone(j=q+2)\overline Y_j(q)-y\sum_{q=1}^{Q=2}\sum_{j=1}^{J=3} (Q-q) \overline Y_j(q)
    \end{array}\right).
\end{align*}
Putting everything together,
\begin{align*}
    \widehat\delta^\ETI(1) =&~ \left(0, \frac{- x + y}{x (x - 2 y)}, \frac{- 2 x + y}{x (x - 2 y)}, \frac{2 (x - y)}{x (x - 2 y)}, \frac{2}{x - 2 y}\right)\sum_{q=1}^{Q=2} \bfF_q'\bfV^{-1}\overline \bfY(1) 
    \displaybreak[0]\\
    =&~\left(\frac{- x + y}{x (x - 2 y)}, \frac{- 2 x + y}{x (x - 2 y)}, \frac{2 (x - y)}{x (x - 2 y)}, \frac{2}{x - 2 y}\right)\left(\begin{array}{c}
        x\sum_{q=1}^{Q=2}\overline Y_2(q)-y\sum_{q=1}^{Q=2}\sum_{j=1}^{J=3} \overline Y_j(q)\\[5pt]
        x\sum_{q=1}^{Q=2}\overline Y_3(q)-y\sum_{q=1}^{Q=2}\sum_{j=1}^{J=3} \overline Y_j(q)\\[5pt]
        x\sum_{q=1}^{Q=2}\sum_{j=1}^{J=3} \bbone(j=q+1)\overline Y_j(q)-y\sum_{q=1}^{Q=2}\sum_{j=1}^{J=3}\overline Y_j(q)\\[5pt]
        x\sum_{q=1}^{Q=2}\sum_{j=1}^{J=3} \bbone(j=q+2)\overline Y_j(q)-y\sum_{q=1}^{Q=2}\sum_{j=1}^{J=3} (Q-q) \overline Y_j(q)
    \end{array}\right) 
    \displaybreak[0]\\
    =&~\frac{-x+y}{x-2y}\sum_{q=1}^{Q=2}\overline Y_2(q)+\frac{- 2 x + y}{x - 2 y}\sum_{q=1}^{Q=2}\overline Y_3(q)+\frac{2 (x - y)}{x - 2 y}\sum_{q=1}^{Q=2}\sum_{j=1}^{J=3} \bbone(j=q+1)\overline Y_j(q) \displaybreak[0]\\
    &~\quad+\frac{2x}{x - 2 y}\sum_{q=1}^{Q=2}\sum_{j=1}^{J=3} \bbone(j=q+2)\overline Y_j(q)-\left\{\frac{- x + y}{x (x - 2 y)} + \frac{- 2 x + y}{x (x - 2 y)} + \frac{2 (x - y)}{x (x - 2 y)}\right\}y\sum_{q=1}^{Q=2}\sum_{j=1}^{J=3} \overline Y_j(q)\displaybreak[0]\\
    &~\quad-\frac{2y}{x - 2 y}\sum_{q=1}^{Q=2}\sum_{j=1}^{J=3} (Q-q) \overline Y_j(q) 
    \displaybreak[0]\\
    =&~\frac{1}{x-2y}\sum_{q=1}^{Q=2}\sum_{j=1}^{J=3} \left\{-(j-1)x+y\bbone(j>1) + 2 (x - y)\bbone(j=q+1)+2x\bbone(j=q+2)+y-2y(Q-q)\right\} \overline Y_j(q) 
    \displaybreak[0]\\
    =&~\sum_{q=1}^{Q=2}\sum_{j=1}^{J=3} \left\{-(j-1)(1+2\phi)+\phi\bbone(j>1)+ 2 (1+\phi)\bbone(j=q+1)+2(1+2\phi)\bbone(j=q+2)+\phi(2q-3)\right\} \overline Y_j(q), 
\end{align*}
and
\begin{align*}
    \widehat\delta^\ETI(2) =&~ \left(0,- \frac{1}{x - 2 y},- \frac{3}{x - 2 y}, \frac{2}{x - 2 y}, \frac{4}{x - 2 y}\right)\sum_{q=1}^{Q=2} \bfF_q'\bfV^{-1}\overline \bfY(1) 
    \displaybreak[0]\\
        =&~\left(- \frac{1}{x - 2 y}, -\frac{3}{x - 2 y}, \frac{2}{x - 2 y}, \frac{4}{x - 2 y}\right)\left(\begin{array}{c}
        x\sum_{q=1}^{Q=2}\overline Y_2(q)-y\sum_{q=1}^{Q=2}\sum_{j=1}^{J=3} \overline Y_j(q)\\[5pt]
        x\sum_{q=1}^{Q=2}\overline Y_3(q)-y\sum_{q=1}^{Q=2}\sum_{j=1}^{J=3} \overline Y_j(q)\\[5pt]
        x\sum_{q=1}^{Q=2}\sum_{j=1}^{J=3} \bbone(j=q+1)\overline Y_j(q)-y\sum_{q=1}^{Q=2}\sum_{j=1}^{J=3}\overline Y_j(q)\\[5pt]
        x\sum_{q=1}^{Q=2}\sum_{j=1}^{J=3} \bbone(j=q+2)\overline Y_j(q)-y\sum_{q=1}^{Q=2}\sum_{j=1}^{J=3} (Q-q) \overline Y_j(q)
    \end{array}\right) 
    \displaybreak[0]\\
    =&~\frac{-x}{x-2y}\sum_{q=1}^{Q=2}\overline Y_2(q)+\frac{- 3 x}{x - 2 y}\sum_{q=1}^{Q=2}\overline Y_3(q)+\frac{2 x}{x - 2 y}\sum_{q=1}^{Q=2}\sum_{j=1}^{J=3} \bbone(j=q+1)\overline Y_j(q)+\frac{4x}{x - 2 y}\sum_{q=1}^{Q=2}\sum_{j=1}^{J=3} \bbone(j=q+2)\overline Y_j(q) \displaybreak[0]\\
    &~\quad-\left(- \frac{1}{x - 2 y} - \frac{3}{x - 2 y} + \frac{2}{x - 2 y}\right)y\sum_{q=1}^{Q=2}\sum_{j=1}^{J=3} \overline Y_j(q)-\frac{4y}{x - 2 y}\sum_{q=1}^{Q=2}\sum_{j=1}^{J=3} (Q-q) \overline Y_j(q) 
    \displaybreak[0]\\
    =&~\frac{1}{x-2y}\sum_{q=1}^{Q=2}\sum_{j=1}^{J=3} \left\{-(j-1)x-x\bbone(j=J) + 2 x\bbone(j=q+1)+4x\bbone(j=q+2)+2y-4y(Q-q)\right\}\overline Y_j(q) 
    \displaybreak[0]\\
    =&~\sum_{q=1}^{Q=2}\sum_{j=1}^{J=3} \left\{-(j-1)(1+2\phi)-(1+2\phi)\bbone(j=J) + 2(1+2\phi)\bbone(j=q+1)\right.\displaybreak[0]\\
    &~\qquad\qquad\left.+4(1+2\phi)\bbone(j=q+2)+2\phi(2q-3)\right\}\overline Y_j(q). 
\end{align*}
We proceed to use these closed-form expressions of $\widehat\delta^\ETI(1)$ and $\widehat\delta^\ETI(2)$ to derive their expectation $\bbE\{\widehat\delta^\ETI(1)\}$ and $\bbE\{\widehat\delta^\ETI(2)\}$ when the true model is either HH-ANT in \eqref{model::HH-ANT} (Theorem \ref{theorem::ETI_1} from Web Appendix A.2.2) or ETI-ANT in \eqref{model::ETI-ANT} (Theorem \ref{theorem::ETI_2} from Web Appendix A.2.3).

\subsubsection*{A.2.2 True Model: HH-ANT} \label{sec::ETI_under_HH-ANT}

\begin{theorem}\label{theorem::ETI_1}
    \emph{For a standard SW-CRT with data generated from the HH-ANT model, the expectations of the ETI model estimators $\widehat\delta^\ETI(1)$ and $\widehat\delta^\ETI(2)$ are
    \begin{align*}
        \bbE\{\widehat\delta^\ETI(1)\} = \delta^\HHANT-(1+\phi)\gamma^\HHANT \quad\text{and}\quad\bbE\{\widehat\delta^\ETI(2)\} = \delta^\HHANT-(1+2\phi)\gamma^\HHANT,
    \end{align*}
    where $\delta^\HHANT$ is the true treatment effect and $\gamma^\HHANT$ is the true anticipation effect.}
\end{theorem}

\begin{proof}

    When the true model is HH-ANT, the expectation of the ETI model estimator $\widehat\delta^\ETI(1)$ is
    \begin{align*}
        \bbE\{\widehat\delta^\ETI(1)\} =&~ \sum_{q=1}^{Q=2}\sum_{j=1}^{J=3} \left\{-(j-1)(1+2\phi)+\phi\bbone(j>1)+ 2 (1+\phi)\bbone(j=q+1)+2(1+2\phi)\bbone(j=q+2)+\phi(2q-3)\right\} \bbE\{\overline Y_j(q)\} \displaybreak[0]\\
        =&~ \sum_{q=1}^{Q=2}\sum_{j=1}^{J=3} \left\{-(j-1)(1+2\phi)+\phi\bbone(j>1)+ 2 (1+\phi)\bbone(j=q+1)+2(1+2\phi)\bbone(j=q+2)+\phi(2q-3)\right\} \displaybreak[0]\\
        &~\quad \times \left\{\mu + \beta_j + \bbone(j = q) \gamma^\HHANT + \bbone(j > q) \delta^\HHANT \right\}. 
    \end{align*}
    For terms related to $\mu$,
    \begin{align*}
        &~\mu\sum_{q=1}^{Q=2}\sum_{j=1}^{J=3} \left\{-(j-1)(1+2\phi)+\phi\bbone(j>1)+ 2 (1+\phi)\bbone(j=q+1)\right.\displaybreak[0]\\
        &\qquad\qquad\quad\left.+2(1+2\phi)\bbone(j=q+2)+\phi(2q-3)\right\} \displaybreak[0]\\
        =&~-6\mu(1+2\phi)+4\mu\phi+4\mu(1+\phi)+2\mu(1+2\phi) 
        \displaybreak[0]\\
        =&~0. 
    \end{align*}
    For terms related to $\beta_j$,
    \begin{align*}
        &~\sum_{q=1}^{Q=2}\sum_{j=1}^{J=3} \left\{-(j-1)(1+2\phi)+\phi\bbone(j>1)+ 2 (1+\phi)\bbone(j=q+1)+2(1+2\phi)\bbone(j=q+2)+\phi(2q-3)\right\}\beta_j \displaybreak[0]\\
        =&~-2(1+2\phi)(\beta_2 + 2\beta_3) + 2\phi(\beta_2 + \beta_3) + 2(1+\phi)(\beta_2+\beta_3)+2(1+2\phi)\beta_3 
        \displaybreak[0]\\
        =&~0. 
    \end{align*}
    For terms related to $\gamma^\HHANT$,
    \begin{align*}
        &~\sum_{q=1}^{Q=2}\sum_{j=1}^{J=3} \left\{-(j-1)(1+2\phi)+\phi\bbone(j>1)+ 2 (1+\phi)\bbone(j=q+1)\right.\displaybreak[0]\\
        &\qquad\qquad\left.+2(1+2\phi)\bbone(j=q+2)+\phi(2q-3)\right\}\bbone(j = q) \gamma^\HHANT \displaybreak[0]\\
        =&~-(1+2\phi)\gamma^\HHANT+\phi\gamma^\HHANT 
        \displaybreak[0]\\
        =&~-(1+\phi)\gamma^\HHANT. 
    \end{align*}
    For terms related to $\delta^\HHANT$,
    \begin{align*}
        &~\sum_{q=1}^{Q=2}\sum_{j=1}^{J=3} \left\{-(j-1)(1+2\phi)+\phi\bbone(j>1)+ 2 (1+\phi)\bbone(j=q+1)\right.\displaybreak[0]\\
        &\qquad\qquad\left.+2(1+2\phi)\bbone(j=q+2)+\phi(2q-3)\right\}\bbone(j > q) \delta^\HHANT \displaybreak[0]\\
        =&~-5(1+2\phi)\delta^\HHANT+3\phi\delta^\HHANT+4(1+\phi)\delta^\HHANT+2(1+2\phi)\delta-\phi\delta^\HHANT 
        \displaybreak[0]\\
        =&~\delta^\HHANT. 
    \end{align*}
    Similarly, the expectation of the ETI model estimator $\widehat\delta^\ETI(2)$ is
    \begin{align*}
        &~\bbE\{\widehat\delta^\ETI(2)\} \displaybreak[0]\\
        =&~\sum_{q=1}^{Q=2}\sum_{j=1}^{J=3} \left\{-(j-1)(1+2\phi)-(1+2\phi)\bbone(j=J) + 2(1+2\phi)\bbone(j=q+1)+4(1+2\phi)\bbone(j=q+2)+2\phi(2q-3)\right\} \bbE\{\overline Y_j(q)\} \displaybreak[0]\\
        =&~\sum_{q=1}^{Q=2}\sum_{j=1}^{J=3} \left\{-(j-1)(1+2\phi)-(1+2\phi)\bbone(j=J) + 2(1+2\phi)\bbone(j=q+1)+4(1+2\phi)\bbone(j=q+2)+2\phi(2q-3)\right\} \displaybreak[0]\\
        &~\quad \times \left\{\mu + \beta_j + \bbone(j = q) \gamma^\HHANT + \bbone(j > q) \delta^\HHANT \right\}. 
    \end{align*}
    For terms related to $\mu$,
    \begin{align*}
        &~\mu\sum_{q=1}^{Q=2}\sum_{j=1}^{J=3} \left\{-(j-1)(1+2\phi)-(1+2\phi)\bbone(j=J) + 2(1+2\phi)\bbone(j=q+1)\right.\displaybreak[0]\\
        &\qquad\qquad\quad\left.+4(1+2\phi)\bbone(j=q+2)+2\phi(2q-3)\right\} \displaybreak[0]\\
        =&~-6\mu(1+2\phi)-2\mu(1+2\phi)+4\mu(1+2\phi)+4\mu(1+2\phi) 
        \displaybreak[0]\\
        =&~0. 
    \end{align*}
    For terms related to $\beta_j$,
    \begin{align*}
        &~\sum_{q=1}^{Q=2}\sum_{j=1}^{J=3} \left\{-(j-1)(1+2\phi)-(1+2\phi)\bbone(j=J) + 2(1+2\phi)\bbone(j=q+1)\right.\displaybreak[0]\\
        &\qquad\qquad\left.+4(1+2\phi)\bbone(j=q+2)+2\phi(2q-3)\right\} \beta_j \displaybreak[0]\\
        =&~-2(1+2\phi)(\beta_2 + 2\beta_3) - 2(1+2\phi)\beta_3+2(1+2\phi)(\beta_2+\beta_3)+4(1+2\phi)\beta_3 
        \displaybreak[0]\\
        =&~0. 
    \end{align*}
    For terms related to $\gamma^\HHANT$,
    \begin{align*}
        &~\sum_{q=1}^{Q=2}\sum_{j=1}^{J=3} \left\{-(j-1)(1+2\phi)-(1+2\phi)\bbone(j=J) + 2(1+2\phi)\bbone(j=q+1)\right.\displaybreak[0]\\
        &~\qquad\qquad\left.+4(1+2\phi)\bbone(j=q+2)+2\phi(2q-3)\right\}\bbone(j = q) \gamma^\HHANT \displaybreak[0]\\
        =&~-(1+2\phi)\gamma^\HHANT. 
    \end{align*}
    For terms related to $\delta^\HHANT$,
    \begin{align*}
        &~\sum_{q=1}^{Q=2}\sum_{j=1}^{J=3} \left\{-(j-1)(1+2\phi)-(1+2\phi)\bbone(j=J) + 2(1+2\phi)\bbone(j=q+1)\right.\displaybreak[0]\\
        &~\qquad\qquad\left.+4(1+2\phi)\bbone(j=q+2)+2\phi(2q-3)\right\}\bbone(j > q) \delta^\HHANT \displaybreak[0]\\
        =&~-5(1+2\phi)\delta^\HHANT-2(1+2\phi)\delta^\HHANT+4(1+2\phi)\delta^\HHANT+4(1+2\phi)\delta^\HHANT-2\phi\delta^\HHANT 
        \displaybreak[0]\\
        =&~\delta^\HHANT. 
    \end{align*}
    Putting everything together,
    \begin{align*}
        \bbE\{\widehat\delta^\ETI(1)\} = \delta^\HHANT-(1+\phi)\gamma^\HHANT \quad\text{and}\quad\bbE\{\widehat\delta^\ETI(2)\} = \delta^\HHANT-(1+2\phi)\gamma^\HHANT.
    \end{align*}
    
\end{proof}

\subsubsection*{A.2.3 True Model: ETI-ANT} \label{sec::ETI_under_ETI-ANT}

\begin{theorem}\label{theorem::ETI_2}
    \emph{For a standard SW-CRT with data generated from the ETI-ANT model, the expectations of the ETI model estimators $\widehat\delta^\ETI(1)$ and $\widehat\delta^\ETI(2)$ are
    \begin{align}
        \bbE\{\widehat\delta^\ETI(1)\} = \delta^\ETIANT(1)-(1+\phi)\gamma^\ETIANT \quad\text{and}\quad\bbE\{\widehat\delta^\ETI(2)\} = \delta^\ETIANT(2)-(1+2\phi)\gamma^\ETIANT, \label{theorem::ETI_2_eqn}
    \end{align}
    where $\delta^\ETIANT(j)$ is the true point treatment effect with exposure time $j \in \{1, 2\}$ and $\gamma^\ETIANT$ is the true anticipation effect.}
\end{theorem}

\begin{proof}

    When the true model is ETI-ANT, the expectation of the ETI model estimator $\widehat\delta^\ETI(1)$ is
    \begin{align*}
        \bbE\{\widehat\delta^\ETI(1)\} =&~ \sum_{q=1}^{Q=2}\sum_{j=1}^{J=3} \left\{-(j-1)(1+2\phi)+\phi\bbone(j>1)+ 2 (1+\phi)\bbone(j=q+1)+2(1+2\phi)\bbone(j=q+2)+\phi(2q-3)\right\} \bbE\{\overline Y_j(q)\} \displaybreak[0]\\
        =&~ \sum_{q=1}^{Q=2}\sum_{j=1}^{J=3} \left\{-(j-1)(1+2\phi)+\phi\bbone(j>1)+ 2 (1+\phi)\bbone(j=q+1)+2(1+2\phi)\bbone(j=q+2)+\phi(2q-3)\right\} \displaybreak[0]\\
        &~\quad \times \left\{\mu + \beta_j + \bbone(j = q) \gamma^\ETIANT + \bbone(j = q + 1) \delta^\ETIANT(1) \right\}. 
    \end{align*}
    Terms related to $\mu$, $\beta_j$, and $\gamma^\ETIANT$ are simplified in A.2.2, and we focus on the only remaining terms related to $\delta^\ETIANT(1)$, which is
    \begin{align*}
        &~\quad\sum_{q=1}^{Q=2}\sum_{j=1}^{J=3} \left\{-(j-1)(1+2\phi)+\phi\bbone(j>1)+ 2 (1+\phi)\bbone(j=q+1)+2(1+2\phi)\bbone(j=q+2)+\phi(2q-3)\right\}\bbone(j = q + 1) \delta^\ETIANT(1) \displaybreak[0]\\
        =&~-(1+2\phi)\sum_{q=1}^{Q=2}\sum_{j=1}^{J=3}(j-1)\bbone(j = q + 1) \delta^\ETIANT(1) +\phi\sum_{q=1}^{Q=2}\sum_{j=1}^{J=3}\bbone(j>1)\bbone(j = q + 1) \delta^\ETIANT(1) \displaybreak[0]\\
        &~\quad +2(1+\phi)\sum_{q=1}^{Q=2}\sum_{j=1}^{J=3}\bbone(j=q+1)\delta^\ETIANT(1) 
        \displaybreak[0]\\
        =&~-3(1+2\phi)\delta^\ETIANT(1)+2\phi\delta^\ETIANT(1)+4(1+\phi)\delta^\ETIANT(1) 
        \displaybreak[0]\\
        =&~\delta^\ETIANT(1). 
    \end{align*}
    Similarly, the expectation of the ETI model estimator $\widehat\delta^\ETI(2)$ is
    \begin{align*}
        &~\bbE\{\widehat\delta^\ETI(2)\} \displaybreak[0]\\
        =&~\sum_{q=1}^{Q=2}\sum_{j=1}^{J=3} \left\{-(j-1)(1+2\phi)-(1+2\phi)\bbone(j=J) + 2(1+2\phi)\bbone(j=q+1)+4(1+2\phi)\bbone(j=q+2)+2\phi(2q-3)\right\} \bbE\{\overline Y_j(q)\} \displaybreak[0]\\
        =&~ \sum_{q=1}^{Q=2}\sum_{j=1}^{J=3} \left\{-(j-1)(1+2\phi)-(1+2\phi)\bbone(j=J) + 2(1+2\phi)\bbone(j=q+1)+4(1+2\phi)\bbone(j=q+2)+2\phi(2q-3)\right\} \displaybreak[0]\\
        &~\quad\times \left\{\mu + \beta_j + \bbone(j = q) \gamma^\ETIANT + \bbone(j = q + 2) \delta^\ETIANT(2) \right\}. 
    \end{align*}
    Terms related to $\mu$, $\beta_j$, and $\gamma^\ETIANT$ are simplified in A.2.2, and we focus on the only remaining terms related to $\delta(2)$, which is
    \begin{align*}
        &~\sum_{q=1}^{Q=2}\sum_{j=1}^{J=3} \left\{-(j-1)(1+2\phi)-(1+2\phi)\bbone(j=J) + 2(1+2\phi)\bbone(j=q+1)\right.\displaybreak[0]\\
        &\qquad\qquad\left.+4(1+2\phi)\bbone(j=q+2)+2\phi(2q-3)\right\}\bbone(j = q + 2) \delta^\ETIANT(2) \displaybreak[0]\\
        =&~-(1+2\phi)\sum_{q=1}^{Q=2}\sum_{j=1}^{J=3}(j-1)\bbone(j = q + 2) \delta^\ETIANT(2) -(1+2\phi)\sum_{q=1}^{Q=2}\sum_{j=1}^{J=3}\bbone(j=J)\bbone(j = q + 2) \delta^\ETIANT(2) \displaybreak[0]\\
        &\quad+4(1+2\phi)\sum_{q=1}^{Q=2}\sum_{j=1}^{J=3}\bbone(j=q+2)\delta^\ETIANT(2)+2\phi\sum_{q=1}^{Q=2}\sum_{j=1}^{J=3}(2q-3)\bbone(j = q + 2) \delta^\ETIANT(2)
        \displaybreak[0]\\
        =&~-2(1+2\phi)\delta^\ETIANT(2)-(1+2\phi)\delta^\ETIANT(2)+4(1+2\phi)\delta^\ETIANT(2)-2\phi\delta^\ETIANT(2)
        \displaybreak[0]\\
        =&~\delta^\ETIANT(2).
    \end{align*}
    Putting everything together,
    \begin{align*}
        \bbE\{\widehat\delta^\ETI(1)\} = \delta^\ETIANT(1)-(1+\phi)\gamma^\ETIANT \quad\text{and}\quad\bbE\{\widehat\delta^\ETI(2)\} = \delta^\ETIANT(2)-(1+2\phi)\gamma^\ETIANT.
    \end{align*}
    
\end{proof}

\subsection*{A.3  Working Model: HH-ANT} \label{sec::HH-ANT}

We consider the scenario where the working model is HH-ANT, which accounts for the anticipation effect but not the exposure-time heterogeneity, and examine the behaviors of the HH-ANT model estimators when the true model is either ETI or ETI-ANT. 

\subsubsection*{A.3.1 Closed-form Expressions of the HH-ANT Estimators} \label{sec::HH-ANT_est}

Based on standard results in mixed model theory, the weighted least squares estimator of the fixed effects $(\widehat\beta_1, \ldots, \widehat\beta_J, \widehat\gamma^\HHANT, \widehat\delta^\HHANT)'$ can be expressed as:\cite{Matthews2017}
\begin{align}
    \left(\widehat\beta_1, \ldots, \widehat\beta_J, \widehat\gamma^\HHANT, \widehat\delta^\HHANT\right)' = \left(\sumq \bfE_q'\bfV^{-1}\bfE_q\right)^{-1}\sumq \bfE_q'\bfV^{-1}\overline \bfY(q). \label{est::HH-ANT}
\end{align}
We derive the closed-form expressions of $\widehat\gamma^\HHANT$ and $\widehat\delta^\HHANT$ based on \eqref{est::HH-ANT}.

Consider a row-column design $\bfZ = (\bfZ_1^\seq, \ldots, \bfZ_Q^\seq)'$ with $Q$ rows and $J$ columns. Let $\bfkappa_1$ be a $J$-dimensional vector giving the number of times treatment $1$ appears in each column of $\bfZ$, and $\kappa_2$ be the number of times treatment $1$ appears in the whole design $\bfZ$. Then, 
\begin{align*}
    \sumq \bfE_q'\bfV^{-1}\bfE_q = \left(\begin{array}{cc|c}
         Qx\bfI_J-Qy\mathbf{1}_J\mathbf{1}_J' & (\bfx_Q', 0)'-Q\bfy_J & x\bfkappa_1-y\kappa_2 \mathbf{1}_J \\[5pt]
        (\bfx_Q', 0)-Q\bfy_J' & Qx-Qy & -y\kappa_2\\[5pt]
        \hline
        \rule[0pt]{0pt}{\heightof{C}+5pt}x\bfkappa_1'-y\kappa_2 \mathbf{1}_J' & -y\kappa_2 & x\kappa_2- y\frac{Q(Q + 1)(2Q + 1)}{6}
    \end{array}\right),
\end{align*}
where $\bfx_Q$ is a $Q$-dimensional vector of $x$, and $\bfy_J$ is a $J$-dimensional vector of $y$. Using the formula of the generalized inverse of a partitioned matrix in Rohde, \cite{Rohde1965} we proceed to derive the inverse of the $(J+1) \times (J+1)$ upper left block of $\sumq \bfE_q'\bfV^{-1}\bfE_q$. To start with,
\begin{align*}
    \left(\begin{array}{cc}
         Qx\bfI_J-Qy\mathbf{1}_J\mathbf{1}_J' & (\bfx_Q', 0)'-Q\bfy_J \\
        (\bfx_Q', 0)-Q\bfy_J' & Qx-Qy
    \end{array}\right)^{-1} = \left(\begin{array}{cc}
         \bfPsi_{11} & \bfPsi_{12}\\
         \bfPsi_{12}' & \bfPsi_{22}
    \end{array}\right)^{-1} = \left(\begin{array}{c|c}
         \bfPsi_{11}^{-1}+\bfPsi_{11}^{-1}\bfPsi_{12}\Psi^{-1}\bfPsi_{12}'\bfPsi_{11}^{-1} & -\bfPsi_{11}^{-1}\bfPsi_{12}\Psi^{-1}\\[5pt]
         \hline
        \rule[0pt]{0pt}{\heightof{C}+5pt}-\Psi^{-1}\bfPsi_{12}'\bfPsi_{11}^{-1} & \Psi^{-1}
    \end{array}\right),
\end{align*}
where
\begin{align*}
    &\Psi = \bfPsi_{22}-\bfPsi_{12}'\bfPsi_{11}^{-1}\bfPsi_{12} = Qx-Qy- \{(\bfx_Q', 0)-Q\bfy_J'\} \bfPsi_{11}^{-1}\{(\bfx_Q', 0)'-Q\bfy_J\} = x(Q-1), \displaybreak[0]\\
    &\Psi^{-1}  = \frac{1}{x(Q-1)}, \displaybreak[0]\\
    &\bfPsi_{11}^{-1}+\bfPsi_{11}^{-1}\bfPsi_{12}\Psi^{-1}\bfPsi_{12}'\bfPsi_{11}^{-1} = \frac{x-Jy}{Qx(x-Jy)}\bfI_J+\frac{y}{Qx(x-Jy)}\mathbf{1}_J\mathbf{1}_J'+\frac{1}{Q^2(Q-1)x}(\mathbf{1}_Q', 0)'(\mathbf{1}_Q', 0), \displaybreak[0]\\
    &-\bfPsi_{11}^{-1}\bfPsi_{12}\Psi^{-1} = -\frac{1}{Q(Q-1)x}(\mathbf{1}_Q', 0)', \displaybreak[0]\\
    &-\Psi^{-1}\bfPsi_{12}'\bfPsi_{11}^{-1} = -\frac{1}{Q(Q-1)x}(\mathbf{1}_Q', 0).
\end{align*}
Thus,
\begin{align*}
    \left(\begin{array}{cc}
         Qx\bfI_J-Qy\mathbf{1}_J\mathbf{1}_J' & (\bfx_Q', 0)'-Q\bfy_J \\
        (\bfx_Q', 0)-Q\bfy_J' & Qx-Qy
    \end{array}\right)^{-1} = \left(\begin{array}{c|c}
         \frac{x-Jy}{Qx(x-Jy)}\bfI_J+\frac{y}{Qx(x-Jy)}\mathbf{1}_J\mathbf{1}_J'+\frac{1}{Q^2(Q-1)x}(\mathbf{1}_Q', 0)'(\mathbf{1}_Q', 0) &  -\frac{1}{Q(Q-1)x}(\mathbf{1}_Q', 0)'\\[5pt]
         \hline
         \rule[0pt]{0pt}{\heightof{C}+5pt}-\frac{1}{Q(Q-1)x}(\mathbf{1}_Q', 0) & \frac{1}{x(Q-1)}
    \end{array}\right).
\end{align*}
Then, we derive the inverse of $\sumq \bfE_q'\bfV^{-1}\bfE_q$ by applying the formula of the generalized inverse of a partitioned matrix in Rohde \cite{Rohde1965} again. Specifically,
\begin{align*}
     &~\left(\begin{array}{cc|c}
         Qx\bfI_J-Qy\mathbf{1}_J\mathbf{1}_J' & (\bfx_Q', 0)'-Q\bfy_J & x\bfkappa_1-y\kappa_2 \mathbf{1}_J \\
        (\bfx_Q', 0)-Q\bfy_J' & Qx-Qy & -y\kappa_2\\[5pt]
        \hline
        \rule[0pt]{0pt}{\heightof{C}+5pt}x\bfkappa_1'-y\kappa_2 \mathbf{1}_J' & -y\kappa_2 & x\kappa_2- y\frac{Q(Q + 1)(2Q + 1)}{6}
    \end{array}\right)^{-1} \displaybreak[0]\\
    =&~\left(\begin{array}{c|c}
         \bfOmega_{11} &  \bfOmega_{12}\\
         \hline
         \bfOmega_{12}' & \bfOmega_{22}
    \end{array}\right)^{-1} = \left(\begin{array}{c|c}
         \bfOmega_{11}^{-1}+\bfOmega_{11}^{-1}\bfOmega_{12}\bfOmega^{-1}\bfOmega_{12}'\bfOmega_{11}^{-1} & -\bfOmega_{11}^{-1}\bfOmega_{12}\bfOmega^{-1}\\[5pt]
         \hline
        \rule[0pt]{0pt}{\heightof{C}+5pt}-\bfOmega^{-1}\bfOmega_{12}'\bfOmega_{11}^{-1} & \bfOmega^{-1}
    \end{array}\right),
\end{align*}
where
\begin{align*}
    \bfOmega_{11}^{-1} =&~ \left(\begin{array}{c|c}
         \frac{x-Jy}{Qx(x-Jy)}\bfI_J+\frac{y}{Qx(x-Jy)}\mathbf{1}_J\mathbf{1}_J'+\frac{1}{Q^2(Q-1)x}(\mathbf{1}_Q', 0)'(\mathbf{1}_Q', 0) &  -\frac{1}{Q(Q-1)x}(\mathbf{1}_Q', 0)'\\[5pt]
         \hline
         \rule[0pt]{0pt}{\heightof{C}+5pt}-\frac{1}{Q(Q-1)x}(\mathbf{1}_Q', 0) & \frac{1}{x(Q-1)}
    \end{array}\right), \displaybreak[0]\\
    \bfOmega =&~ \bfOmega_{22}-\bfOmega_{12}'\bfOmega_{11}^{-1}\bfOmega_{12},\\
    =&~ \frac{Q(Q+1)x}{2}- y\frac{Q(Q + 1)(2Q + 1)}{6} \displaybreak[0]\\
    &~\quad -\left(-\frac{Q(Q+1)y}{2},x-\frac{Q(Q+1)y}{2},2x-\frac{Q(Q+1)y}{2},\ldots,Qx-\frac{Q(Q+1)y}{2},-\frac{Q(Q+1)y}{2}\right)\bfOmega_{11}^{-1} \displaybreak[0]\\
    &~\qquad\times\left(-\frac{Q(Q+1)y}{2},x-\frac{Q(Q+1)y}{2},2x-\frac{Q(Q+1)y}{2},\ldots,Qx-\frac{Q(Q+1)y}{2},-\frac{Q(Q+1)y}{2}\right)'
    \displaybreak[0]\\
    =&~ \frac{Q(Q+1)x}{2}- y\frac{Q(Q + 1)(2Q + 1)}{6}-\left(\frac{1}{2 Q},\frac{3}{2 Q},\frac{5}{2 Q},\ldots,\frac{2Q-1}{2 Q},1,-\frac{1}{2}\right) \displaybreak[0]\\
    &~\quad\times \left(-\frac{Q(Q+1)y}{2},x-\frac{Q(Q+1)y}{2},2x-\frac{Q(Q+1)y}{2},\ldots,Qx-\frac{Q(Q+1)y}{2},-\frac{Q(Q+1)y}{2}\right)'
    \displaybreak[0]\\
    =&~ \frac{Q(Q+1)x}{2}- y\frac{Q(Q + 1)(2Q + 1)}{6} -\left[Qx-\frac{Q(Q+1)y}{2} + \frac{Q(Q+1)y}{4} + \sum_{i=0}^{Q-1} \left\{ ix - \frac{Q(Q+1)y}{2}\right\}\frac{2i+1}{2Q}\right] 
    \displaybreak[0]\\
    =&~ \frac{Q(Q+1)x}{2}- y\frac{Q(Q + 1)(2Q + 1)}{6}-\frac{-3yQ(1 + Q)^2 + x(-1 + 9Q + 4Q^2)}{12} 
    \displaybreak[0]\\
    =&~\frac{x - 3xQ + yQ + 2xQ^2 - yQ^3}{12}, 
    \displaybreak[0]\\
    \bfOmega^{-1} =&~ \frac{12}{x - 3xQ + yQ + 2xQ^2 - yQ^3}, 
    \displaybreak[0]\\
    -\bfOmega^{-1}\bfOmega_{12}'\bfOmega_{11}^{-1} =&~ \frac{-12}{x - 3xQ + yQ + 2xQ^2 - yQ^3}\left(\frac{1}{2 Q},\frac{3}{2 Q},\frac{5}{2 Q},\ldots,\frac{2Q-1}{2 Q},1,-\frac{1}{2}\right).
\end{align*}
Recall that for $\sumq \bfE_q'\bfV^{-1}\overline \bfY(q)$,
\begin{align*}
    \sumq \bfE_q'\bfV^{-1}\overline \bfY(q) = \left(\begin{array}{c}
        x\sumq\overline Y_1(q)-y\sumq\sumj \overline Y_j(q)\\[5pt]
        \vdots\\[5pt]
        x\sumq\overline Y_J(q)-y\sumq\sumj \overline Y_j(q)\\[5pt]
        x\sumq\sumj \bbone(j=q)\overline Y_j(q)-y\sumq\sumj\overline Y_j(q)\\[5pt]
        x\sumq\sumj\bbone(j>q)\overline Y_j(q)-y\sumq\sumj (Q-q+1)\overline Y_j(q)
    \end{array}\right).
\end{align*}
Putting everything together, we first get the closed-form expressions of the treatment effect estimator $\widehat\delta^\HHANT$ as
\begin{align*}
    &~\widehat\delta^\HHANT \displaybreak[0]\\ 
    =&~\left(-\bfOmega^{-1}\bfOmega_{12}'\bfOmega_{11}^{-1}, \bfOmega^{-1}\right)\sumq \bfE_q'\bfV^{-1}\overline \bfY(q) 
    \displaybreak[0]\\
    =&~-\bfOmega^{-1}\left(\frac{1}{2 Q},\frac{3}{2 Q},\frac{5}{2 Q},\ldots,\frac{2Q-1}{2 Q},1,-\frac{1}{2},-1\right) \sumq \bfE_q'\bfV^{-1}\overline \bfY(q) \displaybreak[0]\\
    =&~-\bfOmega^{-1}\left(x\sumq\sumjp \frac{2j-1}{2Q}\overline Y_j(q)-\frac{Qy}{2}\sumq\sumj \overline Y_j(q)+x\sumq\overline Y_J(q)-y\sumq\sumj \overline Y_j(q) \right. \displaybreak[0]\\
    &~\qquad\qquad\left. -\frac{1}{2}x\sumq\sumj \bbone(j=q)\overline Y_j(q)+\frac{1}{2}y\sumq\sumj\overline Y_j(q)-x\sumq\sumj\bbone(j>q)\overline Y_j(q)+y\sumq\sumj (Q-q+1)\overline Y_j(q) \right) 
    \displaybreak[0]\\
    =&~-\bfOmega^{-1}\left[x\sumq\sumjp \frac{2j-1}{2Q}\overline Y_j(q)+x\sumq\overline Y_J(q)+\sumq\sumj \left\{-\frac{Qy}{2}-y-\frac{1}{2}x\bbone(j=q)+\frac{1}{2}y-x\bbone(j>q)+Qy-qy+y\right\}\overline Y_j(q)\right] 
    \displaybreak[0]\\
    =&~-\bfOmega^{-1}\sumq\sumj \left[\left(\frac{Q+1}{2}-q\right)y-\left\{\frac{1}{2}\bbone(j=q)+\bbone(j>q)+\frac{1}{2Q}\bbone(j=J)+\frac{1-2j}{2Q}\right\}x\right]\overline Y_j(q). 
    \displaybreak[0]\\
    =&~\frac{-12}{x - 3xQ + yQ + 2xQ^2 - yQ^3}\sumq\sumj \left[\left(\frac{Q+1}{2}-q\right)y-\left\{\frac{1}{2}\bbone(j=q)+\bbone(j>q)+\frac{1}{2Q}\bbone(j=J)+\frac{1-2j}{2Q}\right\}x\right]\overline Y_j(q) \displaybreak[0]\\
    =&~ \frac{12(\phi - 1)(\phi J - \phi + 1)}{\phi Q(1 - Q^2) + (2Q^2 - 3Q + 1)(\phi J - \phi + 1)} \displaybreak[0]\\
    &~\quad\times\sumq\sumj \left[\frac{\phi \left(Q - 2 q + 1\right)}{2 \left(1-\phi\right) \left(J \phi - \phi + 1\right)} - \left\{\frac{1}{2}\bbone(j=q)+\bbone(j>q)+\frac{1}{2Q}\bbone(j=J)+\frac{1-2j}{2Q}\right\}\frac{1}{1-\phi}\right]\overline Y_j(q)
    \displaybreak[0]\\
    =&~\frac{12(\phi - 1)(\phi J - \phi + 1)}{\phi Q(1 - Q^2) + (2Q^2 - 3Q + 1)(\phi J - \phi + 1)}\displaybreak[0]\\
    &~\quad\times\sumq\sumj \left[\frac{-\phi Q(Q - 2q + 1) + (J\phi - \phi + 1)\{Q\bbone(j=q) + 2Q\bbone(j>q) + \bbone(j=J) - 2j + 1\}}{2Q(\phi - 1)(J\phi - \phi + 1)}\right]\overline Y_j(q) 
    \displaybreak[0]\\
    =&~ \frac{6}{\phi Q(1 - Q^2) + (2Q^2 - 3Q + 1)(\phi J - \phi + 1)} \displaybreak[0]\\
    &~\quad\times\sumq\sumj \left[\frac{-\phi Q(Q - 2q + 1) + (\phi J - \phi + 1)\{Q\bbone(j=q) + 2Q\bbone(j>q) + \bbone(j=J) - 2j + 1)}{Q}\right]\overline Y_j(q) 
    \displaybreak[0]\\
    =&~ \frac{6}{\phi Q(1 - Q^2) + (2Q^2 - 3Q + 1)(\phi Q + 1)} \displaybreak[0]\\
    &~\quad\times \sumq\sumj \left[-\phi(Q - 2q + 1) + \frac{\phi Q + 1}{Q}\{Q\bbone(j=q) + 2Q\bbone(j>q) + \bbone(j=J) - 2j + 1\}\right]\overline Y_j(q). 
\end{align*}
Now, we derive the closed-form expression of the anticipation effect estimator $\widehat\gamma^\HHANT$. We are interested in the transpose of the $(J+1)$-th column of $\left(\sumq \bfE_q'\bfV^{-1}\bfE_q\right)^{-1}$. Specifically, we need to get the expressions for the $(J+1)$-th column of $\bfOmega_{11}^{-1}+(\bfOmega_{11}^{-1}\bfOmega_{12}\bfOmega^{-1})\bfOmega_{12}'\bfOmega_{11}^{-1}$, and the $(J+1)$-th element of $-\bfOmega^{-1}\bfOmega_{12}'\bfOmega_{11}^{-1}$. Note that $\left(\sumq \bfE_q'\bfV^{-1}\bfE_q\right)^{-1}$ is symmetric, then 
\begin{align*}
    &\bfOmega_{11}^{-1}\bfOmega_{12}\bfOmega^{-1} = \left(\bfOmega^{-1}\bfOmega_{12}'\bfOmega_{11}^{-1}\right)' =\frac{12}{x - 3xQ + yQ + 2xQ^2 - yQ^3}\left(\frac{1}{2 Q},\frac{3}{2 Q},\frac{5}{2 Q},\ldots,\frac{2Q-1}{2 Q},1,-\frac{1}{2}\right)', \displaybreak[0]\\
    &\bfOmega_{12}'\bfOmega_{11}^{-1} = \left(\frac{1}{2 Q},\frac{3}{2 Q},\frac{5}{2 Q},\ldots,\frac{2Q-1}{2 Q},1,-\frac{1}{2}\right), \displaybreak[0]\\
    &\bfOmega_{11}^{-1} = \left(\begin{array}{c|c}
         \frac{x-Jy}{Qx(x-Jy)}\bfI_J+\frac{y}{Qx(x-Jy)}\mathbf{1}_J\mathbf{1}_J'+\frac{1}{Q^2(Q-1)x}(\mathbf{1}_Q', 0)'(\mathbf{1}_Q', 0) &  -\frac{1}{Q(Q-1)x}(\mathbf{1}_Q', 0)'\\[5pt]
         \hline
         \rule[0pt]{0pt}{\heightof{C}+5pt}-\frac{1}{Q(Q-1)x}(\mathbf{1}_Q', 0) & \frac{1}{x(Q-1)}
    \end{array}\right).
\end{align*}
Based on expressions above, the $(J+1)$-th column of $\bfOmega_{11}^{-1}+(\bfOmega_{11}^{-1}\bfOmega_{12}\bfOmega^{-1})\bfOmega_{12}'\bfOmega_{11}^{-1}$ is simply
\begin{align*}
    \left(\begin{array}{cc}
         -\frac{1}{Q(Q-1)x}(\mathbf{1}_Q', 0)'\\[5pt]
         \frac{1}{x(Q-1)}
    \end{array}\right) -\frac{1}{2}\frac{12}{x - 3xQ + yQ + 2xQ^2 - yQ^3}\left(\frac{1}{2 Q},\frac{3}{2 Q},\frac{5}{2 Q},\ldots,\frac{2Q-1}{2 Q},1,-\frac{1}{2}\right)' = \left(\begin{array}{c}
         \frac{yQ^2 - 2xQ + yQ - 2x}{12xQ\bfOmega} \\[5pt]
         \frac{yQ^2 - 2xQ + yQ - 8x}{12xQ\bfOmega} \\[5pt]
         \frac{yQ^2 - 2xQ + yQ - 14x}{12xQ\bfOmega} \\[5pt]
         \vdots \\[5pt]
         \frac{yQ^2 - 8xQ + yQ + 4x}{12xQ\bfOmega} \\[5pt]
         -\frac{1}{2\bfOmega}\\[5pt]
         \frac{-yQ^2 + 2xQ - yQ + 2x}{12x\bfOmega}
    \end{array}\right),
\end{align*}
and the $(J+1)$-th element of $-\bfOmega^{-1}\bfOmega_{12}'\bfOmega_{11}^{-1}$ is $2^{-1}\bfOmega^{-1}$. Thus, using the symmetry of $\left(\sumq \bfE_q'\bfV^{-1}\bfE_q\right)^{-1}$ again, the $(J+1)$-th row of $\left(\sumq \bfE_q'\bfV^{-1}\bfE_q\right)^{-1}$ is
\begin{align*}
    \bfnu_\gamma =&~ \left(\frac{yQ^2 - 2xQ + yQ - 2x}{12xQ\bfOmega},\frac{yQ^2 - 2xQ + yQ - 8x}{12xQ\bfOmega}, \frac{yQ^2 - 2xQ + yQ - 14x}{12xQ\bfOmega}, \ldots, \right. \displaybreak[0]\\
    &~\qquad\left. \frac{yQ^2 - 8xQ + yQ + 4x}{12xQ\bfOmega}, -\frac{1}{2\bfOmega}, \frac{-yQ^2 + 2xQ - yQ + 2x}{12x\bfOmega}, \frac{1}{2\bfOmega}\right).
\end{align*}
Putting everything together, we finally get the closed-form expressions of the anticipation effect estimator $\widehat\gamma^\HHANT$ as
\begin{align*}
    \widehat\gamma^\HHANT =&~\bfnu_\gamma\sumq \bfE_q'\bfV^{-1}\overline \bfY(q) 
    \displaybreak[0]\\
     =&~ \frac{1}{12\bfOmega} \left\{x\sumq\sumjp \frac{Q^2 y - 2Q x + Q y + (4-6j)x}{Q x }\overline Y_j(q)-\frac{Q^2 y - 5Q x + Q y + x}{x}y\sumq\sumj \overline Y_j(q)\right. \displaybreak[0]\\
    &~\qquad\qquad -6x\sumq\overline Y_J(q)+6y\sumq\sumj \overline Y_j(q) +(-Q^2 y + 2Q x - Q y + 2x)\sumq\sumj \bbone(j=q)\overline Y_j(q)\displaybreak[0]\\
    &~\qquad\qquad\left.-\frac{-Q^2 y + 2Q x - Q y + 2x}{x}y\sumq\sumj\overline Y_j(q)+6x\sumq\sumj\bbone(j>q)\overline Y_j(q)-6y\sumq\sumj (Q-q+1)\overline Y_j(q) \right\} 
    \displaybreak[0]\\
    =&~\sumq\sumjp \frac{Q^2 y - 2Q x + Q y + (4-6j)x}{12Q\bfOmega}\overline Y_j(q)-\frac{6x}{12\bfOmega}\sumq\overline Y_J(q) \displaybreak[0]\\
    &~\quad+\frac{1}{12\bfOmega}\sumq\sumj \left\{-\frac{Q^2 y - 5Q x + Q y + x}{x}y+6y+(-Q^2 y + 2Q x - Q y + 2x)\bbone(j=q)\right. \displaybreak[0]\\
    &~\qquad\qquad\qquad\qquad\quad\left.-\frac{-Q^2 y + 2Q x - Q y + 2x}{x}y+6x\bbone(j>q)-6(Qy-qy+y)\right\}\overline Y_j(q) 
    \displaybreak[0]\\
    =&~\frac{1}{12\bfOmega}\sumq\sumj \left\{\left(-3Q + 6q - 3\right)y+(-Q^2 y + 2Q x - Q y + 2x)\bbone(j=q)+6x\bbone(j>q)\right. \displaybreak[0]\\
    &~\qquad\qquad\qquad\quad\left.+\frac{Q^2 y - 2Q x + Q y + (4-6j)x}{Q}+\frac{-Q^2y + 2Qx - Qy + 2x}{Q}\bbone(j = J)\right\}\overline Y_j(q) 
    \displaybreak[0]\\
    =&~\frac{1}{(Q - 1)(\phi Q^2 - 2\phi Q + 2Q - 1)} \displaybreak[0]\\
    &~\quad\times\sumq\sumj \left\{(Q+1)(\phi Q + 2)\bbone(j=q)+6(\phi Q + 1)\bbone(j>q)+\frac{(Q+1)(\phi Q + 2)}{Q}\bbone(j = J)\right. \displaybreak[0]\\
    &~\qquad\qquad\qquad\quad\left.-4\phi Q - 6\phi j + 6\phi q + 2\phi - 2 - \frac{6j}{Q} + \frac{4}{Q}\right\}\overline Y_j(q). 
\end{align*}
For ease of notation, let 
\begin{align*}
    T_2 = \frac{6}{\phi Q(1 - Q^2) + (2Q^2 - 3Q + 1)(\phi Q + 1)}
\end{align*} and 
\begin{align*}
    T_3 = \frac{1}{(Q - 1)(\phi Q^2 - 2\phi Q + 2Q - 1)}.
\end{align*}
We proceed to use these closed-form expressions of $\widehat\delta^\HHANT$ and $\widehat\gamma^\HHANT$ to derive their expectation $\bbE(\widehat\delta^\HHANT)$ and $\bbE(\widehat\gamma^\HHANT)$ when the true model is either ETI in \eqref{model::ETI} (Theorem \ref{theorem::HH-ANT_1} from Web Appendix A.3.2) or ETI-ANT in \eqref{model::ETI-ANT} (Theorem \ref{theorem::HH-ANT_2} from Web Appendix A.3.3).

\subsubsection*{A.3.2 True Model: ETI} \label{sec::HH-ANT_under_ETI}

\begin{theorem}\label{theorem::HH-ANT_1}
    \emph{For a standard SW-CRT with data generated from the ETI model, the expectations of $\widehat\delta^\HHANT$ and $\widehat\gamma^\HHANT$ can be expressed as a weighted average of the point treatment effects:
    \begin{align*}
        \bbE(\widehat\delta^\HHANT) = \sumjp \pi_\HHANT^\ETI(j) \delta^\ETI(j) \quad\text{and}\quad \bbE(\widehat\gamma^\HHANT) = \sumjp \psi_\HHANT^\ETI(j) \delta^\ETI(j),,
    \end{align*}
    where $\delta^\ETI(j)$ is the true point treatment effect with exposure time $j \in \{1, \ldots, J-1\}$. The weights of $\delta^\ETI(j)$ for $\bbE(\widehat\delta^\HHANT)$ and $\bbE(\widehat\gamma^\HHANT)$ are given by
    \begin{align*}
        \pi_\HHANT^\ETI(j)&=\frac{6 \left( \phi Q^3 - 3 \phi Q^2 j + \phi Q^2 + Q^2 + 2 \phi Q j^2 - 2 \phi Q j + \phi Q - 2 Q j + j^2 \right)}{Q \left( \phi Q^3 - 3 \phi Q^2  + 2 Q^2 + 2 \phi Q - 3 Q + 1 \right)},
    \end{align*}
    and
    \begin{align*}
        \psi_\HHANT^\ETI(j)&=\frac{2\phi Q^3 - 8\phi Q^2 j + 5\phi Q^2 + Q^2 + 6\phi Qj^2 - 8\phi Qj + 3\phi Q - 4Qj + Q + 3j^2 - j}{Q(Q - 1)(\phi Q^2 - 2\phi Q + 2Q - 1)},
    \end{align*}
    respectively.}
\end{theorem}

\begin{proof}

    When the true model is ETI, the expectation of the HH-ANT model treatment effect estimator $\widehat\delta^\HHANT$ is
    \begin{align*}
        \bbE(\widehat\delta^\HHANT) =&~ T_2\sumq\sumj \left[-\phi(Q - 2q + 1) + \frac{\phi Q + 1}{Q}\{Q\bbone(j=q) + 2Q\bbone(j>q) + \bbone(j=J) - 2j + 1\}\right]\bbE\{\overline Y_j(q)\} \displaybreak[0]\\
        =&~T_2\sumq\sumj \left[-\phi(Q - 2q + 1) + \frac{\phi Q + 1}{Q}\{Q\bbone(j=q) + 2Q\bbone(j>q) + \bbone(j=J) - 2j + 1\}\right] \displaybreak[0]\\
        &~\qquad \times \left\{\mu + \beta_j + \bbone(j > q) \delta^\ETI(j-q) \right\}. 
    \end{align*}
    For terms with $\mu$,
    \begin{align*}
        &~T_2\mu \sumq\sumj \left[-\phi(Q - 2q + 1) + \frac{\phi Q + 1}{Q}\{Q\bbone(j=q) + 2Q\bbone(j>q) + \bbone(j=J) - 2j + 1\}\right] \displaybreak[0]\\
        =&~T_2\mu \left[-\phi\left\{Q^2J-2J\frac{(1+Q)Q}{2}+QJ\right\} + \sumq\sumj \frac{\phi Q + 1}{Q}\{Q\bbone(j=q) + 2Q\bbone(j>q) + \bbone(j=J) - 2j + 1\} \right] 
        \displaybreak[0]\\
        =&~\frac{T_2\mu }{Q}(\phi Q + 1)\sumq\sumj\{Q\bbone(j=q) + 2Q\bbone(j>q) + \bbone(j=J) - 2j + 1\} 
        \displaybreak[0]\\
        =&~\frac{T_2\mu}{Q}(\phi Q + 1)\{Q^2+Q^2(Q+1)+Q-Q(Q+2)(Q+1)+Q(Q+1)\} 
        \displaybreak[0]\\
        =&~T_2\mu(\phi Q + 1)\{Q+Q(Q+1)+1-(Q+2)(Q+1)+(Q+1)\}
        \displaybreak[0]\\
        =&~0. 
    \end{align*}
    For terms with $\beta_j$,
    \begin{align*}
        &~T_2\sumq\sumj\left[-\phi(Q - 2q + 1) + \frac{\phi Q + 1}{Q}\{Q\bbone(j=q) + 2Q\bbone(j>q) + \bbone(j=J) - 2j + 1\}\right]\beta_j \displaybreak[0]\\
        =&~T_2\mu \left[\left\{-\phi Q^2 + 2\phi \frac{Q(1+Q)}{2} - \phi Q\right\}(\beta_1+\ldots+\beta_J)\right.\displaybreak[0]\\
        &\qquad\qquad+\left. \sumq\sumj\frac{\phi Q + 1}{Q}\{Q\bbone(j=q) + 2Q\bbone(j>q) + \bbone(j=J) - 2j + 1\}\beta_j \right] 
        \displaybreak[0]\\
        =&~\frac{T_2\mu}{Q}(\phi Q + 1)\sumq\sumj\{Q\bbone(j=q) + 2Q\bbone(j>q) + \bbone(j=J) - 2j + 1\}\beta_j
        \displaybreak[0]\\
        =&~\frac{T_2\mu}{Q}(\phi Q + 1)\{(Q\beta_1+\dots+Q\beta_Q)+2Q(\beta_2+2\beta_3+\dots+Q\beta_J)+Q\beta_J-2Q(\beta_1+2\beta_2+\dots+J\beta_J)+Q(\beta_1+\dots+\beta_J)\} 
        \displaybreak[0]\\
        =&~T_2\mu(\phi Q + 1)\{(\beta_1+\dots+\beta_Q)+2(\beta_2+2\beta_3+\dots+Q\beta_J)+\beta_J-2(\beta_1+2\beta_2+\dots+J\beta_J)+(\beta_1+\dots+\beta_J)\} 
        \displaybreak[0]\\
        =&~0. 
    \end{align*}
    For terms with $\delta^\ETI(j-q)$,
    \begin{align*}
        &~T_2\sumq\sumj\left[-\phi(Q - 2q + 1) + \frac{\phi Q + 1}{Q}\{Q\bbone(j=q) + 2Q\bbone(j>q) + \bbone(j=J) - 2j + 1\}\right]\bbone(j>q) \delta^\ETI(j-q)\displaybreak[0]\\
        =&~ -\phi(Q+1)T_2\sumq\sumj\bbone(j>q) \delta^\ETI(j-q) + 2\phi T_2\sumq\sumj q \bbone(j>q) \delta^\ETI(j-q)-\frac{2T_2(\phi Q + 1)}{Q}\sumq\sumj j \bbone(j>q) \delta^\ETI(j-q)\displaybreak[0]\\
        &~\quad+T_2(\phi Q + 1)\sumq\sumj \left\{2\bbone(j>q) +\frac{\bbone(j=J) + 1}{Q}\right\} \bbone(j>q) \delta^\ETI(j-q) 
        \displaybreak[0]\\
        =&~-\phi(Q+1)T_2\sumq(Q-q+1)\delta^\ETI(q) +  \phi T_2\sumq(Q-q+1)(Q-q+2)\delta^\ETI(q) \displaybreak[0]\\
        &~\quad - \frac{T_2(\phi Q + 1)}{Q}\sumq \{(Q+1)(Q+2)-q(q+1)\}\delta^\ETI(q) \displaybreak[0]\\
        &~\quad +\frac{T_2(\phi Q + 1)}{Q}\sumq \delta^\ETI(q)+T_2(\phi Q + 1)\sumq \left(2+\frac{1}{Q}\right)(Q-q+1)\delta^\ETI(q) 
        \displaybreak[0]\\
        =&~T_2\sumq \left[-\phi(Q+1)(Q-q+1)+\phi(Q-q+1)(Q-q+2)- \frac{(\phi Q + 1)}{Q}\{(Q+1)(Q+2)-q(q+1)\}\right. \displaybreak[0]\\
        &\left.\qquad\qquad+\frac{(\phi Q + 1)}{Q} + (\phi Q + 1)\left(2+\frac{1}{Q}\right)(Q-q+1)\right]\delta^\ETI(q) 
        \displaybreak[0]\\
        =&~\sumq \frac{6 \left( \phi Q^3 - 3 \phi Q^2 q + \phi Q^2 + Q^2 + 2 \phi Q q^2 - 2 \phi Q q + \phi Q - 2 Q q + q^2 \right)}{Q \left( \phi Q^3- 3 \phi Q^2 + 2 Q^2 + 2 \phi Q - 3 Q + 1 \right)} \delta^\ETI(q). 
    \end{align*}
     Similarly, when the true model is ETI, the expectation of the HH-ANT model anticipation effect estimator $\widehat\gamma^\HHANT$ is
    \begin{align*}
        \bbE(\widehat\gamma^\HHANT) =&~ T_3\sumq\sumj \left\{(Q+1)(\phi Q + 2)\bbone(j=q)+6(\phi Q + 1)\bbone(j>q)+\frac{(Q+1)(\phi Q + 2)}{Q}\bbone(j = J)\right. \displaybreak[0]\\
        &~\qquad\qquad\qquad\left.-4 \phi Q - 6\phi j + 6\phi q + 2\phi - 2 - \frac{6j}{Q} + \frac{4}{Q}\right\}\bbE\{\overline Y_j(q)\} \displaybreak[0]\\
        =&~T_3\sumq\sumj \left\{(Q+1)(\phi Q + 2)\bbone(j=q)+6(\phi Q + 1)\bbone(j>q)+\frac{(Q+1)(\phi Q + 2)}{Q}\bbone(j = J)\right. \displaybreak[0]\\
        &~\qquad\qquad\qquad\left.-4 \phi Q - 6\phi j + 6\phi q + 2\phi - 2 - \frac{6j}{Q} + \frac{4}{Q}\right\}\left\{\mu + \beta_j + \bbone(j > q) \delta^\ETI(j-q) \right\}. 
    \end{align*}
    For terms with $\mu$, we have
    \begin{align*}
        &~ T_3\mu\sumq\sumj \left\{(Q+1)(\phi Q + 2)\bbone(j=q)+6(\phi Q + 1)\bbone(j>q)+\frac{(Q+1)(\phi Q + 2)}{Q}\bbone(j = J) \right. \displaybreak[0]\\
        &~\left. \qquad\qquad\qquad -4 \phi Q - 6\phi j + 6\phi q + 2\phi - 2 - \frac{6j}{Q} + \frac{4}{Q}\right\} \mu \displaybreak[0]\\
        =&~T_3\mu\left[Q(Q+1)(\phi Q + 2) + 3Q(\phi Q + 1)(Q + 1) + (Q+1)(\phi Q + 2) + J \left\{-4\phi Q^2-3\phi Q+2\phi Q-2Q-3(1+J)+4\right\}\right] 
        \displaybreak[0]\\
        =&~0.
    \end{align*}
    For terms with $\beta_j$, we have
    \begin{align*}
        &~T_3\sumq\sumj \left\{(Q+1)(\phi Q + 2)\bbone(j=q)+6(\phi Q + 1)\bbone(j>q)+\frac{(Q+1)(\phi Q + 2)}{Q}\bbone(j = J)\right. \displaybreak[0]\\
        &~\left. \qquad\qquad\qquad -4 \phi Q - 6\phi j + 6\phi q + 2\phi - 2 - \frac{6j}{Q} + \frac{4}{Q}\right\}\beta_j \displaybreak[0]\\
        =&~T_3\{(Q+1)(\phi Q + 2)(\beta_1+\dots+\beta_Q)+6(\phi Q + 1)(\beta_2+2\beta_3+\dots+Q\beta_J)+(Q+1)(\phi Q + 2)\beta_J \displaybreak[0]\\
        &~\quad-(6\phi Q+6)(\beta_1+2\beta_2+\dots+J\beta_J)-(4\phi Q^2-2\phi Q+2Q-4)(\beta_1+\dots+\beta_J)+3\phi Q(1 + Q)(\beta_1+\dots+\beta_J)\} 
        \displaybreak[0]\\
        =&~0.
    \end{align*}
    For terms with $\delta^\ETI(j-q)$, we have
    \begin{align*}
        &~ T_3\sumq\sumj \left\{(Q+1)(\phi Q + 2)\bbone(j=q)+6(\phi Q + 1)\bbone(j>q)+\frac{(Q+1)(\phi Q + 2)}{Q}\bbone(j = J)\right. \displaybreak[0]\\
        &~\qquad\qquad\qquad\left.-4 \phi Q - 6\phi j + 6\phi q + 2\phi - 2 - \frac{6j}{Q} + \frac{4}{Q}\right\}\bbone(j>q) \delta^\ETI(j-q) \displaybreak[0]\\
        =&~ 6(\phi Q + 1)T_3\sumq\sumj\bbone(j>q) \delta^\ETI(j-q) + T_3\frac{(Q+1)(\phi Q + 2)}{Q}\sumq\sumj \bbone(j = J) \bbone(j>q) \delta^\ETI(j-q) \displaybreak[0]\\
        &~\quad-T_3\left(6\phi+\frac{6}{Q}\right)\sumq\sumj j \bbone(j>q) \delta^\ETI(j-q)+6\phi T_3\sumq\sumj q \bbone(j>q)\delta^\ETI(j-q)\displaybreak[0]\\
        &~\quad-T_3\left(4\phi Q-2\phi+2-\frac{4}{Q}\right)\sumq\sumj\bbone(j>q) \delta^\ETI(j-q)
        \displaybreak[0]\\
        =&~ 6(\phi Q + 1)T_3\sumq(Q-q+1)\delta^\ETI(q) + T_3\frac{(Q+1)(\phi Q + 2)}{Q}\sumq \delta^\ETI(q)\displaybreak[0]\\
        &~\quad -T_3\left(3\phi+\frac{3}{Q}\right)\sumq \{(Q+1)(Q+2)-q(q+1)\}\delta^\ETI(q) \displaybreak[0]\\
        &~\quad+3\phi T_3\sumq(Q-q+1)(Q-q+2)\delta^\ETI(q)-T_3\left(4\phi Q-2\phi+2-\frac{4}{Q}\right)\sumq(Q-q+1)\delta^\ETI(q)
        \displaybreak[0]\\
        =&~T_3\sumq \left\{6(\phi Q + 1)(Q-q+1)+\frac{(Q+1)(\phi Q + 2)}{Q}-\left(3\phi+\frac{3}{Q}\right)\{(Q+1)(Q+2)-q(q+1)\}\right. \displaybreak[0]\\
        &~\left.\qquad\qquad+3\phi(Q-q+1)(Q-q+2)-\left(4\phi Q-2\phi+2-\frac{4}{Q}\right)(Q-q+1)\right\}\delta^\ETI(q) 
        \displaybreak[0]\\
        =&~\sumq \frac{2\phi Q^3 - 8\phi Q^2 q + 5\phi Q^2 + Q^2 + 6\phi Qq^2 - 8\phi Qq + 3\phi Q - 4Qq + Q + 3q^2 - q}{Q(Q - 1)(\phi Q^2 - 2\phi Q + 2Q - 1)} \delta^\ETI(q).
    \end{align*}
    Putting everything together,
    \begin{align*}
        \bbE(\widehat\delta^\HHANT) = \sumjp \pi_\HHANT^\ETI(j) \delta^\ETI(j) \quad\text{and}\quad \bbE(\widehat\gamma^\HHANT) = \sumjp \psi_\HHANT^\ETI(j) \delta^\ETI(j),
    \end{align*}
    where
    \begin{align*}
        \pi_\HHANT^\ETI(j)=\frac{6 \left( \phi Q^3 - 3 \phi Q^2 j + \phi Q^2 + Q^2 + 2 \phi Q j^2 - 2 \phi Q j + \phi Q - 2 Q j + j^2 \right)}{Q \left( \phi Q^3 - 3 \phi Q^2  + 2 Q^2 + 2 \phi Q - 3 Q + 1 \right)},
    \end{align*}
    and
    \begin{align*}
        \psi_\HHANT^\ETI(j)=\frac{2\phi Q^3 - 8\phi Q^2 j + 5\phi Q^2 + Q^2 + 6\phi Qj^2 - 8\phi Qj + 3\phi Q - 4Qj + Q + 3j^2 - j}{Q(Q - 1)(\phi Q^2 - 2\phi Q + 2Q - 1)}.
    \end{align*}
    
\end{proof}

\subsubsection*{A.3.3 True Model: ETI-ANT} \label{sec::HH-ANT_under_ETI-ANT}

\begin{theorem}\label{theorem::HH-ANT_2}
    \emph{For a standard SW-CRT with data generated from the ETI-ANT model, the expectation of $\widehat\delta^\HHANT$ and $\widehat\gamma^\HHANT$ can be expressed as a weighted average of the point treatment effects:
    \begin{align*}
        \bbE(\widehat\delta^\HHANT) = \sumjp \pi_\HHANT^\ETIANT(j) \delta^\ETIANT(j) \quad\text{and}\quad \bbE(\widehat\gamma^\HHANT) = \gamma^\ETIANT + \sumjp \psi_\HHANT^\ETIANT(j) \delta^\ETIANT(j),
    \end{align*}
    where  $\delta(j)$ is the true point treatment effect with exposure time $j \in \{1, \ldots, J-1\}$. The weights of $\delta(j)$ for $\bbE(\widehat\delta^\HHANT)$ and $\bbE(\widehat\gamma^\HHANT)$ are given by
    \begin{align*}
        \pi_\HHANT^\ETIANT(j)=\frac{6 \left( \phi Q^3 - 3 \phi Q^2 j + \phi Q^2 + Q^2 + 2 \phi Q j^2 - 2 \phi Q j + \phi Q - 2 Q j + j^2 \right)}{Q \left( \phi Q^3 - 3 \phi Q^2  + 2 Q^2 + 2 \phi Q - 3 Q + 1 \right)},
    \end{align*}
    and
    \begin{align*}
        \psi_\HHANT^\ETIANT(j)=\frac{2\phi Q^3 - 8\phi Q^2 j + 5\phi Q^2 + Q^2 + 6\phi Qj^2 - 8\phi Qj + 3\phi Q - 4Qj + Q + 3j^2 - j}{Q(Q - 1)(\phi Q^2 - 2\phi Q + 2Q - 1)},
    \end{align*}
    respectively.}
\end{theorem}

\begin{proof}

     When the true model is ETI-ANT, the expectation of the HH-ANT model treatment effect estimator $\widehat\delta^\HHANT$ is
     \begin{align*}
        \bbE(\widehat\delta^\HHANT) = T_2\sumq\sumj \left[-\phi(Q - 2q + 1) + \frac{\phi Q + 1}{Q}\{Q\bbone(j=q) + 2Q\bbone(j>q) + \bbone(j=J) - 2j + 1\}\right]\bbE\{\overline Y_j(q)\},
    \end{align*}
    where $\bbE\{\overline Y_j(q)\} = \mu + \beta_j + \bbone(j = q) \gamma^\ETIANT + \bbone(j > q) \delta^\ETIANT(j-q)$. Of note, terms related to $\mu$, $\beta_j$, and $\delta^\ETIANT(j-q)$ are simplified in A.3.2, and we focus on the only remaining terms related to $\gamma^\ETIANT$, which is
    \begin{align*}
        &~T_2\left[-\phi(Q - 2q + 1) + \frac{\phi Q + 1}{Q}\{Q\bbone(j=q) + 2Q\bbone(j>q) + \bbone(j=J) - 2j + 1\}\right]\bbone(j = q)\gamma^\ETIANT\displaybreak[0]\\
        =&~T_2\sumq\sumj\left\{-\phi(Q - 2q + 1) + \frac{(\phi Q + 1)(Q - 2j + 1)}{Q}\right\}\bbone(j = q)\gamma^\ETIANT
        \displaybreak[0]\\
        =&~-T_2\gamma^\ETIANT\phi(Q + 1)Q + T_2\gamma^\ETIANT\phi(1+Q)Q + T_2\gamma^\ETIANT(Q+1) (\phi Q + 1)Q-T_2\gamma^\ETIANT(\phi Q + 1)(1+Q)Q
        \displaybreak[0]\\
        =&~0.
    \end{align*}
    Similarly, when the true model is ETI-ANT, the expectation of the HH-ANT model anticipation effect estimator $\widehat\gamma^\HHANT$ is
    \begin{align*}
        \bbE(\widehat\gamma^\HHANT) =&~ T_3\sumq\sumj \left\{(Q+1)(\phi Q + 2)\bbone(j=q)+6(\phi Q + 1)\bbone(j>q)+\frac{(Q+1)(\phi Q + 2)}{Q}\bbone(j = J)\right. \displaybreak[0]\\
        &~\qquad\qquad\qquad\left.-4 \phi Q - 6\phi j + 6\phi q + 2\phi - 2 - \frac{6j}{Q} + \frac{4}{Q}\right\}\bbE\{\overline Y_j(q)\},
    \end{align*}
    where $\bbE\{\overline Y_j(q)\} = \mu + \beta_j + \bbone(j = q) \gamma^\ETIANT + \bbone(j > q) \delta^\ETIANT(j-q)$. The only left terms are related $\gamma^\ETIANT$, which is
    \begin{align*}
        &~T_3\sumq\sumj \left[(Q+1)(\phi Q + 2)\bbone(j=q)+6(\phi Q + 1)\bbone(j>q)+\frac{(Q+1)(\phi Q + 2)}{Q}\bbone(j = J)\right.\displaybreak[0]\\
        &~\left.\qquad\qquad\qquad-4 \phi Q - 6\phi j + 6\phi q + 2\phi - 2 - \frac{6j}{Q} + \frac{4}{Q}\right]\bbone(j = q)\gamma^\ETIANT\displaybreak[0]\\
        =&~T_3\sumq\sumj \left[(Q+1)(\phi Q + 2)\bbone(j=q)-4 \phi Q - 6\phi j + 6\phi q + 2\phi - 2 - \frac{6j}{Q} + \frac{4}{Q}\right]\bbone(j = q)\gamma^\ETIANT
        \displaybreak[0]\\
        =&~T_3\left\{Q(Q+1)(\phi Q + 2)-Q\left(4 \phi Q-2\phi+2-\frac{4}{Q}\right)-\left(6\phi-6\phi+\frac{6}{Q}\right)\frac{Q(1+Q)}{2}\right\}\gamma^\ETIANT
        \displaybreak[0]\\
        =&~\gamma^\ETIANT.
    \end{align*}
    Putting everything together,
    \begin{align*}
        \bbE(\widehat\delta^\HHANT) = \sumjp \pi_\HHANT^\ETIANT(j) \delta^\ETIANT(j) \quad\text{and}\quad \bbE(\widehat\gamma^\HHANT) = \gamma^\ETIANT + \sumjp \psi_\HHANT^\ETIANT(j) \delta^\ETIANT(j),
    \end{align*}
    where
    \begin{align*}
        \pi_\HHANT^\ETIANT(j)=\frac{6 \left( \phi Q^3 - 3 \phi Q^2 j + \phi Q^2 + Q^2 + 2 \phi Q j^2 - 2 \phi Q j + \phi Q - 2 Q j + j^2 \right)}{Q \left( \phi Q^3 - 3 \phi Q^2  + 2 Q^2 + 2 \phi Q - 3 Q + 1 \right)},
    \end{align*}
    and
    \begin{align*}
        \psi_\HHANT^\ETIANT(j)=\frac{2\phi Q^3 - 8\phi Q^2 j + 5\phi Q^2 + Q^2 + 6\phi Qj^2 - 8\phi Qj + 3\phi Q - 4Qj + Q + 3j^2 - j}{Q(Q - 1)(\phi Q^2 - 2\phi Q + 2Q - 1)},
    \end{align*}
    
\end{proof}

\section*{Web Appendix B: Closed-form Variance Expressions of Linear Mixed Model Estimators} \label{sec::variance}

\subsection*{B.1 General SW-CRTs} \label{sec::variance_general}

To facilitate the sample size calculation for SW-CRTs with the anticipation effect, we give expressions for the variances of treatment effect estimators from the HH model,\cite{Hussey2007, Li2020} the HH-ANT model, the ETI model,\cite{Maleyeff2022} and the ETI-ANT model. We consider a general SW-CRT with $I$ clusters, $J$ periods, and cluster-period size $K_{ij}=K$ for all $i$ and $j$. Also, we let $\lambda_1 = 1-\rho$ and $\lambda_2 = 1+(JK-1)\rho$, with the total variance $\sigma_t^2 = \tau^2+\sigma^2$ and the ICC $\rho = \tau^2/\sigma_t^2$.

\begin{theorem} \label{theorem::var}
    \emph{For an SW-CRT with $I$ clusters, $J$ periods, and $K$ individuals per cluster-period, variances of treatment effect estimators are
    \begin{align} \label{var::HH_general}
        \bbV(\widehat\delta^\HH)= \frac{(\sigma_t^2/K)IJ\lambda_1\lambda_2}{(U^2+IJU-JW_1-IW_2)\lambda_2-(U^2-IW_2)\lambda_1}
    \end{align}
    for the HH model estimator,
    \begin{align} \label{var::HH-ANT_general}
        \bbV(\widehat\delta^\HHANT)= \frac{IJ\lambda_1\lambda_2\sigma_t^2}{K} \left\{\lambda_2\left(U^2+IJU-JW_1-IW_2-\frac{JW_5^2}{I^2-W_3}\right) + \lambda_1\left(IW_2-U^2\right)\right\}^{-1}
    \end{align}
    for the HH-ANT model estimator,
    \begin{align} \label{var::ETI_general}
        \bbV\left\{\frac{1}{J-1}\sums\widehat\delta^\ETI(s)\right\}= \frac{IJ\lambda_1\lambda_2\sigma_t^2}{K(J-1)^2} \mathbf{1}_{J-1}'\left\{\lambda_2\left(\bfU_1^{\otimes 2}+IJ\bfU_2-J\bfW_1-I\bfW_2\right) + \lambda_1\left(I\bfW_2-\bfU_1^{\otimes 2}\right)\right\}^{-1}\mathbf{1}_{J-1}
    \end{align}
    for the ETI model estimator, and
    \begin{equation}\label{var::ETI-ANT_general}
    \begin{aligned} 
        &\bbV\left\{\frac{1}{J-1}\sums\widehat\delta^\ETIANT(s)\right\}\\
        &\qquad\qquad=\frac{IJ\lambda_1\lambda_2\sigma_t^2}{K(J-1)^2} \mathbf{1}_{J-1}'\left\{\lambda_2\left(\bfU_1^{\otimes 2}+IJ\bfU_2-J\bfW_1-I\bfW_2 - \frac{J \bfW_5^{\otimes 2}}{I^2-W_3}\right) + \lambda_1\left(I\bfW_2-\bfU_1^{\otimes 2}\right)\right\}^{-1}\mathbf{1}_{J-1}
    \end{aligned}
    \end{equation}
    for the ETI-ANT model estimator. In \eqref{var::HH_general} and \eqref{var::HH-ANT_general}, $U = \sumi\sumj Z_{ij}$, $W_1 = \sumj (\sumi Z_{ij})^2$, and $W_2 = \sumi(\sumj Z_{ij})^2$, which are design constants that depend on the treatment assignment. In \eqref{var::HH-ANT_general}, $W_3 = (\sumi \bfA_i')(\sumi \bfA_i)$ and $W_5 = (\sumi \bfZ_i')(\sumi \bfA_i)$. When $\ell = 1$, $W_4$ is omitted because $W_4 = \sumi \bfA_i'\mathbf{1}_J\mathbf{1}_J'\bfA_i = I$, where $\mathbf{1}_J$ is a column vector of $J$ ones. In \eqref{var::ETI_general} and \eqref{var::ETI-ANT_general}, $\bfU_1 = \sumi \bfX_i'\mathbf{1}_J$, $\bfU_1^{\otimes 2} = \bfU_1\bfU_1'$, $\bfU_2 = \sumi \bfX_i'\bfX_i$, $\bfW_1 = (\sumi \bfX_i')(\sumi \bfX_i)$, $\bfW_2 = \sumi \bfX_i'\mathbf{1}_J\mathbf{1}_J'\bfX_i$, $W_3 = (\sumi \bfA_i')(\sumi \bfA_i)$, and $\bfW_5 = (\sumi \bfX_i')(\sumi \bfA_i)$. Here, $\bfX_i = \{\mathbf{0}_{(j^*-1)\times (J-1)}', (\bfI_{J-j^*+1}, \mathbf{0}_{(J-j^*+1)\times (j^*-2)})'\}'$ for the treatment adoption time $j^*\geq 2$ (by the definition of SW-CRTs), where $\mathbf{0}_{(j^*-1)\times (J-1)}$ is a $(j^*-1)\times (J-1)$ zero matrix and $\bfI_{J-j^*+1}$ is a $(J-j^*+1)\times (J-j^*+1)$ identity matrix.}
\end{theorem}

\begin{proof}

    Since Li et al. \cite{Li2020} and Maleyeff et al. \cite{Maleyeff2022} have already provided the variances of treatment effect estimators $\bbV(\widehat\delta^\HH)$ and $\bbV(\widehat\Delta^\ETI)$, respectively, we omit the corresponding proofs. For the remaining two variance formulas, we first derive $\bbV(\widehat\Delta^\ETIANT)$ following the procedures in Maleyeff et al. \cite{Maleyeff2022} Then, we derive $\bbV(\widehat\delta^\HHANT)$, which is simply a reduced version of $\bbV(\widehat\Delta^\ETIANT)$. To facilitate the derivation, we can re-parameterize \eqref{model::ETI-ANT} by dropping the grand mean, such that the model becomes
    \begin{align}
        Y_{ijk} = \beta_j + \gamma^\ETIANT A_{ij} + \delta^\ETIANT(s_{ij}) Z_{ij} + \alpha_i + \epsilon_{ijk}, \label{model::ETI-ANT_re}
    \end{align}
    where there is no restriction on $\beta_j$. For \eqref{model::ETI-ANT_re}, the variance of the estimated average treatment effect (i.e., $\bbV\{(J-1)^{-1}\sums\widehat\delta^\ETIANT(s)\}$) is equal to $ (J-1)^{-2}\mathbf{1}_{J-1}'\bbV(\widehat\bfdelta^\ETIANT)\mathbf{1}_{J-1}$, where $\widehat\bfdelta^\ETIANT = (\widehat\delta(1), \ldots, \widehat\delta(J-1))^\prime$ and $\bbV(\widehat\bfdelta^\ETIANT)$ is the lower $(J-1) \times (J-1)$ matrix of $(\bfG'\bfC^{-1}\bfG)^{-1}$. Here,
    \begin{align*}
        \bfG = \left(\begin{array}{ccc}
            \bfI_J & \bfA_1 & \bfX_1\\
            \bfI_J & \bfA_2 & \bfX_2\\
            \vdots & \vdots & \vdots\\
            \bfI_J & \bfA_I & \bfX_I
        \end{array}\right)\otimes \mathbf{1}_K,\quad 
            \bfX_i = \left(\begin{array}{c}
             \mathbf{0}_{(j^*-1)\times (J-1)} \\
             \bfI_{J-j^*+1}, \mathbf{0}_{(J-j^*+1)\times (j^*-2)}
        \end{array}\right),\quad
        \bfA_i = \left(\begin{array}{c}
            \mathbf{0}_{j^*-2}\\
            1\\
            \mathbf{0}_{J-j^*+1}
        \end{array}\right),\quad
            \bfC = \left(\begin{array}{cccc}
            \bfC_1 &  \mathbf{0} & \ldots & \mathbf{0}\\
            \mathbf{0} &  \bfC_2 & \ldots & \mathbf{0}\\
            \vdots &  \vdots & \ddots & \vdots \\
            \mathbf{0} & \mathbf{0} & \ldots & \bfC_I
        \end{array}\right),
    \end{align*}
    where $j^*$ is the treatment adoption time with $j^*\geq 2$ (by the definition of SW-CRTs), and $\bfC_i = \sigma^2\bfI_{JK}+\tau^2\mathbf{1}_{JK}\mathbf{1}_{JK}'$. Using the property of Kronecker products that $(\bfA \otimes \bfB)' = \bfA' \otimes \bfB'$, we have
    \begin{align*}
        \bfG' &= \left(\begin{array}{ccc}
            \bfI_J & \bfA_1 & \bfX_1\\
            \bfI_J & \bfA_2 & \bfX_2\\
            \vdots & \vdots & \vdots\\
            \bfI_J & \bfA_I & \bfX_I
        \end{array}\right)' \otimes \mathbf{1}_K' =\left(\begin{array}{cccc}
            \bfI_J & \bfI_J & \ldots & \bfI_J\\
            \bfA_1' & \bfA_2' & \ldots & \bfA_I'\\
            \bfX_1' & \bfX_2' & \ldots & \bfX_I'
        \end{array}\right) \otimes \mathbf{1}_K'.
    \end{align*}
    Next, we have
    \begin{align*}
        \bfG'\bfC^{-1} &= \left\{\left(\begin{array}{cccc}
            \bfI_J & \bfI_J & \ldots & \bfI_J\\
            \bfA_1' & \bfA_2' & \ldots & \bfA_I'\\
            \bfX_1' & \bfX_2' & \ldots & \bfX_I'
        \end{array}\right) \otimes \mathbf{1}_K'\right\}\left(\begin{array}{cccc}
            \bfC_1^{-1} &  \mathbf{0} & \ldots & \mathbf{0}\\
            \mathbf{0} &  \bfC_2^{-1} & \ldots & \mathbf{0}\\
            \vdots &  \vdots & \ddots & \vdots \\
            \mathbf{0} & \mathbf{0} & \ldots & \bfC_I^{-1}
        \end{array}\right) =\left(\begin{array}{cccc}
            (\bfI_J\otimes \mathbf{1}_K')\bfC_1^{-1} & (\bfI_J\otimes \mathbf{1}_K')\bfC_2^{-1} & \ldots & (\bfI_J\otimes \mathbf{1}_K')\bfC_I^{-1}\\
            (\bfA_1'\otimes \mathbf{1}_K')\bfC_1^{-1} & (\bfA_2'\otimes \mathbf{1}_K')\bfC_2^{-1} & \ldots & (\bfA_I'\otimes \mathbf{1}_K')\bfC_I^{-1}\\
            (\bfX_1'\otimes \mathbf{1}_K')\bfC_1^{-1} & (\bfX_2'\otimes \mathbf{1}_K')\bfC_2^{-1} & \ldots & (\bfX_I'\otimes \mathbf{1}_K')\bfC_I^{-1}
        \end{array}\right).
    \end{align*}
    Then,
    \begin{align*}
        \bfG'\bfC^{-1}\bfG &= \left(\begin{array}{cccc}
            (\bfI_J\otimes \mathbf{1}_K')\bfC_1^{-1} & (\bfI_J\otimes \mathbf{1}_K')\bfC_2^{-1} & \ldots & (\bfI_J\otimes \mathbf{1}_K')\bfC_I^{-1}\\
            (\bfA_1'\otimes \mathbf{1}_K')\bfC_1^{-1} & (\bfA_2'\otimes \mathbf{1}_K')\bfC_2^{-1} & \ldots & (\bfA_I'\otimes \mathbf{1}_K')\bfC_I^{-1}\\
            (\bfX_1'\otimes \mathbf{1}_K')\bfC_1^{-1} & (\bfX_2'\otimes \mathbf{1}_K')\bfC_2^{-1} & \ldots & (\bfX_I'\otimes \mathbf{1}_K')\bfC_I^{-1}
        \end{array}\right)\left\{\left(\begin{array}{ccc}
            \bfI_J & \bfA_1 & \bfX_1\\
            \bfI_J & \bfA_2 & \bfX_2\\
            \vdots & \vdots & \vdots\\
            \bfI_J & \bfA_I & \bfX_I
        \end{array}\right)\otimes \mathbf{1}_K\right\}\displaybreak[0]\\
        &=\left(\begin{array}{ccc}
            \sumi(\bfI_J\otimes \mathbf{1}_K')\bfC_i^{-1}(\bfI_J\otimes\mathbf{1}_K) & \sumi(\bfI_J\otimes \mathbf{1}_K')\bfC_i^{-1}(\bfA_i\otimes\mathbf{1}_K) & \sumi(\bfI_J\otimes \mathbf{1}_K')\bfC_i^{-1}(\bfX_i\otimes\mathbf{1}_K)\\
            \sumi(\bfA_i'\otimes \mathbf{1}_K')\bfC_i^{-1}(\bfI_J\otimes\mathbf{1}_K) & \sumi(\bfA_i'\otimes \mathbf{1}_K')\bfC_i^{-1}(\bfA_i\otimes\mathbf{1}_K) & \sumi(\bfA_i'\otimes \mathbf{1}_K')\bfC_i^{-1}(\bfX_i\otimes\mathbf{1}_K)\\
            \sumi(\bfX_i'\otimes \mathbf{1}_K')\bfC_i^{-1}(\bfI_J\otimes\mathbf{1}_K) & \sumi(\bfX_i'\otimes \mathbf{1}_K')\bfC_i^{-1}(\bfA_i\otimes\mathbf{1}_K) &  \sumi(\bfX_i'\otimes \mathbf{1}_K')\bfC_i^{-1}(\bfX_i\otimes\mathbf{1}_K)
        \end{array}\right).
    \end{align*}
    Since we assume that the covariance matrix is the same for each cluster (due to the assumption of equal cluster size), then $\bfC_1 = \ldots = \bfC_I = \bfC_*$. Thus, we have
    \begin{align*}
        \bfG'\bfC^{-1}\bfG = \left(\begin{array}{ccc}
            I(\bfI_J\otimes \mathbf{1}_K')\bfC_*^{-1}(\bfI_J\otimes\mathbf{1}_K) & (\bfI_J\otimes \mathbf{1}_K')\bfC_*^{-1}\left(\sumi\bfA_i\otimes\mathbf{1}_K\right) & (\bfI_J\otimes \mathbf{1}_K')\bfC_*^{-1}\left(\sumi\bfX_i\otimes\mathbf{1}_K\right)\\
            \left(\sumi\bfA_i'\otimes \mathbf{1}_K'\right)\bfC_*^{-1}(\bfI_J\otimes\mathbf{1}_K) & \sumi(\bfA_i'\otimes \mathbf{1}_K')\bfC_*^{-1}(\bfA_i\otimes\mathbf{1}_K) & \sumi(\bfA_i'\otimes \mathbf{1}_K')\bfC_*^{-1}(\bfX_i\otimes\mathbf{1}_K)\\
            \left(\sumi\bfX_i'\otimes \mathbf{1}_K'\right)\bfC_*^{-1}(\bfI_J\otimes\mathbf{1}_K) & \sumi(\bfX_i'\otimes \mathbf{1}_K')\bfC_*^{-1}(\bfA_i\otimes\mathbf{1}_K) & \sumi(\bfX_i'\otimes \mathbf{1}_K')\bfC_*^{-1}(\bfX_i\otimes\mathbf{1}_K)
        \end{array}\right).
    \end{align*}
    Let
    \begin{align*}
        \left(\begin{array}{c|c}
             \bfLambda_{11} & \bfLambda_{12}\\
             \hline
             \bfLambda_{21} & \bfLambda_{22}
        \end{array}\right) = \left(\begin{array}{c|cc}
            I(\bfI_J\otimes \mathbf{1}_K')\bfC_*^{-1}(\bfI_J\otimes\mathbf{1}_K) & (\bfI_J\otimes \mathbf{1}_K')\bfC_*^{-1}\left(\sumi\bfA_i\otimes\mathbf{1}_K\right) & (\bfI_J\otimes \mathbf{1}_K')\bfC_*^{-1}\left(\sumi\bfX_i\otimes\mathbf{1}_K\right)\\
            \hline
            \left(\sumi\bfA_i'\otimes \mathbf{1}_K'\right)\bfC_*^{-1}(\bfI_J\otimes\mathbf{1}_K) & \sumi(\bfA_i'\otimes \mathbf{1}_K')\bfC_*^{-1}(\bfA_i\otimes\mathbf{1}_K) & \sumi(\bfA_i'\otimes \mathbf{1}_K')\bfC_*^{-1}(\bfX_i\otimes\mathbf{1}_K)\\
            \left(\sumi\bfX_i'\otimes \mathbf{1}_K'\right)\bfC_*^{-1}(\bfI_J\otimes\mathbf{1}_K) & \sumi(\bfX_i'\otimes \mathbf{1}_K')\bfC_*^{-1}(\bfA_i\otimes\mathbf{1}_K) & \sumi(\bfX_i'\otimes \mathbf{1}_K')\bfC_*^{-1}(\bfX_i\otimes\mathbf{1}_K)
        \end{array}\right).
    \end{align*}
    Then, using the definition of the inverse of a partitioned matrix, the lower $J \times J$ matrix of $(\bfG'\bfC^{-1}\bfG)^{-1}$ is given by $\left(\bfLambda_{22}-\bfLambda_{21}\bfLambda_{11}^{-1}\bfLambda_{12}\right)^{-1}$ and so
    \begin{align*}
        &\quad \bbV(\widehat\gamma^\ETIANT, \widehat\bfdelta^\ETIANT)\\
        &= \left\{\left(\begin{array}{c|c}
            \sumi\bfP_i'\bfC_*^{-1}\bfP_i & \sumi\bfP_i'\bfC_*^{-1}\bfO_i\\
            \hline
            \sumi\bfO_i'\bfC_*^{-1}\bfP_i & \sumi\bfO_i'\bfC_*^{-1}\bfO_i
        \end{array}\right)-\frac{1}{I}\left(\begin{array}{c}
            \bfN'\bfC_*^{-1}\bfR\\
            \bfM'\bfC_*^{-1}\bfR
        \end{array}\right)\left(\bfR'\bfC_*^{-1}\bfR\right)^{-1}\left(\begin{array}{cc}
             \bfR'\bfC_*^{-1}\bfN & \bfR'\bfC_*^{-1}\bfM
        \end{array}\right)\right\}^{-1}\displaybreak[0]\\
        &=\left\{\left(\begin{array}{c|c}
            \sumi\bfP_i'\bfC_*^{-1}\bfP_i & \sumi\bfP_i'\bfC_*^{-1}\bfO_i\\
            \hline
            \sumi\bfO_i'\bfC_*^{-1}\bfP_i & \sumi\bfO_i'\bfC_*^{-1}\bfO_i
        \end{array}\right)-\frac{1}{I}\left(\begin{array}{c|c}
            \bfN'\bfC_*^{-1}\bfR\left(\bfR'\bfC_*^{-1}\bfR\right)^{-1}\bfR'\bfC_*^{-1}\bfN & \bfN'\bfC_*^{-1}\bfR\left(\bfR'\bfC_*^{-1}\bfR\right)^{-1}\bfR'\bfC_*^{-1}\bfM\\
            \hline
            \bfM'\bfC_*^{-1}\bfR\left(\bfR'\bfC_*^{-1}\bfR\right)^{-1}\bfR'\bfC_*^{-1}\bfN & \bfM'\bfC_*^{-1}\bfR\left(\bfR'\bfC_*^{-1}\bfR\right)^{-1}\bfR'\bfC_*^{-1}\bfM
        \end{array}\right)\right\}^{-1},
    \end{align*}
    where $\bfM = (\sumi\bfX_i\otimes \mathbf{1}_K)$, $\bfN = (\sumi\bfA_i\otimes \mathbf{1}_K)$, $\bfO_i = \bfX_i\otimes\mathbf{1}_K$, $\bfP_i = \bfA_i\otimes\mathbf{1}_K$, and $\bfR = \bfI_J\otimes\mathbf{1}_K$. We then further simplify the expression. Importantly, we define the (generalized version of) design constant as $\bfU_1 = \sumi \bfX_i'\mathbf{1}_J$ as a $(J-1)$ vector, $\bfW_1 = (\sumi \bfX_i')(\sumi \bfX_i)$ is a $(J-1) \times (J-1)$ matrix, $\bfW_2 = \sumi \bfX_i'\mathbf{1}_J\mathbf{1}_J'\bfX_i$ is a $(J-1)\times (J-1)$ matrix, and $\bfU_2 = \sumi \bfX_i'\bfX_i$ is a $(J-1)\times (J-1)$ matrix. Similarly, $U_3 = \sumi \bfA_i'\mathbf{1}_J$ is a scalar, $W_3 = (\sumi \bfA_i')(\sumi \bfA_i)$ is a scalar, $W_4 = \sumi \bfA_i'\mathbf{1}_J\mathbf{1}_J'\bfA_i$ is a scalar, and $U_4 = \sumi \bfA_i'\bfA_i$ is a scalar. Also, let $\bfW_5 = (\sumi \bfX_i')(\sumi \bfA_i)$, $\bfU_5 = \sumi\bfA_i'\bfX_i$, and $\bfW_6 = \sumi\bfA_i'\mathbf{1}_J\mathbf{1}_J'\bfX_i$.

    In Maleyeff et al., \cite{Maleyeff2022} we have $\bfC_* = \sigma^2_t\{(1-\rho)\bfI_{JK}+\rho\mathbf{1}_{JK}\mathbf{1}_{JK}'\} = \sigma^2_t\bfD$, where the correlation matrix $\bfD$ has two eigenvalues, $\lambda_1 = 1-\rho$ and $\lambda_2 = 1+(JK-1)\rho$. Thus, the inverse is given by
    \begin{align*}
        \bfD^{-1} = \frac{1}{\lambda_1}\bfI_{JK}-\frac{\lambda_2-\lambda_1}{JK\lambda_1\lambda_2}\mathbf{1}_{JK}\mathbf{1}_{JK}'= a\bfI_{JK}+b\mathbf{1}_{JK}\mathbf{1}_{JK}',
    \end{align*}
    where $a = 1/\lambda_1$ and $b = -(\lambda_2-\lambda_1)/(JK\lambda_1\lambda_2)$. Then, we observe
    \begin{align*}
        \bfR'\bfD^{-1}\bfR &= aK\bfI_J+bK^2\mathbf{1}_J\mathbf{1}_J',\displaybreak[0]\\
        \left(\bfR'\bfD^{-1}\bfR\right)^{-1} &= \frac{1}{K}\left(\lambda_1\bfI_J+\frac{\lambda_2-\lambda_1}{J}\mathbf{1}_J\mathbf{1}_J'\right),\displaybreak[0]\\
        \bfM'\bfD^{-1}\bfR &= aK\sumi \bfX_i' + bK^2\sumi \bfX_i'\mathbf{1}_J\mathbf{1}_J' = \frac{K}{\lambda_1}\sumi \bfX_i'-\frac{(\lambda_2-\lambda_1)K}{\lambda_1\lambda_2J}\bfU_1\mathbf{1}_J',\displaybreak[0]\\
        \bfR'\bfD^{-1}\bfM &= aK\sumi \bfX_i + bK^2\sumi \mathbf{1}_J\mathbf{1}_J'\bfX_i = \frac{K}{\lambda_1}\sumi \bfX_i-\frac{(\lambda_2-\lambda_1)K}{\lambda_1\lambda_2J}\left(\bfU_1\mathbf{1}_J'\right)',\displaybreak[0]\\
        \bfN'\bfD^{-1}\bfR &= \left(\sumi\bfA_i'\otimes \mathbf{1}_K'\right)\left(a\bfI_{JK}+b\mathbf{1}_{JK}\mathbf{1}_{JK}'\right)\left(\bfI_J\otimes\mathbf{1}_K\right)\displaybreak[0]\\
        &=aK\sumi\bfA_i'+bK^2\sumi \bfA_i'\mathbf{1}_J\mathbf{1}_J'=\frac{K}{\lambda_1}\sumi\bfA_i'-\frac{(\lambda_2-\lambda_1)K}{\lambda_1\lambda_2J}U_3\mathbf{1}_J',\displaybreak[0]\\
        \bfR'\bfD^{-1}\bfN &=aK\sumi\bfA_i+bK^2\sumi \mathbf{1}_J\mathbf{1}_J'\bfA_i=\frac{K}{\lambda_1}\sumi\bfA_i-\frac{(\lambda_2-\lambda_1)K}{\lambda_1\lambda_2J}\left(U_3\mathbf{1}_J'\right)'.
    \end{align*}
    Let $\bfU_1^{\otimes 2} = \bfU_1\bfU_1'$, $\bfU_{3,1} = U_3\bfU_1'$, and $\bfU_{3,1}' = U_3\bfU_1$. Then,
    \begin{align*}
        &\quad\frac{1}{K}\bfN'\bfD^{-1}\bfR\left(\bfR'\bfD^{-1}\bfR\right)^{-1}\bfR'\bfD^{-1}\bfN\displaybreak[0]\\
        &= \left(\frac{1}{\lambda_1}\sumi\bfA_i'-\frac{\lambda_2-\lambda_1}{\lambda_1\lambda_2J}U_3\mathbf{1}_J'\right)\left(\lambda_1\bfI_J+\frac{\lambda_2-\lambda_1}{J}\mathbf{1}_J\mathbf{1}_J'\right)\left(\frac{1}{\lambda_1}\sumi\bfA_i-\frac{\lambda_2-\lambda_1}{\lambda_1\lambda_2J}\left(U_3\mathbf{1}_J'\right)'\right)\displaybreak[0]\\
        &=\frac{1}{\lambda_1}W_3-\frac{\lambda_2-\lambda_1}{\lambda_1\lambda_2J}U_3^2+\frac{\lambda_2-\lambda_1}{\lambda_1^2J}U_3^2-\frac{(\lambda_2-\lambda_1)^2}{\lambda_1^2\lambda_2J}U_3^2-\frac{\lambda_2-\lambda_1}{\lambda_1\lambda_2J}U_3^2+\frac{(\lambda_2-\lambda_1)^2}{\lambda_1\lambda_2^2J}U_3^2-\frac{(\lambda_2-\lambda_1)^2}{\lambda_1^2\lambda_2J}U_3^2+\frac{(\lambda_2-\lambda_1)^3}{\lambda_1^2\lambda_2^2J}U_3^2\displaybreak[0]\\
        &=\frac{1}{\lambda_1}W_3-\frac{(\lambda_2 - \lambda_1)}{J \lambda_1 \lambda_2} U_3^2,\displaybreak[0]\\
        &\quad\frac{1}{K}\bfN'\bfD^{-1}\bfR\left(\bfR'\bfD^{-1}\bfR\right)^{-1}\bfR'\bfD^{-1}\bfM\displaybreak[0]\\
        &= \left(\frac{1}{\lambda_1}\sumi\bfA_i'-\frac{\lambda_2-\lambda_1}{\lambda_1\lambda_2J}U_3\mathbf{1}_J'\right)\left(\lambda_1\bfI_J+\frac{\lambda_2-\lambda_1}{J}\mathbf{1}_J\mathbf{1}_J'\right)\left(\frac{1}{\lambda_1}\sumi \bfX_i-\frac{\lambda_2-\lambda_1}{\lambda_1\lambda_2J}\left(\bfU_1\mathbf{1}_J'\right)'\right)\displaybreak[0]\\
        &=\frac{1}{\lambda_1}\bfW_5'-\frac{\lambda_2-\lambda_1}{\lambda_1\lambda_2J}\bfU_{3,1}+\frac{\lambda_2-\lambda_1}{\lambda_1^2J}\bfU_{3,1}-\frac{(\lambda_2-\lambda_1)^2}{\lambda_1^2\lambda_2J}\bfU_{3,1}-\frac{\lambda_2-\lambda_1}{\lambda_1\lambda_2J}\bfU_{3,1}+\frac{(\lambda_2-\lambda_1)^2}{\lambda_1\lambda_2^2J}\bfU_{3,1}-\frac{(\lambda_2-\lambda_1)^2}{\lambda_1^2\lambda_2J}\bfU_{3,1}+\frac{(\lambda_2-\lambda_1)^3}{\lambda_1^2\lambda_2^2J}\bfU_{3,1}\displaybreak[0]\\
        &=\frac{1}{\lambda_1}\bfW_5'-\frac{(\lambda_2 - \lambda_1)}{J \lambda_1 \lambda_2} \bfU_{3,1},\displaybreak[0]\\
        &\quad\frac{1}{K}\bfM'\bfD^{-1}\bfR\left(\bfR'\bfD^{-1}\bfR\right)^{-1}\bfR'\bfD^{-1}\bfN\displaybreak[0]\\
        &= \left(\frac{1}{\lambda_1}\sumi \bfX_i'-\frac{(\lambda_2-\lambda_1)K}{\lambda_1\lambda_2J}\bfU_1\mathbf{1}_J'\right)\left(\lambda_1\bfI_J+\frac{\lambda_2-\lambda_1}{J}\mathbf{1}_J\mathbf{1}_J'\right)\left(\frac{1}{\lambda_1}\sumi\bfA_i-\frac{(\lambda_2-\lambda_1)K}{\lambda_1\lambda_2J}\left(U_3\mathbf{1}_J'\right)'\right)\displaybreak[0]\\
        &=\frac{1}{\lambda_1}\bfW_5-\frac{\lambda_2-\lambda_1}{\lambda_1\lambda_2J}\bfU_{3,1}'+\frac{\lambda_2-\lambda_1}{\lambda_1^2J}\bfU_{3,1}'-\frac{(\lambda_2-\lambda_1)^2}{\lambda_1^2\lambda_2J}\bfU_{3,1}'-\frac{\lambda_2-\lambda_1}{\lambda_1\lambda_2J}\bfU_{3,1}'+\frac{(\lambda_2-\lambda_1)^2}{\lambda_1\lambda_2^2J}\bfU_{3,1}'-\frac{(\lambda_2-\lambda_1)^2}{\lambda_1^2\lambda_2J}\bfU_{3,1}'+\frac{(\lambda_2-\lambda_1)^3}{\lambda_1^2\lambda_2^2J}\bfU_{3,1}'\displaybreak[0]\\
        &=\frac{1}{\lambda_1}\bfW_5-\frac{(\lambda_2 - \lambda_1)}{J \lambda_1 \lambda_2} \bfU_{3,1}',\displaybreak[0]\\
        &\quad\frac{1}{K}\bfM'\bfD^{-1}\bfR\left(\bfR'\bfD^{-1}\bfR\right)^{-1}\bfR'\bfD^{-1}\bfM\displaybreak[0]\\
        &= \frac{1}{\lambda_1}\bfW_1-\frac{\lambda_2-\lambda_1}{\lambda_1\lambda_2J}\bfU_1^{\otimes 2},
    \end{align*}
    and further notice that
    \begin{align*}
        \bfP_i'\bfD^{-1}\bfP_i&= aK\bfA_i'\bfA_i+bK^2\bfA_i'\mathbf{1}_J\mathbf{1}_J'\bfA_i,\displaybreak[0]\\
        \bfP_i'\bfD^{-1}\bfO_i&= aK\bfA_i'\bfX_i+bK^2\bfA_i'\mathbf{1}_J\mathbf{1}_J'\bfX_i,\displaybreak[0]\\
        \bfO_i'\bfD^{-1}\bfP_i&= aK\bfX_i'\bfA_i+bK^2\bfX_i'\mathbf{1}_J\mathbf{1}_J'\bfA_i,\displaybreak[0]\\
        \bfO_i'\bfD^{-1}\bfO_i&= aK\bfX_i'\bfX_i+bK^2\bfX_i'\mathbf{1}_J\mathbf{1}_J'\bfX_i.
    \end{align*}
    Then, 
    \begin{align*}
        \sumi\bfP_i'\bfD^{-1}\bfP_i - \frac{1}{I}\bfN'\bfD^{-1}\bfR\left(\bfR'\bfD^{-1}\bfR\right)^{-1}\bfR'\bfD^{-1}\bfN&= \frac{K}{I\lambda_1}\left(IU_4-W_3\right)-\frac{K(\lambda_2-\lambda_1)}{IJ\lambda_1\lambda_2}\left(IW_4-U_3^2\right),\displaybreak[0]\\
        \sumi\bfP_i'\bfD^{-1}\bfO_i - \frac{1}{I}\bfN'\bfD^{-1}\bfR\left(\bfR'\bfD^{-1}\bfR\right)^{-1}\bfR'\bfD^{-1}\bfM&= \frac{K}{I\lambda_1}\left(I\bfU_5-\bfW_5'\right)-\frac{K(\lambda_2-\lambda_1)}{IJ\lambda_1\lambda_2}\left(I\bfW_6-\bfU_{3,1}\right),\displaybreak[0]\\
        \sumi\bfO_i'\bfD^{-1}\bfP_i - \frac{1}{I}\bfM'\bfD^{-1}\bfR\left(\bfR'\bfD^{-1}\bfR\right)^{-1}\bfR'\bfD^{-1}\bfN&= \frac{K}{I\lambda_1}\left(I\bfU_5'-\bfW_5\right)-\frac{K(\lambda_2-\lambda_1)}{IJ\lambda_1\lambda_2}\left(I\bfW_6'-\bfU_{3,1}'\right),\displaybreak[0]\\
        \sumi\bfO_i'\bfD^{-1}\bfO_i - \frac{1}{I}\bfM'\bfD^{-1}\bfR\left(\bfR'\bfD^{-1}\bfR\right)^{-1}\bfR'\bfD^{-1}\bfM&= \frac{K}{I\lambda_1}\left(I\bfU_2-\bfW_1\right)-\frac{K(\lambda_2-\lambda_1)}{IJ\lambda_1\lambda_2}\left(I\bfW_2-\bfU_1^{\otimes 2}\right).
    \end{align*}
    Thus, the covariance matrix for an anticipation effect estimator and point treatment effect estimators is
    \begin{align*}
        &\quad\bbV(\widehat\gamma^\ETIANT, \widehat\bfdelta^\ETIANT)\displaybreak[0]\\
        &= \sigma_t^2\left(\begin{array}{c|c}
            \frac{K}{I\lambda_1}\left(IU_4-W_3\right)-\frac{K(\lambda_2-\lambda_1)}{IJ\lambda_1\lambda_2}\left(IW_4-U_3^2\right) & \frac{K}{I\lambda_1}\left(I\bfU_5-\bfW_5'\right)-\frac{K(\lambda_2-\lambda_1)}{IJ\lambda_1\lambda_2}\left(I\bfW_6-\bfU_{3,1}\right)\\
            \hline
            \frac{K}{I\lambda_1}\left(I\bfU_5'-\bfW_5\right)-\frac{K(\lambda_2-\lambda_1)}{IJ\lambda_1\lambda_2}\left(I\bfW_6'-\bfU_{3,1}'\right) & \frac{K}{I\lambda_1}\left(I\bfU_2-\bfW_1\right)-\frac{K(\lambda_2-\lambda_1)}{IJ\lambda_1\lambda_2}\left(I\bfW_2-\bfU_1^{\otimes 2}\right)
        \end{array}\right)^{-1}\displaybreak[0]\\
        &= \sigma_t^2\frac{IJ\lambda_1\lambda_2}{K}\left(\begin{array}{c|c}
            \lambda_2\left(U_3^2+IJU_4-JW_3-IW_4\right) + \lambda_1\left(IW_4-U_3^2\right) & \lambda_2\left(\bfU_{3, 1}+IJ\bfU_5-J\bfW_5'-I\bfW_6\right) + \lambda_1\left(I\bfW_6-\bfU_{3, 1}\right)\\
            \hline
            \lambda_2(\bfU_{3, 1}'+IJ\bfU_5'-J\bfW_5-I\bfW_6') + \lambda_1(I\bfW_6'-\bfU_{3, 1}') & \lambda_2\left(\bfU_1^{\otimes 2}+IJ\bfU_2-J\bfW_1-I\bfW_2\right) + \lambda_1\left(I\bfW_2-\bfU_1^{\otimes 2}\right)
        \end{array}\right)^{-1}.
    \end{align*}
    Applying the formula of the generalized inverse of a partitioned matrix in Rohde \cite{Rohde1965} again,
    \begin{align*}
        &~\quad \bbV(\widehat\bfdelta^\ETIANT)\displaybreak[0]\\
        =&~ \sigma_t^2\frac{IJ\lambda_1\lambda_2}{K} \Bigr[\lambda_2\left(\bfU_1^{\otimes 2}+IJ\bfU_2-J\bfW_1-I\bfW_2\right) + \lambda_1\left(I\bfW_2-\bfU_1^{\otimes 2}\right)\displaybreak[0]\\
        &~\quad \left.-\left\{\lambda_2\left(U_3^2+IJU_4-JW_3-IW_4\right) + \lambda_1\left(IW_4-U_3^2\right)\right\}^{-1} \left\{\lambda_2(\bfU_{3, 1}'+IJ\bfU_5'-J\bfW_5-I\bfW_6') + \lambda_1(I\bfW_6'-\bfU_{3, 1}')\right\}^{\otimes 2}\right]^{-1}.
    \end{align*}
    Further notice that
    \begin{align*}
        U_3 = \sumi \bfA_i'\mathbf{1}_J = I, \quad
        U_4 = \sumi \bfA_i'\bfA_i = I, \quad
        \bfU_5 = \sumi\bfA_i'\bfX_i = \mathbf{0}_S',
    \end{align*}
    and
    \begin{align*}
        W_4 = \sumi \bfA_i'\mathbf{1}_J\mathbf{1}_J'\bfA_i = I, \quad
        \bfW_6 = \sumi\bfA_i'\mathbf{1}_J\mathbf{1}_J'\bfX_i = \bfU_1'.
    \end{align*}
    Thus,
    \begin{align*}
        \bbV(\widehat\bfdelta^\ETIANT) &= \sigma_t^2\frac{IJ\lambda_1\lambda_2}{K} \left[\lambda_2\left(\bfU_1^{\otimes 2}+IJ\bfU_2-J\bfW_1-I\bfW_2\right) + \lambda_1\left(I\bfW_2-\bfU_1^{\otimes 2}\right)-\left\{\lambda_2\left(I^2J-JW_3\right)\right\}^{-1} \left(\lambda_2J\bfW_5\right)^{\otimes 2}\right]^{-1}.
    \end{align*}
    Finally, the variance of the average treatment effect estimator is
    \begin{align*}
        \bbV(\widehat\Delta^\ETIANT) &=\bbV\left\{\frac{1}{J-1}\sums\widehat\delta^\ETIANT(s)\right\} \displaybreak[0]\\
        &= \frac{1}{(J-1)^2}\mathbf{1}_{J-1}'\bbV(\widehat\bfdelta^\ETIANT)\mathbf{1}_{J-1} \displaybreak[0]\\
        &= \frac{IJ\lambda_1\lambda_2\sigma_t^2}{K(J-1)^2} \mathbf{1}_{J-1}'\left\{\lambda_2\left(\bfU_1^{\otimes 2}+IJ\bfU_2-J\bfW_1-I\bfW_2 - \frac{J \bfW_5^{\otimes 2}}{I^2-W_3}\right) + \lambda_1\left(I\bfW_2-\bfU_1^{\otimes 2}\right)\right\}^{-1}\mathbf{1}_{J-1}.
    \end{align*}
    We then proceed to derive $\bbV(\widehat\delta^\HHANT)$. Of note, $\bbV\{(J-1)^{-1}\sums\widehat\delta^\ETIANT(s)\}$ can be reduced to $\bbV(\widehat\delta^\HHANT)$ by setting $J=2$ and replace $\bfX_i$ with $\bfZ_i=(Z_{i1},\ldots,Z_{iJ})'$. Note that
    \begin{align*}
        \bfU_1 &= \sumi \bfX_i'\mathbf{1}_J \implies \sumi \bfZ_i'\mathbf{1}_J = \sumi \sumj Z_{ij} = U,\displaybreak[0]\\
        \bfU_2 &= \sumi \bfX_i'\bfX_i \implies \sumi \bfZ_i'\bfZ_i = \sumi \sumj Z_{ij} = U,\displaybreak[0]\\
        \bfW_1 &= \left(\sumi \bfX_i'\right)\left(\sumi \bfX_i\right) \implies \left(\sumi \bfZ_i'\right)\left(\sumi \bfZ_i\right) = \sumj \left(\sumi Z_{ij}\right)^2 = W_1,\displaybreak[0]\\
        \bfW_2 &= \sumi \bfX_i'\mathbf{1}_J\mathbf{1}_J'\bfX_i \implies \sumi \bfZ_i'\mathbf{1}_J\mathbf{1}_J'\bfZ_i = \sumi \left(\sumj Z_{ij}\right)^2 = W_2,\displaybreak[0]\\
        \bfW_5 &= \left(\sumi \bfX_i'\right)\left(\sumi \bfA_i\right) \implies \left(\sumi \bfZ_i'\right)\left(\sumi \bfA_i\right) = W_5.
    \end{align*}
    By plugging above terms into $\bbV\{(J-1)^{-1}\sums\widehat\delta^\ETIANT(s)\}$, we get
    \begin{align*}
        \bbV(\widehat\delta^\HHANT)= \frac{IJ\lambda_1\lambda_2\sigma_t^2}{K} \left\{\lambda_2\left(U^2+IJU-JW_1-IW_2-\frac{JW_5^2}{I^2-W_3}\right) + \lambda_1\left(IW_2-U^2\right)\right\}^{-1}.
    \end{align*}
    
\end{proof}

\subsection*{B.2 Standard SW-CRTs} \label{sec::variance_standard}

In this section, we compare $\bbV(\widehat\delta^\HH)$ with $\bbV(\widehat\delta^\HHANT)$ in a standard SW-CRT. Specifically, we assume $I$ is divisible by $Q$ with $Q = J-1$, and sequence $q$ has $I_q = I/Q$ clusters for $q \in \{1,\ldots,Q\}$. The treatment adoption time for clusters in sequence $q$ is at period $j=q+1$. Furthermore, we assume the number of individuals per cluster at each period $K_{ij}=K$ for all $i$ and $j$.

\begin{theorem} \label{theorem::var_std}
    \emph{For a standard SW-CRT with $I$ clusters, $J$ periods, and $K$ individuals per cluster-period, variances of treatment effect estimators are
    \begin{align} \label{var::HH_std}
        \bbV(\widehat\delta^\HH)= \frac{12 Q \sigma_t^2 \lambda_1 \lambda_2}{I K (Q - 1)\{Q\lambda_1+(Q+2)\lambda_2\}}
    \end{align}
    for the HH model estimator,
    \begin{align} \label{var::HH-ANT_std}
        \bbV(\widehat\delta^\HHANT)= \frac{12 Q \sigma_t^2 \lambda_1 \lambda_2}{I K (Q - 1)\{Q\lambda_1+(Q-1)\lambda_2\}}
    \end{align}
    for the HH-ANT model estimator.}
\end{theorem}

\begin{proof}

    We first derive the closed-form expression of $\bbV(\widehat\delta^\HH)$ in a standard SW-CRT. To facilitate the derivation, we can re-parameterize \eqref{model::HH} by dropping the grand mean, and then rewrite it in terms of cluster-period means, such that the model becomes
    \begin{align}
        \overline Y_{ij}(q) = \beta_j + \delta^\HH Z_{ij} + \alpha_i + \overline\epsilon_{ij}, \label{model::HH_re}
    \end{align}
    for the HH model. For \eqref{model::HH_re}, let $\bfH_q = (\bfI_J, \bfZ_q^\seq)$ be the $J \times (J+1)$ design matrix for a cluster assigned to sequence $q$. Based on standard results in mixed model theory, the weighted least squares estimator of the fixed effects $(\widehat\beta_1, \ldots, \widehat\beta_J, \widehat\delta^\HH)'$ can be expressed as:\cite{Matthews2017}
    \begin{align*}
        \left(\widehat\beta_1, \ldots, \widehat\beta_J, \widehat\delta^\HH\right)' &= \left(\sumq \bfH_q'\bfV^{-1}\bfH_q\right)^{-1}\sumq \bfH_q'\bfV^{-1}\overline \bfY(q) \displaybreak[0]\\
        &= \left(\sumq\sum_{n=1}^{I/Q} \bfH_q'\bfV^{-1}\bfH_q\right)^{-1}\sumq \bfH_q'\bfV^{-1} \sum_{n=1}^{I/Q} \overline \bfY_{n+(q-1)I/Q}(q).
    \end{align*}
    Of note, $\bbV(\widehat\delta^\HH)$ correspond to the $(J+1, J+1)$-th entry of $(\sumq\sum_{n=1}^{I/Q} \bfH_q'\bfV^{-1}\bfH_q)^{-1} = Q/I(\sumq \bfH_q'\bfV^{-1}\bfH_q)^{-1}$. We derive the closed-form expressions of $(\sumq \bfH_q'\bfV^{-1}\bfH_q)^{-1}$ by applying the formula of the generalized inverse of a partitioned matrix in Rohde, \cite{Rohde1965} which gives
    \begin{align*}
        \left(\sumq \bfH_q'\bfV^{-1}\bfH_q\right)^{-1} &= \left( \begin{array}{cc}
             Qx\bfI_J-Qy\mathbf{1}_J\mathbf{1}_J' &  x\bfkappa_1-\kappa_2y \mathbf{1}_J\\
             x\bfkappa_1'-\kappa_2y \mathbf{1}_J' & \kappa_2x- \frac{Q(Q + 1)(2Q + 1)}{6}y
        \end{array} \right)^{-1} = \left( \begin{array}{cc}
             \bfDelta_{11} &  \bfDelta_{12}\\
             \bfDelta_{12}' & \bfDelta_{22}
        \end{array} \right)^{-1} \displaybreak[0]\\
        &= \left(\begin{array}{cc}
             \bfDelta_{11}^{-1}+\bfDelta_{11}^{-1}\bfDelta_{12}\bfDelta^{-1}\bfDelta_{12}'\bfDelta_{11}^{-1} & -\bfDelta_{11}^{-1}\bfDelta_{12}\bfDelta^{-1}\\
            -\bfDelta^{-1}\bfDelta_{12}'\bfDelta_{11}^{-1} & \bfDelta^{-1}
        \end{array}\right),
    \end{align*}
    where
    \begin{align*}
        &\bfDelta_{11}^{-1} = \frac{x-Jy}{Qx(x-Jy)}\bfI_J+\frac{y}{Qx(x-Jy)}\mathbf{1}_J\mathbf{1}_J',\displaybreak[0]\\
        &\bfDelta_{11}^{-1}\bfDelta_{12} = \left(0, \frac{1}{Q}, \frac{2}{Q},\ldots,1\right)',\displaybreak[0]\\
        &\bfDelta_{12}'\bfDelta_{11}^{-1}\bfDelta_{12} = \frac{(Q+1)(2Q+1)}{6} x - \frac{Q+1}{2} \kappa_2  y,\displaybreak[0]\\
        &\bfDelta = \bfDelta_{22}-\bfDelta_{12}'\bfDelta_{11}^{-1}\bfDelta_{12} = \frac{(Q+1)(Q-1)}{6} x - \frac{Q(Q+1)(Q-1)}{12} y.
    \end{align*}
    Thus,
    \begin{align*}
        \bbV(\widehat\delta^\HH) &= \frac{Q}{I}\bfDelta^{-1} \displaybreak[0]\\
        &= \frac{Q}{I}\left\{\frac{(Q+1)(Q-1)}{6} x - \frac{Q(Q+1)(Q-1)}{12} y\right\}^{-1} \displaybreak[0]\\
        &= \frac{Q}{I}\left\{\frac{(Q+1)(Q-1)}{6} \frac{1}{\eta^2(1-\phi)} - \frac{Q(Q+1)(Q-1)}{12} \frac{\phi}{\eta^2(1-\phi)(1+\phi J-\phi)}\right\}^{-1} 
        \displaybreak[0]\\
        &=\frac{12Q\tau^2(1 - \phi)(1+\phi Q )}{\phi I(Q - 1)(Q + 1) (2+\phi Q)} 
        \displaybreak[0]\\
        &=\frac{12 Q \sigma_t^2 \lambda_1 \lambda_2}{I K (Q - 1)\{Q\lambda_1+(Q+2)\lambda_2\}}. 
    \end{align*}
    Then, we derive the closed-form expression of $\bbV(\widehat\delta^\HHANT)$ in a standard SW-CRT. Similarly, $\bbV(\widehat\delta^\HHANT)$ corresponds to $(J+2, J+2)$-th entry of $(\sumq\sum_{n=1}^{I/Q} \bfE_q'\bfV^{-1}\bfE_q)^{-1} = Q/I(\sumq \bfE_q'\bfV^{-1}\bfE_q)^{-1}$, which is simply $Q/I\bfOmega^{-1}$ in A.3.1. Thus, we have
    \begin{align*}
        \bbV(\widehat\delta^\HHANT) &= \frac{Q}{I}\bfOmega^{-1} \displaybreak[0]\\
        &= \frac{Q}{I}\frac{12}{x - 3xQ + yQ + 2xQ^2 - yQ^3} \displaybreak[0]\\
        &= \frac{12Q\tau^2(1-\phi)(1 + \phi Q)}{\phi I(Q - 1)(\phi Q^2 - 2\phi Q + 2Q - 1)} 
        \displaybreak[0]\\
        &= \frac{12 Q \sigma_t^2 \lambda_1 \lambda_2}{I K (Q - 1)\{Q\lambda_1+(Q-1)\lambda_2\}}. 
    \end{align*}
    
\end{proof}

\section*{Web Appendix C: Additional information on trial planning} \label{sec::trial_planning}

In this section, we demonstrate the planning of an SW-CRT with or without the anticipation effect. Of note, for working models \eqref{model::HH}-\eqref{model::ETI-ANT}, their corresponding analytical variance expressions \eqref{var::HH_general}-\eqref{var::ETI-ANT_general} are independent of the true values of the point treatment effects. Instead, they only depend on the trial resources (number of periods, clusters, individuals), treatment assignment (design constant), and the outcome ICC, which are either pre-specified during the trial planning stage or known based on prior knowledge.

Consider a standard SW-CRT with $J = 7$ periods, $Q = 6$ exposure times, and $I = 18$ clusters. Suppose our target outcome $Y_{ijk}$ is continuous and is based on the model choice from \eqref{model::HH}-\eqref{model::ETI-ANT} depending on research protocols. In this hypothetical trial, Theorem \ref{theorem::var} can be used for power analyses. We provide an \texttt{R} function \texttt{Trt.Ant.Power} in the Supporting Information for practical applications. This function first calculates the analytical variance from Theorem \eqref{theorem::var} using input parameters in Web Table \ref{tab::power_function}. It then assumes a two-tailed Wald test and calculates the power of the trial to detect a pre-specified (average) treatment effect.

\begin{table}[htbp]
    \centering
    \caption{Descriptions of the input parameters in an R function \textit{Trt.Ant.Power}.}
    \label{tab::power_function}
    \begin{threeparttable}
    \begin{tabular}{ll}
    \toprule
    \textbf{Input parameter} & \textbf{Description}\\
    \hline
    \texttt{model} & Choose one model from HH, HH-ANT, ETI, ETI-ANT \\
    \texttt{trt} & $\delta$, treatment effect for HH and HH-ANT; $\Delta$, average treatment effect for ETI and ETI-ANT \\
    \texttt{I} & $I$, number of clusters \\
    \texttt{J} & $J$, number of periods \\
    \texttt{K} & $K$, number of individuals per cluster-period (assuming an equal cluster-period size) \\
    \texttt{design} & $(\bfZ_1', \ldots, \bfZ_I')$, an $IJ$-dimensional treatment indicator vector for cluster $i\in\{1,\ldots,I\}$ \\
    \texttt{rho} & $\rho$, outcome ICC $\rho = \tau^2/(\tau^2+\sigma^2)$ \\
    \texttt{sigma\_sq} & $\sigma^2$, individual-level random error \\
    \texttt{alpha} & $\alpha$, type I error rate (default value is $0.05$) \\
    \bottomrule
    \end{tabular}\smallskip
    \end{threeparttable}
\end{table}

If we further fix $\rho = 0.1$, $\sigma^2 = 1$, $K = 50$, and decide to use an ETI-ANT working model, then the following codes return the study power for detecting an average treatment effect of size $0.1$ with a significance level of $0.05$:

\begin{verbatim}
    Trt.Ant.Power(model = ETI-ANT, I = 18, J = 7, K = 50, design = design, rho = 0.1, 
        sigma_sq = 1, alpha = 0.05)
\end{verbatim}

In addition, if the outcome ICC is unknown, the user can explore a variety of parameter values. For instance, Web Table \ref{tab::power_ETI-ANT} provides the minimum detectable average treatment effect size with varying outcome ICCs for our hypothetical SW-CRT if the target study power is $80\%$ using ETI-ANT at a significance level of $0.05$.

\section*{Web Appendix D: Additional information on power comparison} \label{sec::power_comparison}

In this section, we present two numerical studies to compare the analytic power of 1) the HH model versus the HH-ANT model (see Web Appendix D.1), and 2) the ETI model versus the ETI-ANT model (see Web Appendix D.2), in the presence of an anticipation effect. We follow the setting in Web Appendix A by considering a standard SW-CRT with $I$ clusters, $J$ periods, and equal cluster-period sizes $K_{ij}=K$. Also, we let $\lambda_1 = 1-\rho$ and $\lambda_2 = 1+(JK-1)\rho$, with the total variance $\sigma_t^2 = \tau^2+\sigma^2$ and the ICC $\rho = \tau^2/\sigma_t^2$. At the significance level of $\alpha$, the power of a two-tailed Wald test is expressed as:
\begin{align}
    \text{Power} = \Phi \left\{\frac{|\delta^*|}{\sqrt{\bbV(\widehat\delta)}}-z_{1-\alpha/2}\right\},\label{power_formula}
\end{align}
where $\delta^*$ is the treatment effect under the alternative hypothesis, $\widehat\delta$ is the estimated treatment effect, $\Phi$ is the cumulative standard Normal distribution function, and $z_{1-\alpha/2}$ is the $(1-\alpha/2)$-th quantile of the standard Normal distribution function. In a similar fashion, the power formula for the TATE can be retrieved by replacing $(\delta^*, \widehat\delta)$ in \eqref{power_formula} with $(\Delta^*, \widehat\Delta)$, where $\Delta^*$ is the true time-average treatment effect under the alternative hypothesis, and $\widehat\Delta$ is the estimated time-average treatment effect.

\subsubsection*{D.1 True Model: HH-ANT}

Suppose the true data-generating process is HH-ANT in \eqref{model::HH-ANT_re}, where $\delta^\HHANT$ is the true treatment effect and $\gamma^\HHANT$ is the true anticipation effect. The working models are HH and HH-ANT. Based on Theorem \ref{theorem::HH_1} (bias formula) and Theorem \ref{theorem::var_std} (variance formula), we can write
\begin{align*}
    \delta^{*, \HH} = \delta^\HHANT -\frac{6(1+\phi Q)}{(Q+1)(2+\phi Q)} \gamma^\HHANT \quad&\text{and}\quad
    \delta^{*, \HHANT} = \delta^\HHANT, \tag{from Theorem \ref{theorem::HH_1}}\displaybreak[0]\\
    \bbV(\widehat\delta^\HH) = \frac{12 Q \sigma_t^2 \lambda_1 \lambda_2}{I K (Q - 1)\{Q\lambda_1+(Q+2)\lambda_2\}} \quad&\text{and}\quad
    \bbV(\widehat\delta^\HHANT) = \frac{12 Q \sigma_t^2 \lambda_1 \lambda_2}{I K (Q - 1)\{Q\lambda_1+(Q-1)\lambda_2\}},\tag{from Theorem \ref{theorem::var_std}}
\end{align*}
which indicate that the power formula in \eqref{power_formula} generally depends on four factors: 1) magnitude of effects ($|\delta^\HHANT|$, $|\gamma^\HHANT|$), 2) trial resources ($I$, $J$, $K$), 3) ICC-related parameters ($\tau^2$, $\sigma^2$), and 4) choice of working models (HH or HH-ANT). 

To focus ideas, we fix $I = 32$, $K = 100$, $J \in \{5, 9, 17, 33\}$, and $\sigma^2 = 1$. The goal is to compare the analytic power of the treatment effect estimator $\widehat\delta^\HH$ from the HH model and the treatment effect estimator $\widehat\delta^\HHANT$ from the HH-ANT model in terms of $\gamma^\HHANT$, $\delta^\HHANT$, and $\rho$ (by varying $\tau^2$) at a significant level of $\alpha = 0.05$. In the first scenario, we fix $\delta^\HHANT \in \{0.01, 0.04\}$ and vary $\gamma^\HHANT/\delta^\HHANT \in [0, 1]$ and $\rho \in [0, 0.25]$ in Web Figure \ref{figure_power_ratio_delta0.01} and \ref{figure_power_ratio_delta0.04}. Interestingly, the relative power of these two models critically depends on the ICC $\rho$ and relative magnitude of the anticipation effect (to the treatment effect) $\gamma^\HHANT/\delta^\HHANT$. For example, in Web Figure 3 (a) under the scenario where $J = 5$ and $\delta^\HHANT = 0.1$, over a range of ICC parameters from 0 to 0.25, the test based on the HH-ANT model has higher power when the anticipation effect is relatively large (e.g., $\gamma^\HHANT/\delta^\HHANT>0.295$), whereas the test based on the HH model is more powerful otherwise (e.g., $\gamma^\HHANT/\delta^\HHANT\leq0.295$).

In the second scenario, we fix $\gamma^\HHANT/\delta^\HHANT \in \{0.2, 0.3, 0.4\}$ and vary $\delta^\HHANT \in [0, 0.1]$ and $\rho \in [0, 0.25]$ in Web Figure \ref{figure_power_ratio_0.2}-\ref{figure_power_ratio_0.4}. Of note, the absolute value of the difference in power between the two models generally increases as $\rho$ increases or $J$ decreases. For example, under the scenario where $J = 9$ and $\gamma^\HHANT/\delta^\HHANT = 0.4$, Web Figure \ref{figure_power_ratio_0.4} (b) shows that the test based on the HH-ANT model is at most $15\%$ more powerful than the test based on the HH model when $0.06 \leq \delta^\HHANT \leq 0.07$ and $\rho > 0.04$. In contrast, under a similar scenario where $J = 5$ and holding all other conditions the same, Web Figure \ref{figure_power_ratio_0.4} (a) indicates that the test based on the HH-ANT model is at least $30\%$ more powerful than the test based on the HH model when $0.06 \leq \delta^\HHANT \leq 0.07$ and $\rho > 0.04$.

\subsubsection*{D.2 True Model: ETI-ANT}

Suppose the true data generating process is ETI-ANT in \eqref{model::ETI-ANT_re}, where $\delta^\ETIANT(s_{ij})$ is the true point treatment effect at exposure time $s_{ij}$ observed in cluster-period $(i,j)$ and $\gamma^\ETIANT$ is the true anticipation effect. The estimand of interests is the TATE, $\Delta^\ETIANT = (J-1)^{-1}\sums\delta^\ETIANT(s)$, and the working models are ETI and ETI-ANT. Based on standard results in mixed model theory, the weighted least squares estimator of the fixed effects $(\beta_1, \ldots, \beta_J, \delta^\ETI(1), \ldots, \delta^\ETI(J-1))'$ under the ETI model can be expressed as:\cite{Matthews2017}
\begin{align}
    \left(\widehat\beta_1, \ldots, \widehat\beta_J, \widehat\delta^\ETI(1), \ldots, \widehat\delta^\ETI(J-1)\right)' = \left(\sumq \bfF_q'\bfV^{-1}\bfF_q\right)^{-1}\sumq \bfF_q'\bfV^{-1}\overline \bfY(q). \label{est::ETI_general}
\end{align}
We proceed to derive the closed-form expression of $\bbE(\widehat\Delta^\ETI) = (J-1)^{-1}\sums \bbE\{\widehat\delta^\ETI(s)\}$ in \eqref{est::ETI_general} when the true model is ETI-ANT. Let $\bfkappa_3$ be a $Q$-dimensional vector giving the number of times treatment $1$ appears in each row of a row-column design $\bfZ$ from Web Appendix A.3.1. For $s \in \{1, \ldots, J-1\}$, we can write
\begin{align*}
    \bbE\{\widehat\delta^\ETI(s)\} &= \left(\begin{array}{c}
         \mathbf{0}_{J+s-1}  \\
         1 \\
         \mathbf{0}_{J-s-1}
    \end{array}\right)'\left(\sumq \bfF_q'\bfV^{-1}\bfF_q\right)^{-1}\sumq \bfF_q'\bfV^{-1} \bbE\{\overline\bfY(q)\} \displaybreak[0]\\
    &= \left(\begin{array}{c}
         \mathbf{0}_{J+s-1}  \\
         1 \\
         \mathbf{0}_{J-s-1}
    \end{array}\right)' \left(\sumq \bfF_q'\bfV^{-1}\bfF_q\right)^{-1}\sumq \bfF_q'\bfV^{-1} \left(\begin{array}{c}
         \bbE\{\overline Y_1(q)\} \\[5pt]
         \vdots \\[5pt]
         \bbE\{\overline Y_J(q)\}
    \end{array}\right) 
    \displaybreak[0]\\
    &= \left(\begin{array}{c}
         \mathbf{0}_{J+s-1}  \\
         1 \\
         \mathbf{0}_{J-s-1}
    \end{array}\right)' \left(\sumq \bfF_q'\bfV^{-1}\bfF_q\right)^{-1}\sumq \bfF_q'\bfV^{-1} \left(\begin{array}{c}
         \mu + \beta_1 + \bbone(1 = q) \gamma^\ETIANT + \bbone(1 = q + s) \delta^\ETIANT(s) \\[5pt]
         \vdots \\[5pt]
         \mu + \beta_J + \bbone(J = q) \gamma^\ETIANT + \bbone(J = q + s) \delta^\ETIANT(s)
    \end{array}\right) 
    \displaybreak[0]\\
    &= \left(\begin{array}{c}
         \mathbf{0}_{J+s-1}  \\
         1 \\
         \mathbf{0}_{J-s-1}
    \end{array}\right)' \left(\begin{array}{c|c}
         Qx\bfI_J-Qy\mathbf{1}_J\mathbf{1}_J' &  x\bfZ' -y\mathbf{1}_J\bfkappa_3'\\[5pt]
         \hline
         x\bfZ-y\bfkappa_3\mathbf{1}_J'& \bfXi
    \end{array}\right)^{-1}\left(\begin{array}{c}
        x\sumq\bbE\{\overline Y_1(q)\} - y\sumq\sumj \bbE\{\overline Y_j(q)\}\\[5pt]
        \vdots\\[5pt]
        x\sumq\bbE\{\overline Y_J(q)\} - y\sumq\sumj \bbE\{\overline Y_j(q)\}\\[5pt]
        x\sumq\sumj \bbone(j=q+1)\bbE\{\overline Y_j(q)\}-y\sumq\sumj\bbE\{\overline Y_j(q)\}\\[5pt]
        \vdots\\[5pt]
        x\sumq\sumj \bbone(j=q+Q)\bbE\{\overline Y_j(q)\}-y\sum_{q=1}^1\sumj\bbE\{\overline Y_j(q)\}
    \end{array}\right).
\end{align*}
Here, we denote $(\zeta, \zeta') \in \{1,\ldots,Q\}$ as the row and column index of $\bfXi$, respectively. Then, $\bfXi$ can be re-expressed as:
\begin{align*}
    \bfXi_{\zeta, \zeta'} = \begin{cases}
        (Q+1-\zeta)(x-y) &\text{if } \zeta = \zeta',\\
        -\{Q+1-\max(\zeta, \zeta')\}y &\text{if } \zeta \neq \zeta'.
    \end{cases}
\end{align*}
Recall that $\bbE\{\overline Y_j(q)\} = \mu + \beta_j + \bbone(j = q) \gamma^\ETIANT + \bbone(j = q + s) \delta^\ETIANT(s)$, we can simply set $\mu = 0$ and $\beta_j = 0$ for $j \in \{1, \ldots, J\}$ without loss of generality. Given $\Delta^{*,\ETI} = \bbE(\widehat\Delta^\ETI)$ and $\Delta^{*,\ETIANT} = \Delta^\ETIANT$ with $\bbV(\widehat\Delta^\ETI)$ in \eqref{var::ETI_general} and $\bbV(\widehat\Delta^\ETIANT)$ in \eqref{var::ETI-ANT_general}, the power formula in \eqref{power_formula} generally depends on four factors: 1) magnitude of effects ($|\Delta^\ETIANT|$, $|\gamma^\ETIANT|$), 2) trial resources ($I$, $J$, $K$), 3) ICC-related parameters ($\tau^2$, $\sigma^2$), and 4) choice of working models (ETI or ETI-ANT).

To provide some basic explorations, we fix $I = 32$, $K = 100$, $J \in \{5, 9\}$, $\sigma^2 = 1$, and $\delta^\ETIANT(s) = -0.5\sin\{2\pi(s-1)/(J-2)\}+\log(1.2)$. The goal is to compare the analytic power of the treatment effect estimator $\widehat\Delta^\ETI$ from the ETI model and the treatment effect estimator $\widehat\Delta^\ETIANT$ from the ETI-ANT model in terms of $\gamma^\ETIANT$, $\delta^\ETIANT$, and $\rho$ (by varying $\tau^2$) at a significant level of $\alpha = 0.05$. In the first scenario, we fix $\delta^\ETIANT \in \{0.01, 0.04,0.07\}$ and vary $\gamma^\ETIANT/\delta^\ETIANT \in [0, 1]$ and $\rho \in [0, 0.25]$ in Web Figure \ref{figure_power_ratio_fixed_Delta}. In the second scenario, we fix $\gamma^\ETIANT/\delta^\ETIANT \in \{0.2, 0.3, 0.4\}$ and vary $\delta^\ETIANT \in [0, 0.2]$ and $\rho \in [0, 0.25]$ in Web Figure \ref{figure_power_ratio_fixed_ratio}. We omit discussions for these two scenarios because the comparative findings on analytical power between the ETI model and the ETI-ANT model are similar to those for the HH model and the HH-ANT model in Web Appendix D.1.

\section*{Web Appendix E: Additional information on higher-order anticipation effects}

In this section, we consider a standard SW-CRT with a higher-order anticipation effect, i.e., the anticipation effect can occur more than one period before the intervention adoption. Recall that the anticipation effect indicator $A_{ij}$ is formally defined as
\begin{align*}
    A_{ij} = \begin{cases}
                1, & j^*-\ell \leq j < j^*,\\
                0, & \text{otherwise,}
             \end{cases}
\end{align*}
where $j^*\in\{1,\ldots,J\}$ is the treatment adoption period (the period index with $Z_{ij^*}=1$ and $Z_{i,j^*-1}=0$ for cluster $i$), and $\ell$ is the order of the anticipation effect. For instance, when $\ell = 2$, we assume a second-order anticipation effect, i.e., the anticipation effect occurs two periods before the intervention adoption. To start with, we consider the scenario where the working model is HH, which does not account for the anticipation effect or exposure-time heterogeneity, and examine the behaviors of the HH model estimator for the treatment effect when the true model is HH-ANT with an $\ell$th-order anticipation effect. We proceed to use \eqref{est::HH} in Web Appendix A.1.1. to derive its expectation $\bbE(\widehat\delta^\HH)$ when the true model is HH-ANT with an $\ell$th-order anticipation effect.

\begin{theorem}\label{theorem::HH_ell}
    \emph{For a standard SW-CRT with data generated from the HH-ANT model with an $\ell$th-order anticipation effect, assuming $1 \leq \ell \leq Q$, the expectation of the HH model estimator $\widehat\delta^\HH$ is
    \begin{align*}
        \bbE(\widehat\delta^\HH) = \delta^\HHANT+\omega_\HH^{\HHANT,\ell} \gamma^\HHANT,
    \end{align*}
    where $\delta^\HHANT$ is the true treatment effect and $\gamma^\HHANT$ is the true anticipation effect. The weight of $\gamma^\HHANT$ is given by
    \begin{align*}
        \omega_\HH^{\HHANT,\ell} = -\frac{\ell (6 \phi Q^3 - 9 \phi \ell Q^2 + 3 \phi Q^2 + 6 Q^2 + 4 \phi \ell^2 Q - 3 \phi \ell Q - 6 \ell Q - \phi Q + 2 \ell^2 - 2)}{Q(Q + 1)(\phi Q^2 + 2Q - \phi Q - 2)}.
    \end{align*}}
\end{theorem}

\begin{proof}
    When the true model is HH-ANT with an $\ell$th-order anticipation effect, the expectation of the HH model estimator $\widehat\delta^\HH$ is
    \begin{align*}
        \bbE(\widehat\delta^\HH) &= T_1 \sumq \sumj \left\{ Q \bbone(j > q) - j + 1 + \frac{\phi Q(2q - Q - 1)}{2(1+\phi Q)} \right\} \bbE\{\overline Y_j(q)\} \displaybreak[0]\\
        &= T_1 \sumq \sumj \left\{ Q \bbone(j > q) - j + 1 + \frac{\phi Q(2q - Q - 1)}{2(1 + \phi Q)} \right\} \left[\mu + \beta_j + \bbone\{\max(1, q-\ell+1) \leq j \leq q\} \gamma^\HHANT + \bbone(j > q) \delta^\HHANT \right] 
        \displaybreak[0]\\
        &= T_1 \sumq \sumj \left\{ Q \bbone(j > q) - j + 1 + \frac{\phi Q(2q - Q - 1)}{2(1 + \phi Q)} \right\} [\bbone\{\max(1, q-\ell+1) \leq j \leq q\} \gamma^\HHANT + \bbone(j > q) \delta^\HHANT ]. 
    \end{align*}
    Recall that the terms related to $\delta^\HHANT$ are equal to $\delta^\HHANT$ as shown in Web Appendix A.1.2. Then, assuming $1 \leq \ell \leq Q$, we simplify terms related to $\gamma^\HHANT$,
    \begin{align*}
        &~T_1 \sumq \sumj \left\{ Q \bbone(j > q) - j + 1 + \frac{\phi Q(2q - Q - 1)}{2(1 + \phi Q)} \right\} \bbone\{\max(1, q-\ell+1) \leq j \leq q\} \gamma^\HHANT \displaybreak[0]\\
        =&~T_1 \sumq \sumj \left\{- j + 1 + \frac{\phi Q(2q - Q - 1)}{2(1 + \phi Q)} \right\} \bbone\{\max(1, q-\ell+1) \leq j \leq q\} \gamma^\HHANT 
        \displaybreak[0]\\
        =&~T_1 \left\{-\frac{\ell(3Q^2 - 3\ell Q + 6Q + \ell^2 - 3\ell + 2)}{6} + \frac{\ell(2Q - \ell + 1)}{2}\right. \displaybreak[0]\\
        &\qquad\quad \left. + \frac{\phi Q}{1 + \phi Q} \frac{\ell (3Q^2 + 3Q - \ell^2 + 1)}{6} -\frac{\phi Q(Q + 1)}{2(1 + \phi Q)} \frac{\ell(2Q - \ell + 1)}{2} \right\} \gamma^\HHANT \displaybreak[0]\\
        =&~T_1 \frac{- 6 \phi \ell Q^3  + 9 \phi \ell^2 Q^2 - 3 \phi \ell Q^2  - 6 \ell Q^2  - 4 \phi \ell^3 Q  + 3 \phi \ell^2 Q  + 6 \ell^2 Q + \phi \ell Q - 2 \ell^3 + 2 \ell}{12 (\phi Q + 1)} \gamma^\HHANT \displaybreak[0]\\
        =&~-\frac{\ell (6 \phi Q^3 - 9 \phi \ell Q^2 + 3 \phi Q^2 + 6 Q^2 + 4 \phi \ell^2 Q - 3 \phi \ell Q - 6 \ell Q - \phi Q + 2 \ell^2 - 2)}{Q(Q + 1)(\phi Q^2 + 2Q - \phi Q - 2)} \gamma^\HHANT.
    \end{align*}
    Therefore,
    \begin{align*}
        \bbE(\widehat\delta^\HH) = \delta^\HHANT+\omega_\HH^{\HHANT,\ell} \gamma^\HHANT,
    \end{align*}
    where
    \begin{align*}
        \omega_\HH^{\HHANT,\ell} = -\frac{\ell (6 \phi Q^3 - 9 \phi \ell Q^2 + 3 \phi Q^2 + 6 Q^2 + 4 \phi \ell^2 Q - 3 \phi \ell Q - 6 \ell Q - \phi Q + 2 \ell^2 - 2)}{Q(Q + 1)(\phi Q^2 + 2Q - \phi Q - 2)}.
    \end{align*}
\end{proof}

The weight of $\gamma^\HHANT$, $\omega_\HH^{\HHANT,\ell}$, is negative for $\ell \in \{1, \ldots, Q-1\}$ and $\phi \in [0, 1]$ in Web Figure \ref{figure_coeff_HH_higher_order}. Of note, when $\ell = Q$, it can be easily shown that $\omega_\HH^{\HHANT,\ell} = -1$. Clearly, $\omega_\HH^{\HHANT,\ell}$ becomes more negative as $\ell$ increases and $\phi \to 1$. Three numerical examples illustrating the bias is given in Web Figure \ref{figure_HH_higher_order} (a)-(c). In those Figures, we assume a constant $\delta^\HHANT$ and a $\gamma^\HHANT > 0$ with an order of $\ell = \{1, 2, 3\}$ and present both the true effect curve and the estimated effect curve under the HH working model; clearly, $\widehat\delta^\HH$ underestimates $\delta^\HHANT$ more severely with a higher-order anticipation effect.

\newpage
\section*{Web Appendix F: Figures and Tables} \label{sec::figures_tables}

\begin{figure}[htbp]
    \centering
    \includegraphics[width=\textwidth]{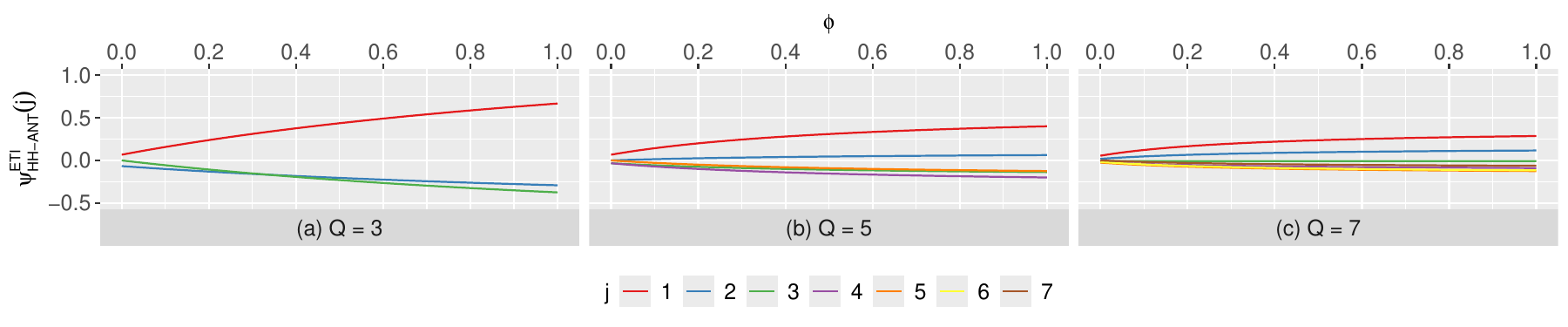}
    \caption{Weights $\psi_\HHANT^\ETI(j)$ for $j\in\{1,\ldots,J-1\}$ in $\bbE(\widehat\gamma^\HHANT)$ under the HH-ANT working model (panels (a)-(c)), assuming the true model is ETI. Here, $Q\in\{3, 5, 7\}$ and $\phi\in[0,1]$.}\label{figure_coeff_HH-ANT}
\end{figure}

\begin{figure}[htbp]
    \centering
    \includegraphics[width=\textwidth]{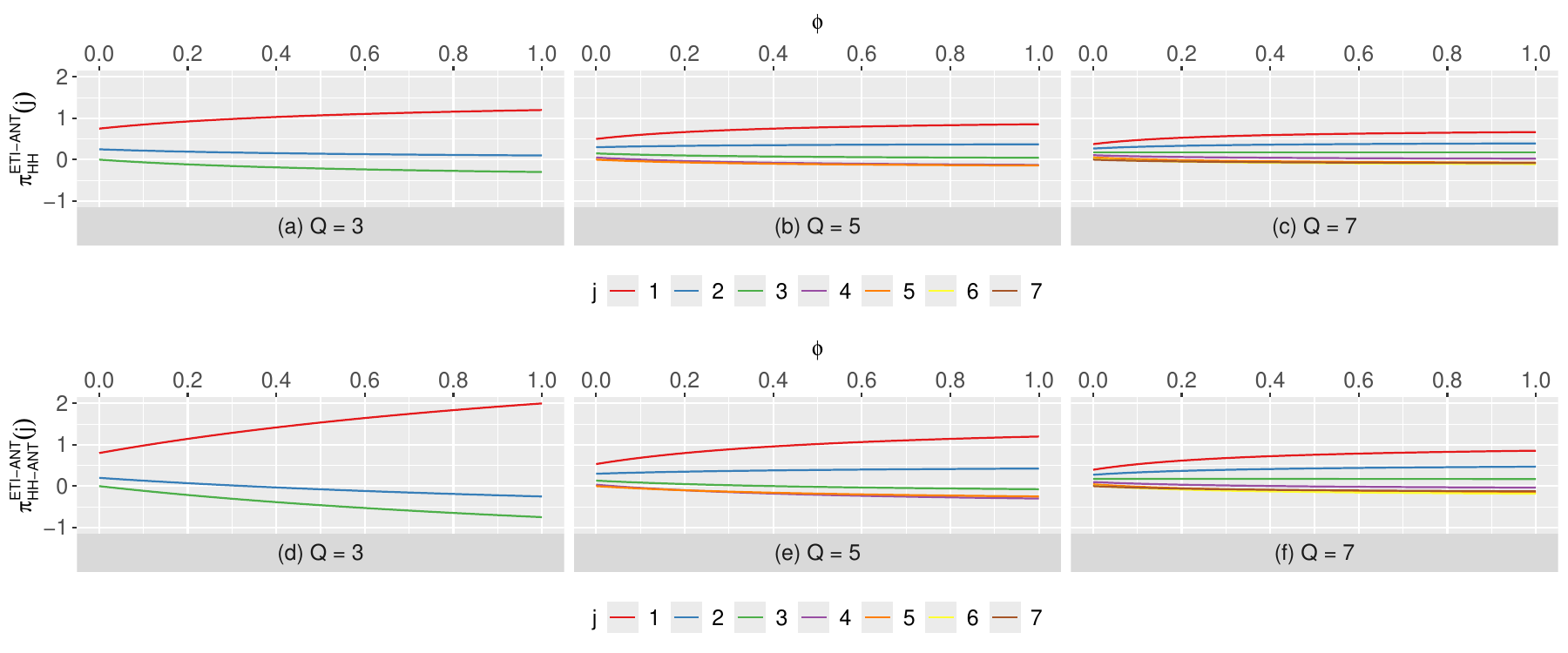}
    \caption{Weights $\pi_\HH^{\ETIANT}(j)$ for $s\in\{1,\ldots,J-1\}$) in $\bbE(\widehat\delta^\HH)$ under the HH working model (panels (a)-(c)), and $\pi_\HHANT^\ETIANT(j)$ for $s\in\{1,\ldots,J-1\}$) in $\bbE(\widehat\delta^\HHANT)$ under the HH-ANT working model (panels (d)-(f)), assuming the true model is ETI-ANT. Here, $Q\in\{3, 5, 7\}$ and $\phi\in[0,1]$.}\label{figure_coeff_HH_vs_HH-ANT}
\end{figure}

\begin{figure}[htbp]
    \centering
    \includegraphics[width=\textwidth]{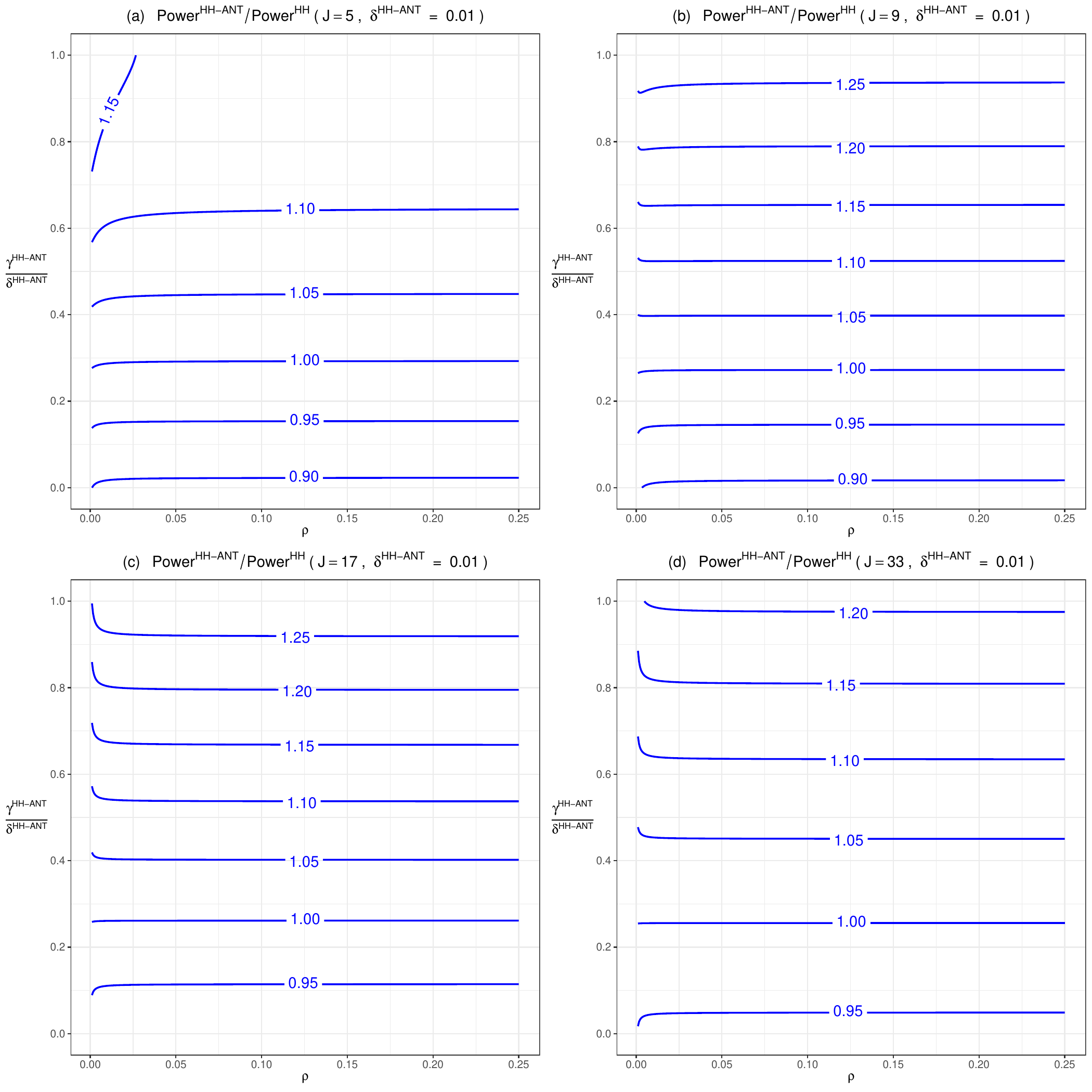}
    \caption{Contour plots of power ratio between $\widehat\delta^\HHANT$ from the HH-ANT model and $\widehat\delta^\HHANT$ from the HH model as functions of $\gamma^\HHANT/\delta^\HHANT$ and $\rho$ at a significant level of $\alpha = 0.05$ in a standard SW-CRT. Here, $I = 32$, $J \in \{5, 9, 17, 33\}$, $K = 100$, $\delta^\HHANT = 0.01$, $\gamma^\HHANT/\delta^\HHANT \in [0, 1]$, $\sigma^2 = 1$, and $\rho \in [0, 0.25]$. Blue lines are $\text{Power}^\HHANT/\text{Power}^\HH$ in all panels.}\label{figure_power_ratio_delta0.01}
\end{figure}

\begin{figure}[htbp]
    \centering
    \includegraphics[width=\textwidth]{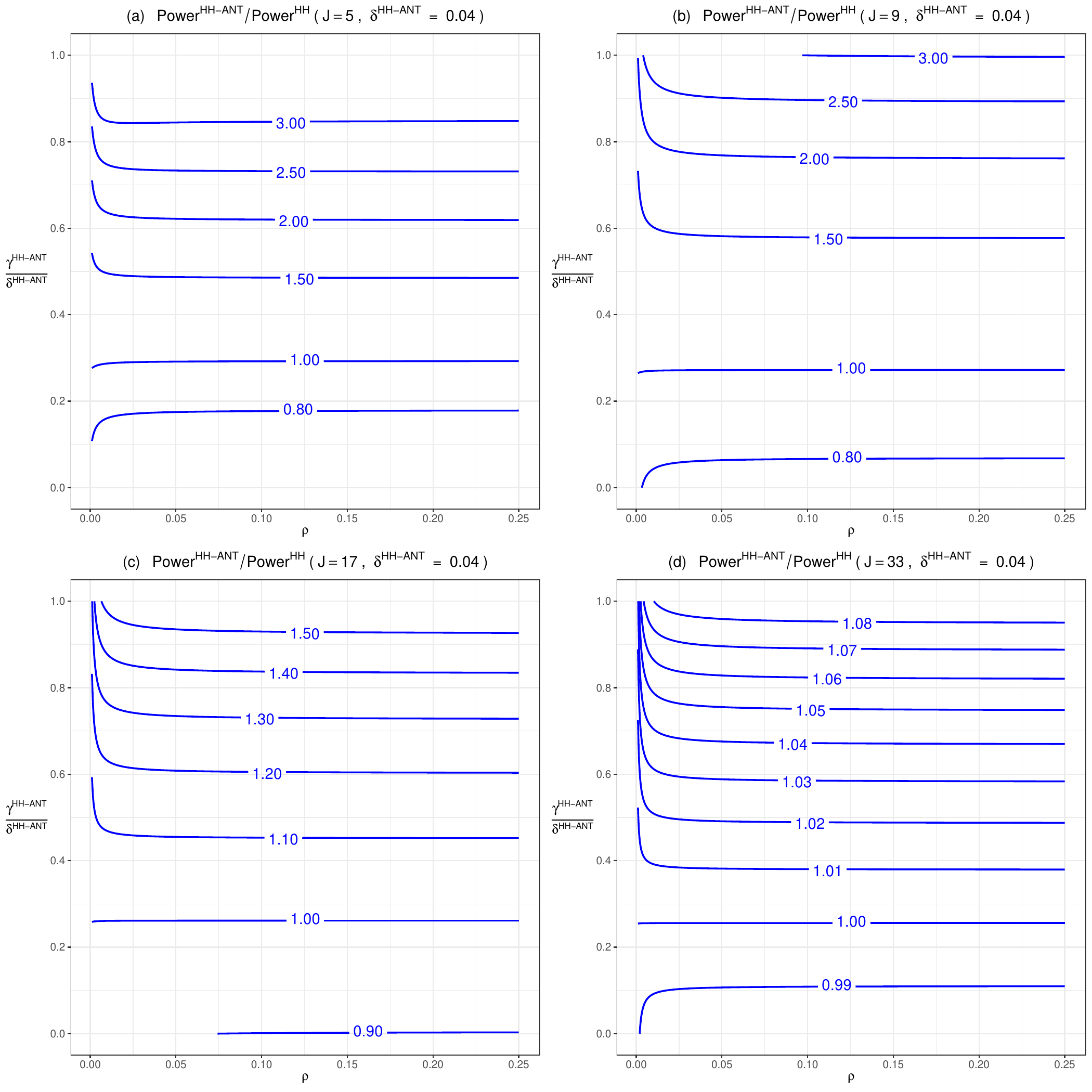}
    \caption{Contour plots of power ratio between $\widehat\delta^\HHANT$ from the HH-ANT model and $\widehat\delta^\HHANT$ from the HH model as functions of $\gamma^\HHANT/\delta^\HHANT$ and $\rho$ at a significant level of $\alpha = 0.05$ in a standard SW-CRT. Here, $I = 32$, $J \in \{5, 9, 17, 33\}$, $K = 100$, $\delta^\HHANT = 0.04$, $\gamma^\HHANT/\delta^\HHANT \in [0, 1]$, $\sigma^2 = 1$, and $\rho \in [0, 0.25]$. Blue lines are $\text{Power}^\HHANT/\text{Power}^\HH$ in all panels.}\label{figure_power_ratio_delta0.04}
\end{figure}

\begin{figure}[htbp]
    \centering
    \includegraphics[width=\textwidth]{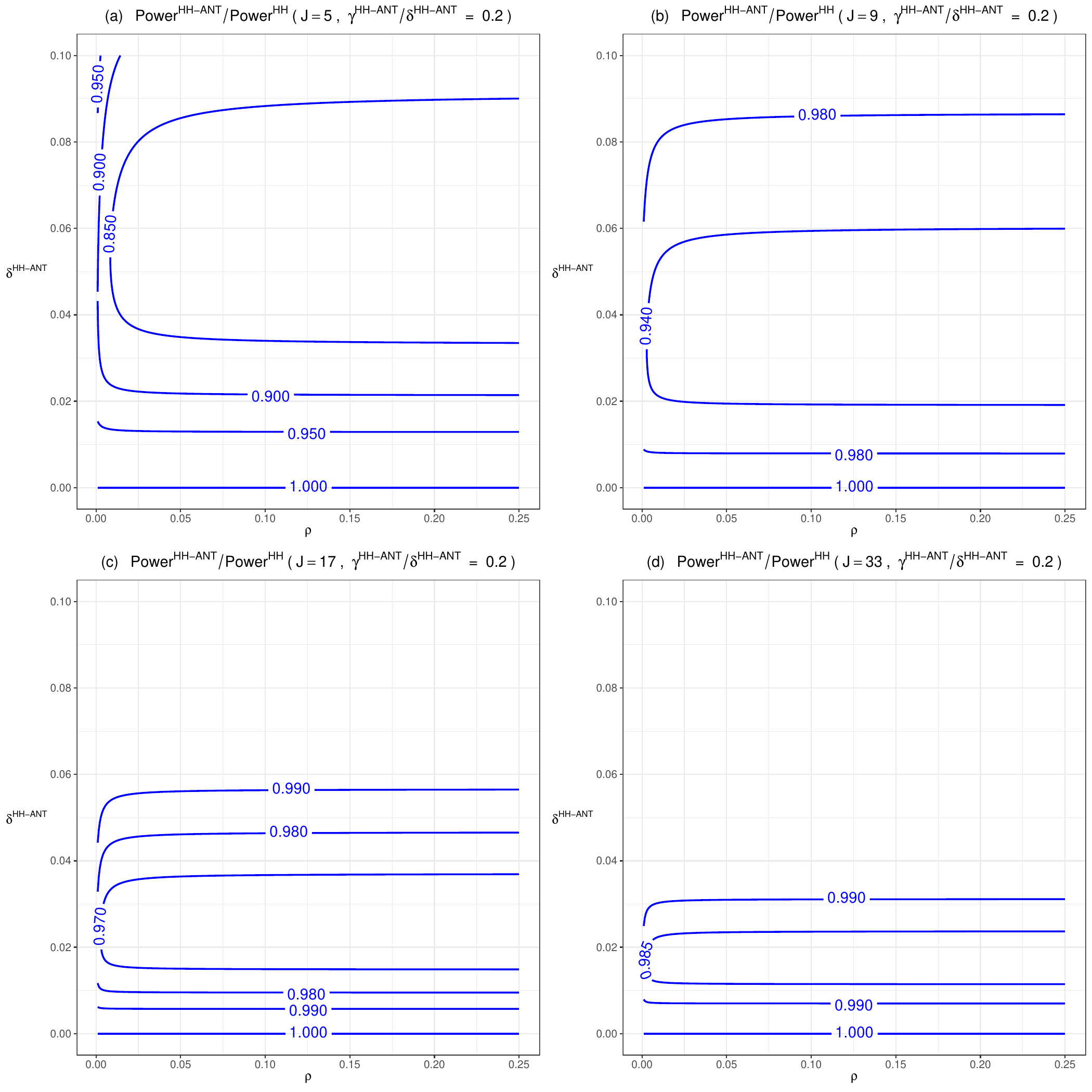}
    \caption{Contour plots of power ratio between $\widehat\delta^\HHANT$ from the HH-ANT model and $\widehat\delta^\HHANT$ from the HH model as functions of $\delta^\HHANT$ and $\rho$ at a significant level of $\alpha = 0.05$ in a standard SW-CRT. Here, $I = 32$, $J \in \{5, 9, 17, 33\}$, $K = 100$, $\gamma^\HHANT/\delta^\HHANT = 0.2$, $\delta^\HHANT \in [0, 0.1]$, $\sigma^2 = 1$, and $\rho \in [0, 0.25]$. Blue lines are $\text{Power}^\HHANT/\text{Power}^\HH$ in all panels.}\label{figure_power_ratio_0.2}
\end{figure}

\begin{figure}[htbp]
    \centering
    \includegraphics[width=\textwidth]{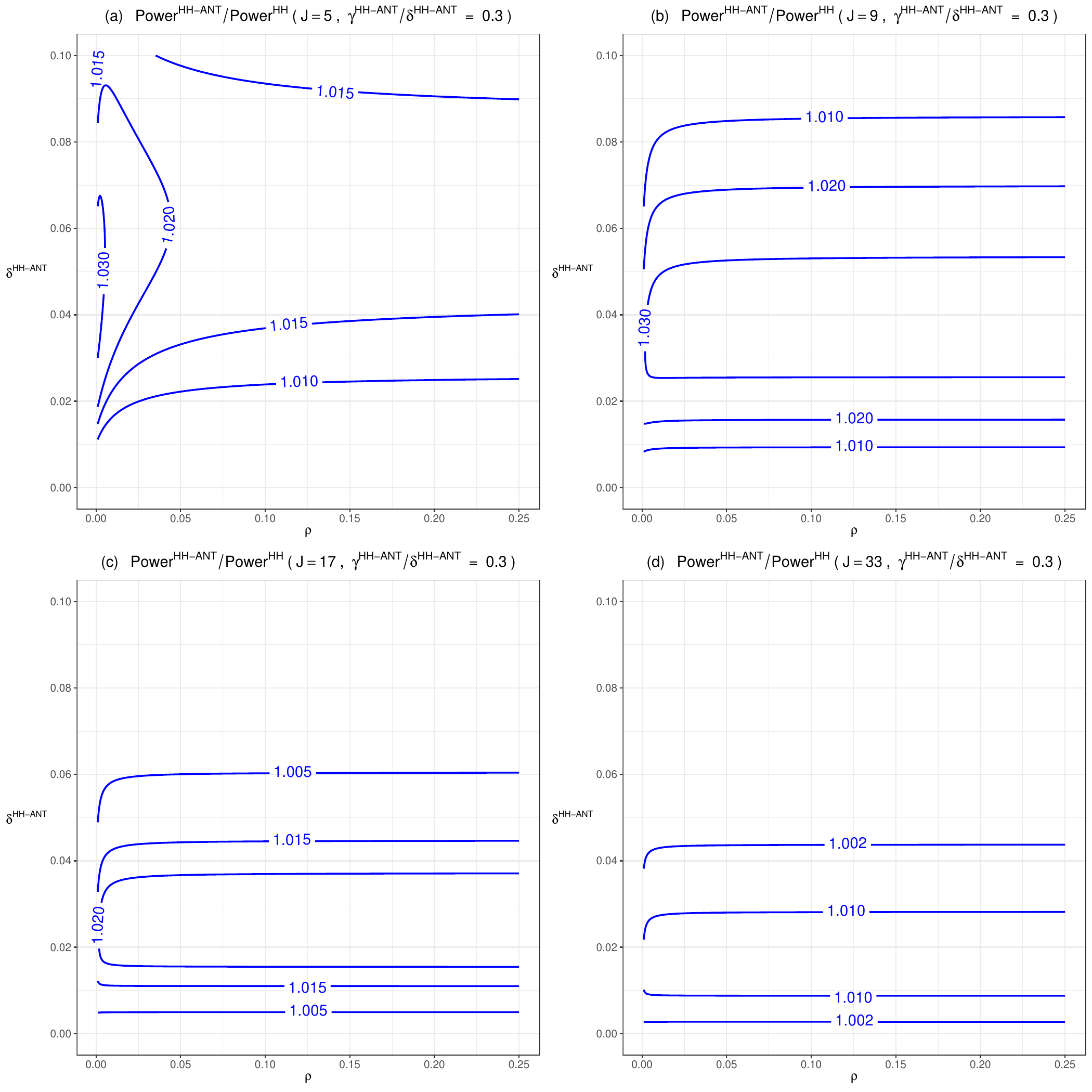}
    \caption{Contour plots of power ratio between $\widehat\delta^\HHANT$ from the HH-ANT model and $\widehat\delta^\HHANT$ from the HH model as functions of $\delta^\HHANT$ and $\rho$ at a significant level of $\alpha = 0.05$ in a standard SW-CRT. Here, $I = 32$, $J \in \{5, 9, 17, 33\}$, $K = 100$, $\gamma^\HHANT/\delta^\HHANT = 0.3$, $\delta^\HHANT \in [0, 0.1]$, $\sigma^2 = 1$, and $\rho \in [0, 0.25]$. Blue lines are $\text{Power}^\HHANT/\text{Power}^\HH$ in all panels.}\label{figure_power_ratio_0.3}
\end{figure}

\begin{figure}[htbp]
    \centering
    \includegraphics[width=\textwidth]{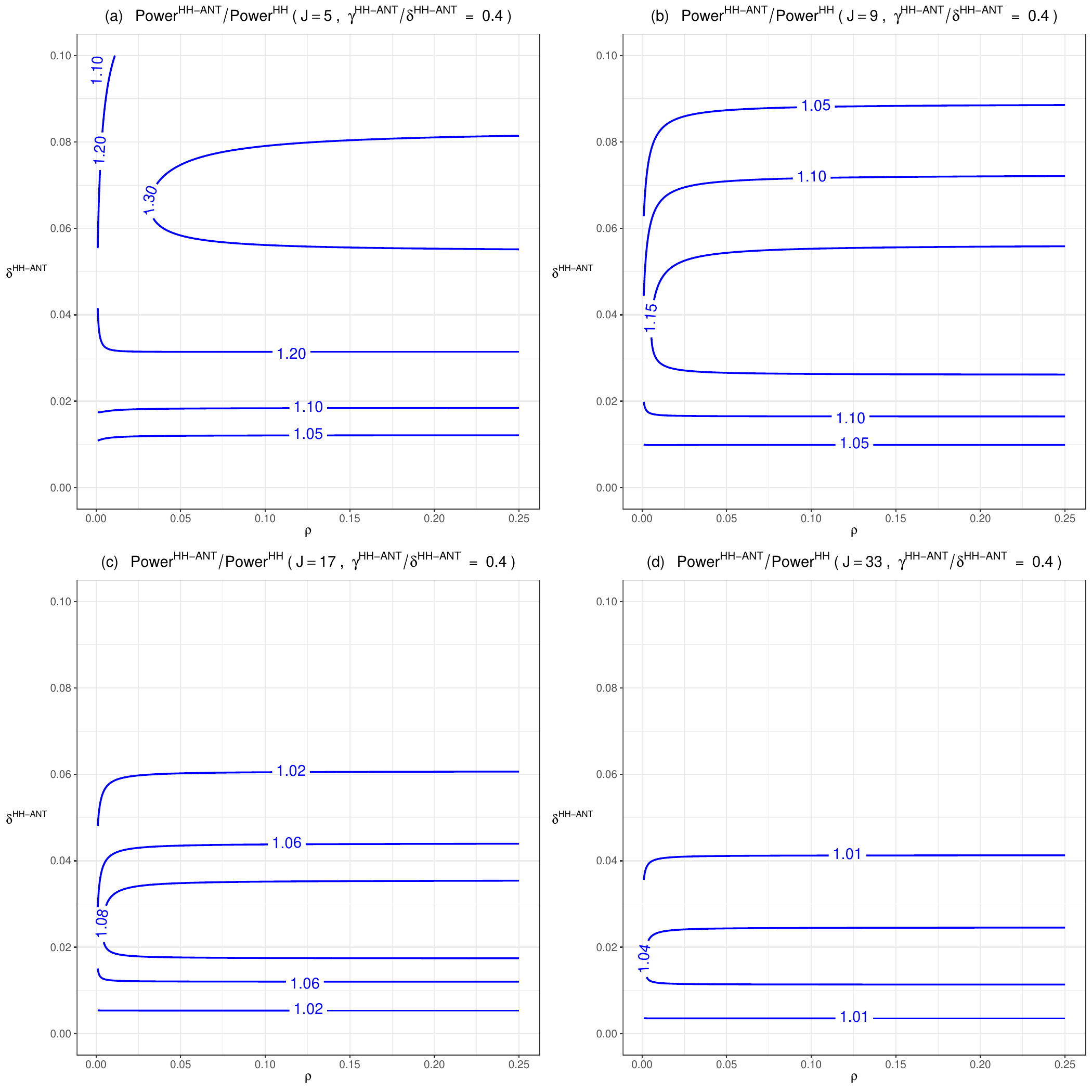}
    \caption{Contour plots of power ratio between $\widehat\delta^\HHANT$ from the HH-ANT model and $\widehat\delta^\HHANT$ from the HH model as functions of $\delta^\HHANT$ and $\rho$ at a significant level of $\alpha = 0.05$ in a standard SW-CRT. Here, $I = 32$, $J \in \{5, 9, 17, 33\}$, $K = 100$, $\gamma^\HHANT/\delta^\HHANT = 0.4$, $\delta^\HHANT \in [0, 0.1]$, $\sigma^2 = 1$, and $\rho \in [0, 0.25]$. Blue lines are $\text{Power}^\HHANT/\text{Power}^\HH$ in all panels.}\label{figure_power_ratio_0.4}
\end{figure}

\begin{figure}[htbp]
    \centering
    \includegraphics[width=\textwidth]{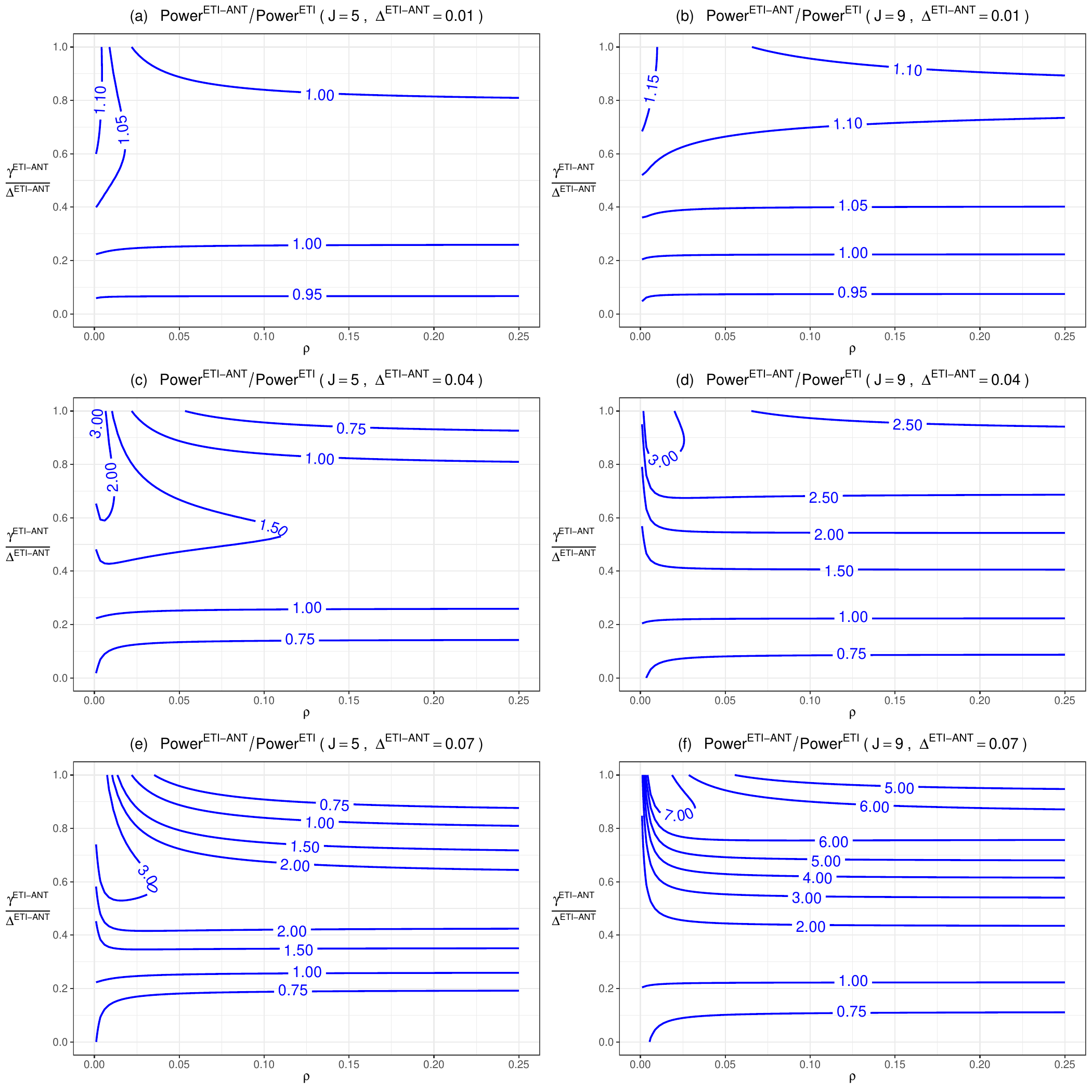}
    \caption{Contour plots of power ratio between $\widehat\Delta^\ETIANT$ from the ETI-ANT model and $\widehat\Delta^\ETI$ from the ETI model as functions of $\gamma^\ETIANT/\Delta^\ETIANT$ and $\rho$ at a significant level of $\alpha = 0.05$ in a standard SW-CRT. Here, $I = 32$, $J \in \{5, 9\}$, $K = 100$, $\gamma^\ETIANT/\Delta^\ETIANT \in [0, 1]$, $\Delta^\ETIANT \in \{0.01, 0.04, 0.07\}$, $\sigma^2 = 1$, and $\rho \in [0, 0.25]$. Blue lines are $\text{Power}^\ETIANT/\text{Power}^\ETI$ in all panels.}\label{figure_power_ratio_fixed_Delta}
\end{figure}

\begin{figure}[htbp]
    \centering
    \includegraphics[width=\textwidth]{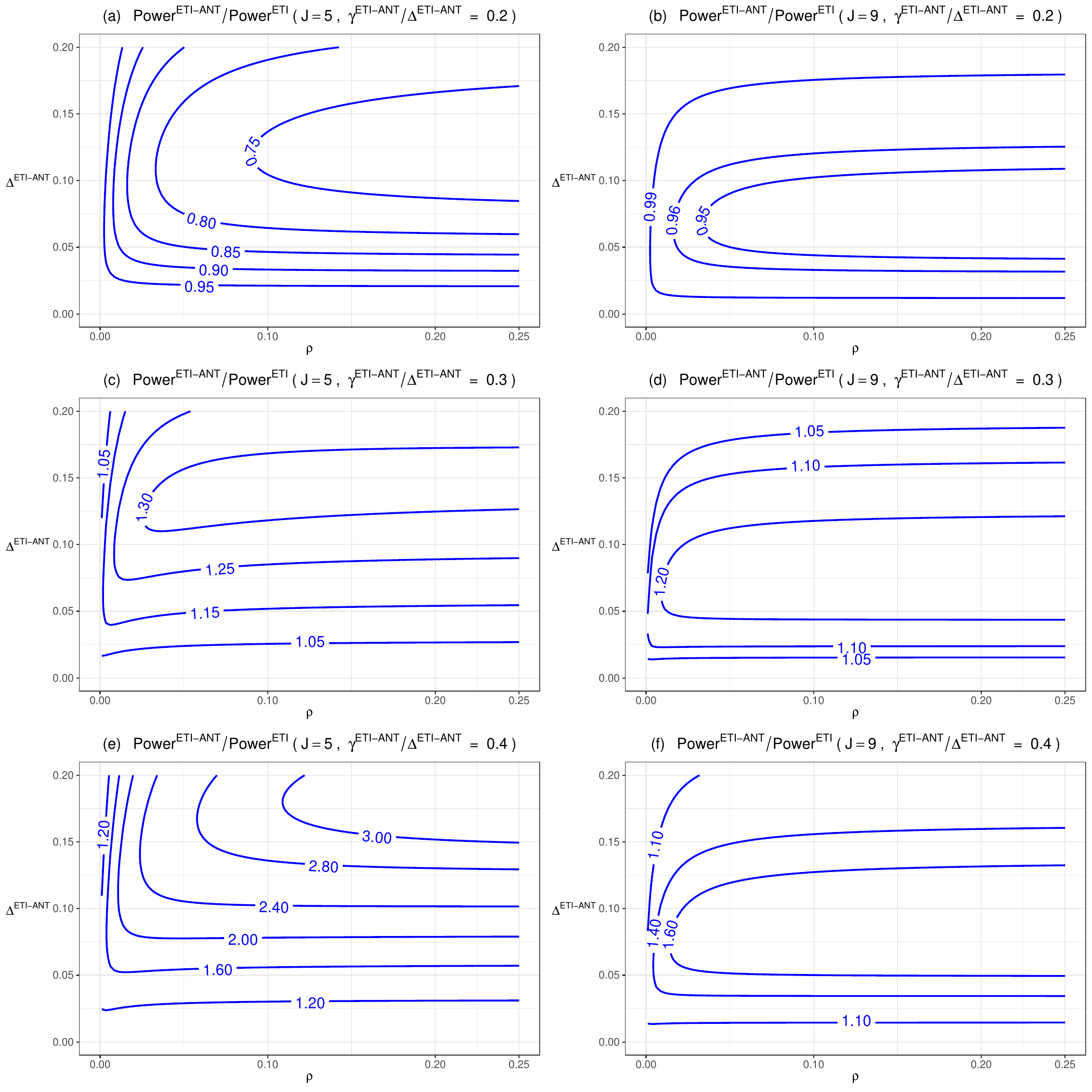}
    \caption{Contour plots of power ratio between $\widehat\Delta^\ETIANT$ from the ETI-ANT model and $\widehat\Delta^\ETI$ from the ETI model as functions of $\gamma^\ETIANT/\Delta^\ETIANT$ and $\rho$ at a significant level of $\alpha = 0.05$ in a standard SW-CRT. Here, $I = 32$, $J \in \{5, 9\}$, $K = 100$, $\gamma^\ETIANT/\Delta^\ETIANT \in \{0.2, 0.3, 0.4\}$, $\Delta^\ETIANT \in [0, 0.2]$, $\sigma^2 = 1$, and $\rho \in [0, 0.25]$. Blue lines are $\text{Power}^\ETIANT/\text{Power}^\ETI$ in all panels.}\label{figure_power_ratio_fixed_ratio}
\end{figure}

\begin{figure}[htbp]
    \centering
    \includegraphics[width=\textwidth]{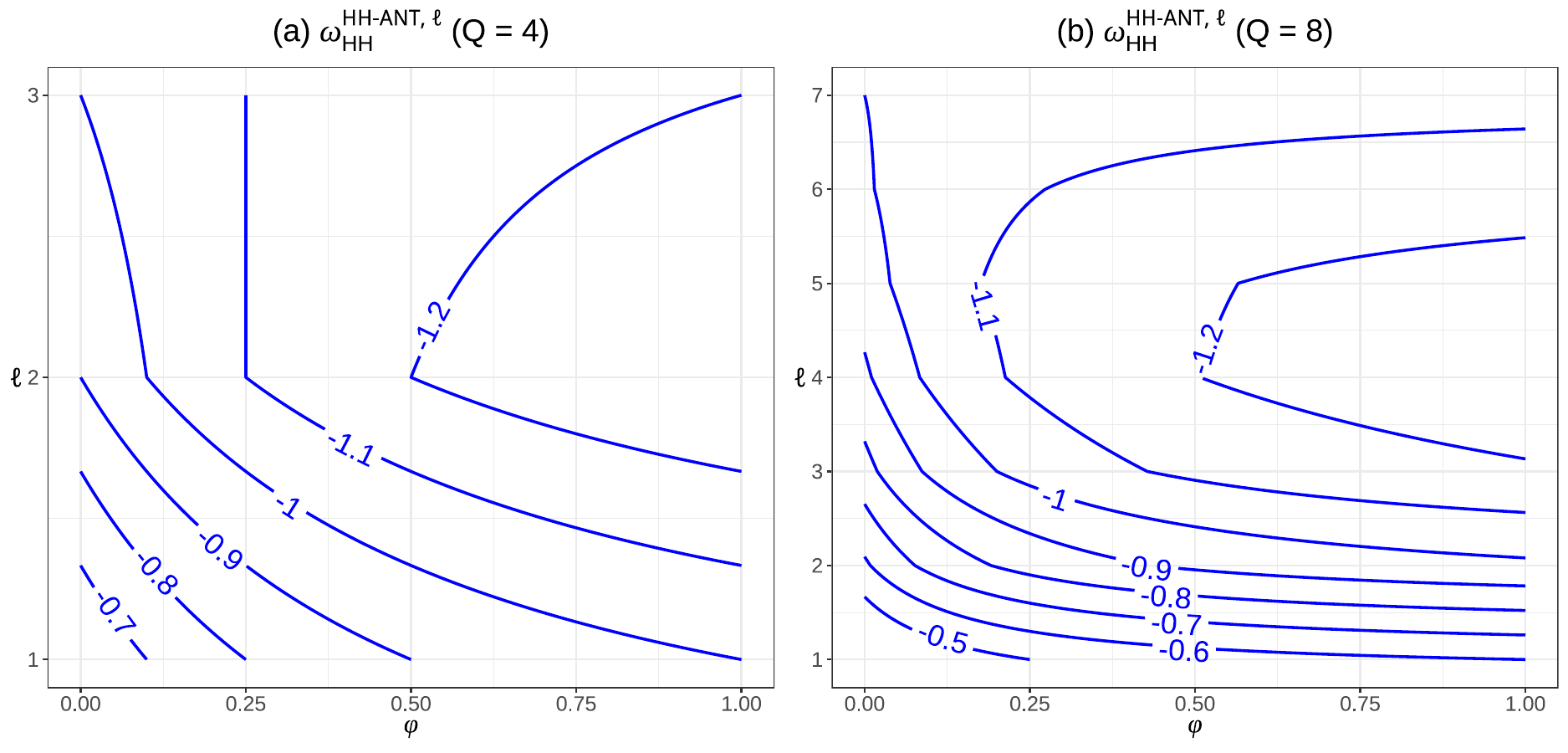}
    \caption{Contour plots of $\omega_\HH^{\HHANT,\ell}$ as functions of $\ell$ and $\rho$ under the HH working model, assuming the true model is HH-ANT with an $\ell$-th order anticipation effect. Here, $Q \in \{4, 8\}$, $\ell \in \{1, \ldots, Q-1\}$, and $\rho \in [0, 1]$. Blue lines are $\omega_\HH^{\HHANT,\ell}$ in all panels.}\label{figure_coeff_HH_higher_order}
\end{figure}

\begin{figure}[htbp]
    \centering
    \includegraphics[width=\linewidth]{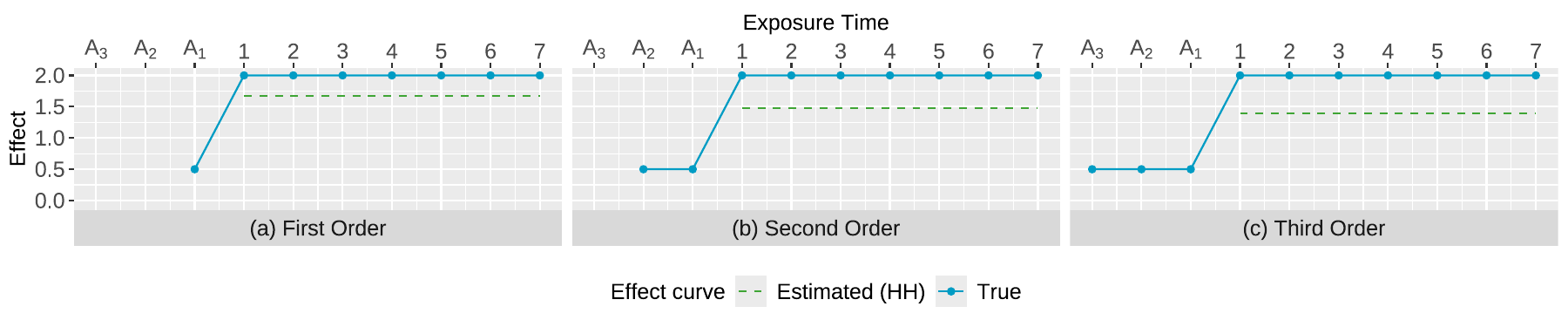}
    \caption{Constant (true) treatment effect curves with their estimated effect curves under the HH working model. We consider standard cross-sectional SW-CRTs with $I = 28$, $J = 8$, and $K = 50$, which gives $Q=7$ treatment sequences with $I_q=4$ clusters each and $\phi \approx 0.80$. The true data-generating processes are HH-ANT with $\ell$-order anticipation effects ($\ell \in \{1, 2, 3\}$, panel (a)-(c)), where $\mu = 14$, $\beta_j = 0.5\sin\{2\pi (j-1)/7\}$, $\tau = 0.141$, $\sigma = 0.5$, $\gamma^\HHANT = 0.5$, and $\delta^\HHANT = 2$. The constant anticipation effect occurs at exposure time $A_t$, where $t$ is the number of periods before the treatment adoption time. The point treatment effects occur when exposure times range from 1 to 7. The blue line is the true effect curve and the green dashed line is the estimated effect curve using an HH model via 2,000 simulated datasets.} \label{figure_HH_higher_order}
\end{figure}

\begin{figure}[htbp]
    \centering
    \includegraphics[width=\textwidth]{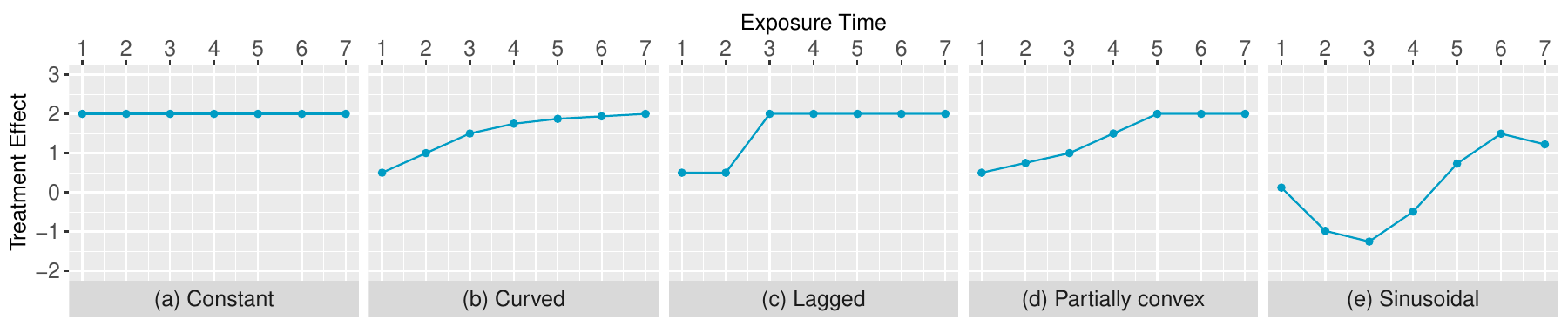}
    \caption{Five functional forms of exposure-time treatment effect heterogeneity. The panels provide (a) constant, (b) curved, (c) lagged, (d) partially convex, and (e) sinusoidal patterns. For the constant treatment effect, $\delta(s) = 2$ for $s\in\{1,\ldots,7\}$. For the curved treatment effect, $\delta(1) = 0.5$, $\delta(s) = 2-0.5^{s-2}$ for $s\in\{2,\ldots,6\}$, and $\delta(7) = 2$. For the lagged treatment effect, $\delta(s) = 0.5$ for $s\in\{1,2\}$ and $\delta(s) = 2$ for $s\in\{3,\ldots,7\}$. For the partially convex treatment effect, $\delta(s) = 0.25+0.25s$ for $s\in\{1,2,3\}$, $\delta(4) = 1.5$, and $\delta(s) = 2$ for $s\in\{5,6,7\}$. For the sinusoidal treatment effect used in Simulation Studies III and IV, $\delta(s) = -1.41\sin{2\pi(s-1)/7}+0.12$ for $s\in\{1,\ldots,7\}$. The blue line shows the treatment effect at each exposure time.}\label{figure_functional_forms}
\end{figure}

\clearpage
\begin{table}[htbp]
    \centering
    \caption{Detectable average treatment effect size with varying outcome ICCs. The working model is ETI-ANT with $80\%$ power at a significance level of $0.05$. The hypothetical study is a standard SW-CRT with $I = 18$, $J = 7$, $K = 50$, $\ell = 1$, and $\sigma^2 = 1$.}
    \label{tab::power_ETI-ANT}
    \begin{threeparttable}
    \begin{tabular}{cc}
    \toprule
    Outcome ICC ($\rho$) & Average treatment effect ($\Delta^\ETIANT$)\\
    \hline
    0.00 & 0.124 \\
    0.01 & 0.212 \\
    0.05 & 0.281 \\
    0.10 & 0.299 \\
    0.20 & 0.310 \\
    \bottomrule
    \end{tabular}\smallskip
    \end{threeparttable}
\end{table}

\begin{table}[htbp]
    \centering
    \caption{An overview of simulation studies, research questions, and settings (Simulation Studies I and II).}\label{tab::reseach_question_1}
    \begin{tabular}{|>{\raggedright\arraybackslash}p{0.05\linewidth}|>{\raggedright\arraybackslash}p{0.30\linewidth}|>{\raggedright\arraybackslash}p{0.55\linewidth}|}
    \hline
    \textbf{Study} & \textbf{Simulation settings} & \textbf{Research questions to address} \\
    \hline
    \textbf{I} &
    \begin{itemize}[leftmargin=*,labelsep=0.5em,itemsep=0pt,parsep=0pt,topsep=0pt]
        \item A simple exchangeable correlation structure.
        \item Constant treatment effect $\delta^\HH = 0, 0.06$.
    \end{itemize}
    & 
    The true model is HH in \eqref{model::HH}.
    \begin{itemize}[leftmargin=*,labelsep=0.5em,itemsep=0pt,parsep=0pt,topsep=0pt]
        \item Are the treatment effect estimators from the HH model and the HH-ANT model consistent?
        \item Is the treatment effect estimator from the HH model more efficient than that of the HH-ANT model?
        \item What are the type I error rates of the treatment effect estimators from the HH model and the HH-ANT model?
        \item Does the test for the treatment effect based on the HH model have more power than that based on the HH-ANT model?
        \item Are the cluster random effect variance estimators, i.e., $\alpha_i\sim\calN(0,\tau^2)$, from the HH model and the HH-ANT model consistent?
    \end{itemize}
    \\
    \hline
    \textbf{II} & 
    \begin{itemize}[leftmargin=*,labelsep=0.5em,itemsep=0pt,parsep=0pt,topsep=0pt]
        \item A simple exchangeable correlation structure.
        \item Constant treatment effect $\delta^\HHANT = 0.075$ and anticipation effect $\gamma^\HHANT = 0.04$.
    \end{itemize}
    & 
    The true model is HH-ANT in \eqref{model::HH-ANT}.
    \begin{itemize}[leftmargin=*,labelsep=0.5em,itemsep=0pt,parsep=0pt,topsep=0pt]
        \item Are the cluster random effect variance estimators from the four models consistent?
    \end{itemize}
    Among working models that do not account for the anticipation effect (HH, ETI):
    \begin{itemize}[leftmargin=*,labelsep=0.5em,itemsep=0pt,parsep=0pt,topsep=0pt]
        \item Which treatment effect estimators are the most biased?
    \end{itemize}
    Among working models that account for the anticipation effect (HH-ANT, ETI-ANT):
    \begin{itemize}[leftmargin=*,labelsep=0.5em,itemsep=0pt,parsep=0pt,topsep=0pt]
        \item Which treatment effect estimators are most biased?
        \item Are the anticipation effect estimators of the three models consistent?
        \item Which treatment effect estimators are more efficient?
        \item Which anticipation effect estimators are more efficient?
        \item What are the type I error rates of tests for the treatment effect?
        \item What are the type I error rates of tests for the anticipation effect?
        \item Which test for the treatment effect has more power?
        \item Which test for the anticipation effect has more power?
    \end{itemize}
    \\
    \hline
    \end{tabular}
\end{table}

\begin{table}[htbp]
    \centering
    \caption{An overview of simulation studies, research questions, and settings (Simulation Studies III and IV).} \label{tab::reseach_question_2}
    \begin{tabular}{|>{\raggedright\arraybackslash}p{0.05\linewidth}|>{\raggedright\arraybackslash}p{0.30\linewidth}|>{\raggedright\arraybackslash}p{0.55\linewidth}|}
    \hline
    \textbf{Study} & \textbf{Simulation settings} & \textbf{Research questions to address} \\
    \hline
    \textbf{III} & 
    \begin{itemize}[leftmargin=*,labelsep=0.5em,itemsep=0pt,parsep=0pt,topsep=0pt]
        \item A simple exchangeable correlation structure.
        \item Time-varying point treatment effects $\delta^\ETI(s) = -1.41\sin\{2\pi(s-1)/7\}+0.12$ for $s\in\{1,\ldots,8\}$ with the TATE $\Delta^\ETI = \sum_{s=1}^8\delta^\ETI(s)/8 = 0.12$.
    \end{itemize}
    & 
    The true model is ETI in \eqref{model::ETI}.
    \begin{itemize}[leftmargin=*,labelsep=0.5em,itemsep=0pt,parsep=0pt,topsep=0pt]
        \item Are the cluster random effect variance estimators of the four models consistent?
    \end{itemize}
    Among working models that do not account for exposure-time heterogeneity (HH, HH-ANT):
    \begin{itemize}[leftmargin=*,labelsep=0.5em,itemsep=0pt,parsep=0pt,topsep=0pt]
        \item Is the treatment effect estimator from the HH-ANT model more biased than that from the HH model?
        \item Is the anticipation effect estimator from the HH-ANT model biased?
    \end{itemize}
    Among working models that account for exposure-time heterogeneity (ETI, ETI-ANT):
    \begin{itemize}[leftmargin=*,labelsep=0.5em,itemsep=0pt,parsep=0pt,topsep=0pt]
        \item Which treatment effect estimators are more efficient?
        \item Which anticipation effect estimators are more efficient?
        \item What are the type I error rates of tests for the treatment effect?
        \item What are the type I error rates of tests for the anticipation effect?
        \item Which test for the treatment effect has more power?
        \item Which test for the anticipation effect has more power?
    \end{itemize}
    \\
    \hline
    \textbf{IV} & 
    \begin{itemize}[leftmargin=*,labelsep=0.5em,itemsep=0pt,parsep=0pt,topsep=0pt]
        \item A simple exchangeable correlation structure.
        \item Time-varying point treatment effects $\delta^\ETIANT(s) = -1.41\sin\{2\pi(s-1)/7\}+0.12$ for $s\in\{1,\ldots,8\}$ with the TATE $\Delta^\ETIANT = \sum_{s=1}^8\delta^\ETIANT(s)/8 = 0.12$. The anticipation effect $\gamma^\ETIANT = 0.04$.
    \end{itemize}
    & 
    The true model is ETI-ANT in \eqref{model::ETI-ANT}. 
    \begin{itemize}[leftmargin=*,labelsep=0.5em,itemsep=0pt,parsep=0pt,topsep=0pt]
        \item Are the cluster random effect variance estimators of the four models consistent?
    \end{itemize}
    Among working models that do not consider exposure-time heterogeneity (HH, HH-ANT),
    \begin{itemize}[leftmargin=*,labelsep=0.5em,itemsep=0pt,parsep=0pt,topsep=0pt]
        \item Is the treatment effect estimator from the HH-ANT model more biased than that from the HH model?
        \item Is the anticipation effect estimator from the HH-ANT model biased?
    \end{itemize}
    For the working model that accounts for exposure-time heterogeneity but not the anticipation effect (ETI):
    \begin{itemize}[leftmargin=*,labelsep=0.5em,itemsep=0pt,parsep=0pt,topsep=0pt]
        \item Is the treatment effect estimator from the ETI model biased for the true treatment effect?
    \end{itemize}
    For the working model that accounts for exposure-time heterogeneity and the anticipation effect (ETI-ANT):
    \begin{itemize}[leftmargin=*,labelsep=0.5em,itemsep=0pt,parsep=0pt,topsep=0pt]
        \item What are the type I error rates of the test for the treatment effect?
        \item What are the type I error rates of the test for the anticipation effect?
        \item What is the power of testing the treatment effect?
        \item Which is the power of testing the anticipation effect?
    \end{itemize}
    \\
    \hline
    \end{tabular}
\end{table}

\clearpage
\bibliography{SWDanticipation}